\documentclass[12pt]{article}

\newif\ifShowsupplement
\Showsupplementtrue

\RequirePackage{amsthm,amsmath,amsfonts,amssymb}
\usepackage{enumerate}
\usepackage[margin=2cm]{geometry}
\usepackage{bm}
\usepackage{anyfontsize}
\usepackage{appendix}
\usepackage[english]{babel}
\RequirePackage[colorlinks,citecolor=blue,urlcolor=blue]{hyperref}
\usepackage{mathtools}
\usepackage{xspace}
\usepackage{enumitem}
\RequirePackage{luatex85}
\usepackage[all]{xy}
\usepackage[dvipsnames]{xcolor}
\usepackage{tikz}
\usepackage{tikz-cd}
\usetikzlibrary{calc}
\usepackage{comment}
\usepackage{graphicx}
\usepackage{booktabs}

\usepackage{pdfpages}
\RequirePackage[authoryear]{natbib}

\numberwithin{equation}{section}
\allowdisplaybreaks

\theoremstyle{plain}
\newtheorem{Theorem}{Theorem}[section]
\newtheorem{Lemma}[Theorem]{Lemma}
\newtheorem{Corollary}[Theorem]{Corollary}
\newtheorem{Proposition}[Theorem]{Proposition}
\theoremstyle{definition}
\newtheorem{Assumption}[Theorem]{Assumption}

\newtheorem{Definition}[Theorem]{Definition}
\newtheorem{Remark}[Theorem]{Remark}

\newcommand\thesupplement{Supplementary Materials}

\newcommand\dsim{\sim_{\text{d}}} 

\DeclareMathOperator{\theLaw}{ {\mathsf{P} }}
\DeclareMathOperator{\Expected}{ {\mathsf{E} }}

\DeclareMathOperator{\Exp}{\mathrm{Exp}}
\DeclareMathOperator{\Log}{\mathrm{Log}}
\DeclareMathOperator{\vspan}{span}
\DeclareMathOperator{\image}{Im}
\DeclareMathOperator{\trace}{trace}
\DeclareMathOperator{\support}{support}
\DeclareMathOperator\PT{{PT}}

\newcommand\indicator{{1}}
\newcommand\esssup{\mathrm{ess\,sup}}
\newcommand{\Real}{\mathbb{R}}

\newcommand{\vep}{\varepsilon}

\newcommand\ud{\mathrm{d}}

\newcommand\Mprocess{\mathbb{M}}
\newcommand\Pprocess{\mathbb{P}}

\newcommand\Cregion{\mathfrak{C}}
\newcommand\Dregion{\mathfrak{D}}

\newcommand\Uset{\mathcal{U}}
\newcommand\Vset{\mathcal{V}}
\newcommand\covering{\mathcal{N}}
\newcommand\bracketing{\mathcal{N}_{[\,]}}

\newcommand\risk{\mathfrak{R}}

\newcommand\Fset{\mathcal{F}}
\newcommand\Gset{\mathcal{G}}
\newcommand\KIntegralOperator{\mathfrak{k}}
\newcommand\embeddingCov{\mathfrak{t}}
\newcommand\origin{\mathfrak{o}}
\newcommand\Xspace{\mathfrak{X}}
\newcommand\Zspace{\mathfrak{Z}}

\newcommand\Yspace{\mathfrak{Y}}

\newcommand\Hess{\mathrm{Hess}}
\newcommand\Hspace{\mathcal{H}}

\newcommand\Mspace{\mathcal{M}}
\newcommand\continuous{\mathcal{C}}
\newcommand{\sphere}{ \mathbb{S}}

\newcommand{\sphereDist}{ \mathrm{d}_\sphere}
\newcommand{\innerProd}[1]{\langle #1 \rangle}
\newcommand{\YinnerProd}[1]{\langle #1 \rangle_\Yspace}
\newcommand{\HinnerProd}[1]{\langle #1 \rangle_\Hspace}
\newcommand{\HkinnerProd}[1]{\langle #1 \rangle_{\Hspace_k}}
\newcommand{\norm}[1]{\lVert #1 \rVert}
\newcommand{\Ynorm}[1]{\left\lVert #1 \right\rVert_{\Yspace}}

\newcommand{\Hnorm}[1]{\lVert #1 \rVert_{\Hspace}}
\newcommand{\Hknorm}[1]{\lVert #1 \rVert_{\Hspace_k}}
\newcommand{\Ltwonorm}[1]{\lVert #1 \rVert_{L^2}}
\newcommand{\LtwoX}{{L^2(\theLaw_X)}}
\newcommand{\LtwoXTS}{{L^2(\Xspace, T_\origin \sphere; \theLaw_X)}}
\newcommand{\inftynorm}[1]{\lVert #1 \rVert_{\infty}}
\newcommand{\operatornorm}[1]{\lVert #1 \rVert_{\mathrm{op}}}
\newcommand\proj{{\mathcal{P}}}
\newcommand\ball{\mathcal{B}}
\newcommand\tensor{\otimes}
\newcommand{\bxi}{{\boldsymbol{\xi}}}

\newcommand\ASconv{\xrightarrow{\mathrm{a.s.}}}
\newcommand\outerASconv{\xrightarrow{\mathrm{a.s.}*}}
\newcommand{\sinc}{\mathrm{sinc}}
\newcommand{\adjoint}{\dagger}
\newcommand{\transpose}{\mathsf{T}}
\newcommand{\eval}{\mathcal{E}}
\newcommand{\orthogonal}{\mathcal{O}}

\newcommand\bounded{\mathcal{L}}
\newcommand\id{\mathrm{Id}}

\DeclareMathOperator*{\hy}{\mathbb{Y}}

\newcommand{\nv}{\| v \|}
\newcommand{\vy}{\langle v, y \rangle}
\newcommand{\py}{\langle \origin, y \rangle}
\DeclareMathOperator*{\Ccal}{\mathcal{C}}
\DeclareMathOperator*{\Ucal}{\mathcal{U}}

\newcommand{\n}{\omega} 
\newcommand{\yo}{\rho} 
\newcommand{\vangle}{\theta} 
\newcommand{\dom}{\mathfrak{A}_{\pi/4}}
\newcommand{\domm}{\tilde{\mathfrak{A}}_{\pi/4}}
\newcommand{\h}{\operatorname{sinc}}
\newcommand{\hh}{h}

\newcommand\frechetModel{\textsl{Fréchet regression}}
\newcommand\bayesModel{\textsl{Bayes Hilbert regression}}
\newcommand\sphericalModel{\textsl{Spherical regression}}
\newcommand\logNNModel{\textsl{lognormal-normal model}}

\begin{document}

\title{Infinite-Dimensional Spherical Kernel ridge Regression}
\author{%
  Beatrice Matteo\thanks{GSI and RISIS, GSEM, Universit\'e de Gen\`eve, Switzerland. \texttt{beatrice.matteo@unige.ch}}
  \and
  Almond St\"ocker\thanks{EPFL, Switzerland. \texttt{almond.stoecker@epfl.ch}. Supported by SNSF Grant 200020\_207367}
  \and
  Shahin Tavakoli\thanks{RISIS, GSEM, Universit\'e de Gen\`eve, Switzerland. \texttt{shahin.tavakoli@unige.ch}}
}
\date{}
\maketitle

\begin{abstract}
We introduce an intrinsic nonlinear regression framework for responses on a sphere of finite or infinite dimension, viewed as a Riemannian manifold. 
Unlike traditional tangent-space regressions, which lift responses to a common tangent space and thereby violate intrinsic spherical distances, our proposed method employs an intrinsic approach. 
We model the intrinsic conditional mean, fitting the responses with respect to the geodesic distance. 
We use a linear function space as a model space, which we compose with the Riemannian exponential map to obtain predictions on the sphere. 
Our approach generalises geodesic regression and parametric intrinsic regression models.
Leveraging vector-valued reproducing kernel Hilbert space theory, our approach reduces the infinite-dimensional estimation problem to a finite-dimensional problem via a novel representer theorem, leading to an efficient BFGS-based estimation algorithm.
We establish convergence rates and analyse the finite-sample behaviour of our estimator, concluding with an application to density regression.
The full implementation is available in R.
\end{abstract}



\section{Introduction}

This paper is concerned with intrinsic regression for responses on a sphere $\sphere$ of arbitrary, possibly infinite, dimension, viewed as a Riemannian manifold, with covariates in a general space.
Such responses include directional observations, pre-shapes, square-root densities, and functional data constrained to have unit norm.

Our target is the intrinsic conditional mean, the point of $\sphere$ minimising the expected squared geodesic distance to the response given the covariate.
Intrinsic means of random points on finite-dimensional Riemannian manifolds have been extensively studied \citep{karcherRiemannianCenterMass1977a,bhattacharyaLargeSampleTheory2003,bhattacharyaLargeSampleTheory2005,pennecIntrinsicStatisticsRiemannian2006,afsari2011riemannian,eltznerSmearyCentralLimit2019}, and arise in many applications \citep[e.g.,][]{drydenNonEuclideanStatisticsCovariance2009}. 
Several regression approaches build on them. 
Tangent-space regression maps the responses to a single tangent space with the log map and fits them with a linear-space loss \citep[Section~13.4.1]{dryden2016statistical}: its target depends on the base point and, because the sphere is curved, generally differs from the intrinsic conditional mean. 
Rather than performing tangent-space regression at a single base point, \citet{jupp1987fitting,kim2021smoothing} ``unroll'' the manifold along a path.
Intrinsic parametric models instead retain the geodesic distance.
Geodesic regression represents the conditional mean as $\Exp_p(xv)$ for a scalar predictor \citep{thomasfletcherGeodesicRegressionTheory2013}, and multivariate extensions replace $xv$ by a linear combination of tangent vectors \citep{zhu2009intrinsic,kimMultivariateGeneralLinear2014} or by a parametric nonlinear mapping of the covariates \citep{corneaRegressionModelsRiemannian2017}.
Nonparametric approaches include intrinsic kernel smoothing \citep{davis2010population}, extrinsic local regression \citep{linExtrinsicLocalRegression2017}, and penalised intrinsic approaches \citep{steinkeNonparametricRegressionManifolds2008}.
For low-dimensional spheres, regression with directional responses also has a long history in directional statistics \citep{mardiajupp2009directional,dimarzioNonparametricRegressionSpherical2014,paineSphericalRegressionModels2020}.

We extend intrinsic Riemannian regression nonparametrically.
We model the intrinsic conditional mean as $\mu(x) = \Exp_\origin(f(x))$, where $\origin \in \sphere$ is an intercept and $f$ is a $T_\origin\sphere$-valued function in a vector-valued reproducing kernel Hilbert space (VVRKHS), and we estimate $f$ by penalised least squares with respect to the geodesic distance.
In contrast to tangent-space regression, the tangent space $T_\origin \sphere$ serves only to parametrise $\mu$. The target and the fitting criterion remain intrinsic (Remark~\ref{remark:tangent_space_regression}).

The unit sphere of an infinite-dimensional Hilbert space is a Riemannian Hilbert manifold \citep{klingenbergRiemannianGeometry2011} that arises naturally in applications, such as density regression under the Fisher--Rao metric \citep{srivastava2007densities} (Section~\ref{sec:application}), elastic shape analysis \citep{srivastava2016functional,stockerFunctionalAdditiveModels2023,matteo2026elastickernelridgeregression}, and unit-norm functional data analysis \citep{ramsayFunctionalDataAnalysis2005}.
Infinite dimension raises difficulties beyond those of finite-dimensional manifold regression: intrinsic means need not exist \citep{tavakoli2026existencenonexistencecentresmass}, the optimisation is infinite-dimensional and nonconvex, and empirical process arguments that require finite entropy of the response space do not apply.
We address each of these difficulties.

Since the geodesic distance makes $\sphere$ a metric space, the intrinsic conditional mean is a conditional version of the Fr\'echet mean \citep{frechetElementsAleatoiresNature1948}. 
An approach to estimating conditional Fréchet means of random elements in a metric space using the metric alone was proposed by \citet{petersenFrechetRegressionRandom2019} and termed `Fréchet regression'. 
This approach conceptually differs from our intrinsic regression model, as it is designed to be agnostic to any (Riemannian) structure beyond the metric of the considered space. 
Recent extensions of it include \citet{linTotalVariationRegularized2021,chenUniformConvergenceLocal2022,ghosalFrechetSingleIndex2023,bhattacharjeeSingleIndexFrechet2023,qiuRandomForestWeighted2024}.
In particular, \citet{bhattacharjeeNonlinearGlobalFrechet2025} use an RKHS on the predictor space to construct a nonlinear global Fr\'echet regression---Section~\ref{sec:regression_problem} explains how it differs from ours.
Using the Riemannian structure of the sphere, rather than its metric alone, allows us to obtain a representer theorem, gradient-based algorithms, and explicit geometric assumptions for the theory.

\subsection*{Contributions}

We propose an intrinsic method for regression with outputs on a sphere of arbitrary dimension. While we use the sphere's geodesic distance to define our estimand (the intrinsic conditional mean), we parametrise it on a tangent space. Using tools from vector-valued reproducing kernel Hilbert spaces \citep[VVRKHS;][Chapter~6]{paulsenIntroductionTheoryReproducing2016}, our model is very general and accommodates many existing models as special cases. Our main contributions are:
\begin{enumerate}
    \item a novel representer result (Theorem~\ref{theorem:existence_of_global_minimizer_unconstrained}) reducing model fitting to a finite-dimensional optimisation even for the infinite-dimensional sphere.
    \item convergence rates under explicit geometric assumptions: smoothness-free rates for the infinite-dimensional case, smoothness-dependent rates under spectral assumptions for the finite-dimensional case, and strong consistency under compactness assumptions (Section~\ref{sec:theoretical_results}).
    \item an efficient algorithm for model fitting that combines BFGS with low-rank approximations (Section~\ref{sec:computational_aspects}), 
    an application to density regression, showcasing our approach as an attractive and competitive solution to a class of imbalanced design problems (Section~\ref{sec:application}), and a ready-to-use R package \texttt{sphereg}\footnote{developer version: \url{https://github.com/Almond-S/sphereg}} together with code to reproduce our results.
\end{enumerate}

The remainder of this paper is organised as follows. 
Section~\ref{sec:notation} introduces the necessary mathematical preliminaries employed throughout the paper. 
We then present our model and estimation methodology in Section~\ref{sec:regression_problem}. 
Our theoretical results are presented and discussed in Section~\ref{sec:theoretical_results}. 
Section~\ref{sec:computational_aspects} gives an overview of the computational aspects of the proposed method. 
We present a simulation study in Section~\ref{sec:simulation} and demonstrate the method's practical utility via an application to regression with density responses in Section~\ref{sec:application}. 
We conclude with a discussion in Section~\ref{sec:discussion}.
Proofs and additional technical results, including compactness properties of VVRKHS balls, appear in the \thesupplement.

\section{Mathematical Preliminaries and Notation}
\label{sec:notation}

\subsection{The sphere as a Riemannian manifold}

Let $\Yspace$ be a real separable Hilbert space with inner product $\YinnerProd{\cdot, \cdot}$ and induced norm $\Ynorm{\cdot}$. We denote by $\bounded(\Yspace)$ the space of bounded linear operators on $\Yspace$, and by $S^\adjoint$ the adjoint operator of $S \in \bounded(\Yspace)$.
For $y_1, y_2 \in \Yspace$, let $y_1 \tensor y_2 \in \bounded(\Yspace)$ be the bounded linear map defined by $(y_1 \tensor y_2)(y) = \YinnerProd{y, y_1} y_2$. Denote by $\sphere = \{ y \in \Yspace \mid \Ynorm{y} = 1 \} \subset \Yspace$ the unit sphere in $\Yspace$. The unit sphere $\sphere$ has a natural Riemannian submanifold structure induced by the inner product $\YinnerProd{\cdot, \cdot}$ \citep{klingenbergRiemannianGeometry2011}. At any point $p \in \sphere$, the tangent space $T_p \sphere$ is given by $T_p \sphere = \{ v \in \Yspace \mid \YinnerProd{v, p} = 0 \}$ and the Riemannian log map $\Log_p: \sphere \setminus \{-p\} \to T_p \sphere$ is defined by 
\(
    \Log_p(q) 
    = \arccos ( \YinnerProd{p, q} ) \frac{\proj_p q }{ \Ynorm{\proj_p q} }, \; q \neq p,
\)
and $\Log_p(p)=0$,
where $\proj_p = \id - p \tensor p : \Yspace \to \Yspace$ is the orthogonal projector with null space $\vspan(p)$, and $\id$ is the identity operator on $\Yspace$. The exponential map at $p \in \sphere$ is the $\continuous^\infty$ map $\Exp_p: T_p \sphere \to \sphere$ defined by
\(
    \Exp_p(v) = \cos \Ynorm{v}\;\cdot  p + \sin\Ynorm{v} \; \cdot \frac{v}{\Ynorm{v}}, \; v \in T_p \sphere,
\)
with the convention $\Exp_p(0) = p$.   Restricted to the open ball in $T_p \sphere$ of radius $\pi$ centred at the origin, $\Exp_p$ is injective, with image $\sphere \setminus \{-p\}$, and its inverse $\Log_p$ is $\continuous^\infty$ on $\sphere \setminus \{-p\}$.
The geodesic distance between two points $p, q \in \sphere$ on the unit sphere is $\sphereDist(p,q) = \arccos( \innerProd{p,q}_\Yspace )$, which is equal to  $\Ynorm{ \Log_p(q) }$ provided $q \neq -p$, and corresponds to the angle between $p$ and $q$. Since $T_p \sphere$ is a Hilbert subspace of $\Yspace$, we denote the norm on $T_p \sphere$ by $\Ynorm{\cdot}$ and its inner product by $\YinnerProd{\cdot, \cdot}$ regardless of $p$, by slight abuse of notation.
When there is no risk of confusion, we shall write $\innerProd{\cdot, \cdot}$ and $\norm{\cdot}$ instead of 
$\YinnerProd{\cdot, \cdot}$ and $\Ynorm{\cdot}$ to alleviate notation.

\subsection{Vector-valued reproducing kernel Hilbert spaces} \label{sec:vvrkhs}

A Hilbert space $\Hspace$ with inner product $\HinnerProd{\cdot, \cdot}$ of functions from an arbitrary space $\Xspace$ to a Hilbert space $\Yspace$ is a \emph{vector-valued reproducing kernel Hilbert space} (VVRKHS) if the pointwise evaluation operators $\eval_x: \Hspace \to \Yspace$ defined by $\eval_x f = f(x)$ are continuous operators for all $x \in \Xspace$. The operator-kernel of the VVRKHS is $K(x, \check x) = \eval_x \eval_{\check x}^\adjoint \in \bounded(\Yspace)$. We often use the notation $K_x := \eval_x^\adjoint :\Yspace \to \Hspace$, which implies $(K_x y)(x') = K(x', x)y$ for $x, x' \in \Xspace, y \in \Yspace$. The reproducing property is
\(
    \YinnerProd{y, f(x)} = \HinnerProd{K_x y, f}, \; \forall x \in \Xspace, y \in \Yspace, f \in \Hspace.
\)
A VVRKHS is uniquely defined by its operator-kernel $K$. 
We shall use in this paper \emph{simple multiplicative operator-kernels} (hereafter \emph{SMO kernels}, which are special cases of \emph{separable} kernels) that are of the form $K(x, x') = k(x,x') \id$, where $k: \Xspace \times \Xspace \to \Real$ is a positive semi-definite kernel, and $\id$ is the identity operator on $\Yspace$. 
Note that if $\Yspace = \Real$ then the VVRKHS is a (scalar) RKHS.
A succinct treatment of VVRKHS is given in \citet[][Chapter 6]{paulsenIntroductionTheoryReproducing2016}.

\subsection{Notation} 
Throughout this paper, we employ standard asymptotic notation to characterise the limiting behaviour of sequences as $n \to \infty$. Let $(a_n)_{n \ge 1}$ and $(b_n)_{n \ge 1}$ be two sequences of real numbers, where $b_n \neq 0$ for all sufficiently large $n$. We write $a_n \lesssim b_n$ if there exists a universal constant $C > 0$ such that $|a_n| \le C |b_n|$ for all sufficiently large $n$, and  $a_n \asymp b_n$ if both $a_n \lesssim b_n$ and $b_n \lesssim a_n$ hold.

We denote the underlying probability space by $(\Omega, \mathcal B, \theLaw)$ and the expectation operator by $\Expected$. For a random element $Z$, we denote its induced probability measure by $\theLaw_Z$.
Let $(Z_n)_{n \ge 1}$ be a sequence of random variables and $(a_n)_{n \ge 1}$ a sequence of strictly positive real numbers. We write $Z_n = O_{\theLaw}(a_n)$ if the sequence $(Z_n / a_n)$ is stochastically bounded, and denote its outer-probability counterpart by $O_{{\theLaw}^*}(a_n)$. 

\section{The Sphere-on-Hilbert regression problem}
\label{sec:regression_problem}

Let $\Xspace$ be a Polish space.
Assume $(X,Y) \in \Xspace \times \sphere$ is a random element of the product space $\Xspace \times \sphere$ (equipped with the Borel $\sigma$-algebra) with joint distribution $\theLaw_{X,Y}$ and marginal distributions $\theLaw_X, \theLaw_Y$. 
The best predictor of $Y$ given $X=x \in \Xspace$ in the least squares sense (with respect to the spherical distance) 
is given by the intrinsic conditional mean \citep[cf.{}][]{davis2010population,petersenFrechetRegressionRandom2019},
\begin{equation} \label{eq:muX}
    \mu(x) = \arg\min_{q \in \sphere} \Expected\left( \sphereDist^2(q, Y) \mid X=x \right).
\end{equation}
The following result shows that $\mu$ is well defined under a bounded conditional support condition.
For the next result, fix a version $x\mapsto\theLaw_{Y\mid X=x}$ of the regular conditional distribution, regarded as a Borel probability kernel on all of $\Xspace$ \citep[Theorem~5.3]{kallenberg1997foundations}. 
\begin{Proposition} \label{proposition:mu_well_defined_everywhere}
    If, for each $x \in \Xspace$, there exist $p_x \in \sphere$ and $\eta_x >0$ such that
 $\theLaw( \sphereDist(Y, p_x) \leq \pi/2 - \eta_x \mid X=x) = 1$,
then there is a measurable function $\mu: \Xspace \to \sphere$ such that for all $x \in \Xspace$, $\mu(x)$ satisfies \eqref{eq:muX} uniquely. 
\end{Proposition}
A bounded support condition is usually required when considering intrinsic means without parametric assumptions, see \citet{afsari2011riemannian,daiStatisticalInferenceHilbert2022}.

For a fixed $\origin \in \sphere$ such that $\mu(x) \neq -\origin$ for all $x \in \Xspace$,
we parametrise the conditional mean on the tangent space through
\begin{equation} \label{eq:f0}
    f_\origin(x) := \Log_\origin(\mu(x)), \quad \forall x \in \Xspace,
\end{equation}
and assume that $f_\origin: \Xspace \to T_\origin \sphere$ belongs to a VVRKHS $\Hspace$ of functions $\Xspace \to T_\origin \sphere$ with SMO kernel $K(x, \check x) = k(x,\check x) \id \in \bounded(T_\origin \sphere)$ for all $x, \check x \in \Xspace$, where $k$ is a scalar Borel measurable kernel and $\id$ is the identity operator on $T_\origin \sphere$.

\begin{Remark}
Even without assuming $f_\origin \in \Hspace$, we have $f_\origin \in L^2(\Xspace, T_\origin \sphere; \theLaw_X)$.
If $\Xspace$ is a locally compact Polish space and $k$ is a $c_0$-universal kernel on $\Xspace$, then for any $\epsilon >0$ there is a $g_\epsilon \in \Hspace$ satisfying $\Ltwonorm{f_\origin - g_\epsilon} \leq \epsilon$ \citep{sriperumbudurRelationUniversalityCharacteristic2010,carmeliVectorValuedReproducing2010}.
\end{Remark}
Let $(X_i,Y_i)_{i=1}^n$ be an i.i.d.\ sample from $(X,Y)$.
We estimate $f_\origin$ by minimising the empirical risk
\begin{equation} \label{eq:penalized_empirical_risk}
    \risk_n(f, \lambda) = \frac 1n \sum_{i=1}^n \sphereDist^2( \Exp_{\origin}(  f(X_i)), Y_i) + \lambda^2 \Hnorm{f}^2,
\end{equation}
where $\lambda > 0$ is a regularisation parameter, and $\Hnorm{\cdot}$ is the VVRKHS norm.

The empirical risk $\risk_n$ can be interpreted as a penalised least-squares loss. However, the squared distance is taken with respect to the sphere's distance $\sphereDist$. Note that the empirical risk is not a convex function over $\Hspace$ since $v \in T_\origin \sphere \mapsto \Exp_\origin(v)$ ``wraps around'' the sphere for $\Ynorm{v} > \pi$.   

\begin{Remark}[Comparison with tangent-space regression]
    \label{remark:tangent_space_regression}
Tangent-space regression models \( x \mapsto \mu_p(x) := \Expected[\Log_p(Y)\mid X=x] \) for a chosen base point \( p \), treating the log-mapped responses as Euclidean data \citep[see, e.g.,][for recent theoretical results]{choiHighdimensionalHilbertSchmidt2025}.
Unlike intrinsic regression, its target $\mu_p$ depends on \( p \) since the sphere is curved, and differs from the conditional intrinsic mean under the manifold's geodesic distance. 
Our approach instead represents the intrinsic conditional mean, and uses the tangent space only to parametrise it: the fitting error remains measured with the geodesic distance.
\end{Remark}

\begin{Remark}[Types of data on a sphere]
\label{remark:types_data_sphere}
Beyond directional data \citep{mardiajupp2009directional} and observations on $\sphere^2 \subset \Real^3$, higher-dimensional spheres arise in several applications. Shape spaces are quotient spaces of pre-shape spheres, making spherical regression a foundation for shape regression \citep{dryden2016statistical,huangGeoFARMGeodesicFactor2021a}. Square-root proportions embed compositional data in a sphere \citep{stephens1982use,watson1989measures}, and square-root density representations identify spherical distance with the Fisher--Rao distance \citep{rao1945,srivastava2007densities}. 
\end{Remark}

Our regression setting accommodates functional predictors and covers density regression as a special case (see Section~\ref{sec:application}). 
Different choices of kernel lead to different
models, including geodesic regression for one scalar covariate \citep{thomasfletcherGeodesicRegressionTheory2013} or its multivariate intrinsic generalisation \citep{corneaRegressionModelsRiemannian2017}.
The global nonlinear Fréchet regression of \citet{bhattacharjeeNonlinearGlobalFrechet2025} also uses an RKHS, but on the predictor side, to form pointwise weighted Fréchet means. We instead fit a single tangent-valued function by penalised intrinsic least squares. Their theoretical assumptions do not cover the sphere, see Remark~\ref{remark:comparison_to_literature} of the \thesupplement.

\section{Theoretical results}
\label{sec:theoretical_results}

\subsection{Representer theorem}

Existence of a minimiser of $f \mapsto \risk_n(f, \lambda)$ is nontrivial: closed bounded sets in $\Hspace$ need not be compact, and the spherical loss is nonconvex.
Indeed, Fr\'echet means need not exist on an infinite-dimensional sphere \citep{tavakoli2026existencenonexistencecentresmass}. 
The following theorem establishes existence of a minimiser for our penalised risk by reducing the infinite-dimensional optimisation to a finite-dimensional one.
Throughout, $(X_i,Y_i)_{i=1}^n$ are independent copies of $(X,Y)$.

\begin{Theorem} \label{theorem:existence_of_global_minimizer_unconstrained}
  Assume that $Y \neq - \origin$ almost surely and $\sup_{x \in \Xspace} k(x,x) < \infty$, and that \newline
  $\dim \vspan(Y_1, \ldots, Y_n, \origin) +n  < \dim \Yspace$.
  For any choice of orthogonal vectors $w_1,\ldots, w_n \in \vspan(Y_1,\ldots, Y_n, \origin)^\perp \cap \sphere$ 
  we have
    \begin{equation} \label{eq:inf_reduces_to_finite_dim_inf}
        \inf_{f \in \Hspace} \risk_n(f, \lambda)= \inf_{f \in \widetilde \Hspace} \risk_n(f, \lambda),
    \end{equation}
    where 
    \(
        \widetilde \Hspace = \left\{ f = \sum_{i=1}^n K_{X_i} \xi_i \mid \xi_i \in \vspan\left( \{ \Log_\origin Y_i \}_{i=1}^n, \{ w_i \}_{i=1}^n \right) \right\}.
    \)
    In particular, for $\lambda > 0$, the infimum in \eqref{eq:inf_reduces_to_finite_dim_inf} is achieved for some $\hat f_n \in \widetilde \Hspace$.
    If $\hat f_n$ is unique, then 
    $\hat f_n = \sum_{i=1}^n K_{X_i} \xi_i$ with
    \begin{equation} \label{eq:xi_range_if_unique_min}
      \xi_1,\ldots, \xi_n \in \vspan\{ \Log_\origin Y_i \}_{i=1}^n.
    \end{equation}

    If $\dim \vspan(Y_1, \ldots, Y_n, \origin) +n  \geq \dim \Yspace$, then we can take $\xi_i \in T_\origin \sphere$ in the definition of $\widetilde \Hspace$, and \eqref{eq:xi_range_if_unique_min} holds if $\hat f_n$ is unique.
    
\end{Theorem}
Theorem~\ref{theorem:existence_of_global_minimizer_unconstrained} reduces estimation to finitely many scalar coefficients. 
In the infinite-dimensional case, each $\xi_i$ lies in the span of at most $2n$ fixed vectors, giving at most $2n^2$ coefficients overall to optimise over.
The proof combines spherical isometries with a minimum-norm interpolation result for VVRKHS (Lemma~\ref{Lemma:minimum_interpolant_vvrkhs} in the \thesupplement).

\subsection{Convergence rates}
\label{sec:convergence_rates}

We study convergence to $f_\origin$, the minimiser of the population risk
\begin{equation} \label{eq:population_risk}
    \risk(f) = \Expected \risk_n(f, 0) = \Expected \sphereDist^2( \Exp_\origin f(X), Y),
\end{equation}
see Proposition~\ref{proposition:mu_well_defined_everywhere}. 
The analysis combines local convexity, stability arguments, and empirical process techniques.
Note that the basic term in our empirical and population risk is the function 
\begin{equation} \label{eq:ell_function}
    \ell_y(v) := \sphereDist^2(\Exp_\origin v, y),\quad   v \in T_\origin \sphere, y \in \sphere. 
\end{equation}
For $v \neq 0$ fixed, the function $t \in \Real \mapsto \ell_y(tv)$ is $2\pi/\Ynorm{v}$-periodic, so $\risk(f)$ is not globally convex. 

\subsubsection{Local convexity}
\label{sec:local_convexity}

We reduce the study of the convexity of $\risk_n(\cdot, \lambda)$ to the study of the function $\ell_y$. Indeed, the empirical risk can be written as
\(    \risk_n(f, \lambda) = \frac 1n \sum_{i=1}^n \ell_{Y_i}(f(X_i)) + \lambda^2 \Hnorm{f}^2. \)
Although we know that the mapping $s \in \sphere \mapsto \sphereDist^2(s,y)$ has a positive definite Hessian if $\sphereDist(s,y) < \pi/2$ \citep{pennecBarycentricSubspaceAnalysis2018}, the convexity of $\ell_y$ does not follow because of the composition $s = \Exp_\origin (v)$, which does not map all lines in $T_\origin \sphere$ to geodesics in $\sphere$. The Hessian of $\ell_y$ is a self-adjoint operator on $\Yspace$ and thus high-dimensional, but it turns out that
analysing the eigenvalues of the Hessian of $\ell_y$ boils down to analysing the eigenvalues of a $3 \times 3$ matrix, which yields the following result.
\begin{Theorem} \label{theorem:ell_Hessian_new}
    Let 
      $\Cregion := \{ v \in T_\origin \sphere \mid \Ynorm{v} \leq \pi/4 \}.$ 
    If $\sphereDist(\origin, y) \leq \pi/4$,
    \begin{enumerate}
        \item The minimal eigenvalue of the Hessian of $\ell_y$ is at least $\epsilon $ on $\Cregion$, for some $\epsilon > 0$.
        \item $\nabla \ell_y(v)$, the gradient of $\ell_y$ in $T_\origin \sphere$, is well defined on $\Cregion$ and
            \[
                \ell_y(v') \geq \ell_y(v) + \YinnerProd{v'-v, \nabla \ell_y(v)} + \frac{\epsilon}{2} \Ynorm{v'-v}^2.
            \]
            In particular, $\ell_y$ is convex on $\Cregion$.
    \end{enumerate}
    
\end{Theorem}
The proof of this result is technically challenging, as it requires bounding several functions by polynomials with rational coefficients, and then checking that these polynomials are strictly positive. If one allows for an argument using numerical analysis, $\Cregion$ can be extended and the condition $\sphereDist(\origin, y) \leq \pi/4$ can be relaxed. A thorough discussion is given in Section~\ref{sec:discussion}.

As a consequence of Theorem~\ref{theorem:ell_Hessian_new}, define the deterministic set
$ \Uset := \left\{ f \in \Hspace \,\middle|\, \Ynorm{f(X)} \leq \pi/4, \; \mathrm{a.s.} \right\}.$
If $\sup_x k(x,x) < \infty$, $\Uset \subset \Hspace$ is non-empty and closed (Lemma~\ref{Lemma:Uset_is_closed_and_nonempty_new} in the \thesupplement). 
We shall rely on the following assumption for our next results.
\begin{Assumption}
    \label{assumption:support_of_Y_for_theory_new}
    Assume that $\sphereDist(\origin, Y) \leq \pi/4$ almost surely, and that
    $f_\origin \in \Uset$.
\end{Assumption}
Note that the first part of Assumption~\ref{assumption:support_of_Y_for_theory_new} automatically implies that $\sphereDist(\origin, \mu(X)) \leq \pi/4$ almost surely  \citep[][Theorem B and 57]{yokotaConvexFunctionsBarycenter2017}, and hence $\Ynorm{f_\origin(X)} \leq \pi/4$ almost surely.
The second part of Assumption~\ref{assumption:support_of_Y_for_theory_new} is not a new requirement: since $f_\origin \in \Hspace$ was already assumed (a smoothness condition on $\mu$), the first part gives $f_\origin \in \Uset$. We nevertheless state it explicitly so that the results below can refer to a single numbered assumption.
Although the bounded support assumption $\sphereDist(\origin, Y) \leq \pi/4$ might seem quite restrictive, it is comparable to implicit assumptions in related literature on manifold regression or metric space regression  \citep{choiHighdimensionalHilbertSchmidt2025,petersenFrechetRegressionRandom2019}. 
The latter works operate in a general abstract setting, but assume existence and uniqueness of the population risk minimiser and make abstract assumptions about the growth of the population risk, see  \citet[][Assumptions (M1), (M4)]{choiHighdimensionalHilbertSchmidt2025} or \citet[][Assumptions (U0), (P0), (U2), (P2)]{petersenFrechetRegressionRandom2019}. 
Our work provides concrete and interpretable assumptions.
Using numerical analysis, the assumption can be relaxed to $\sphereDist(\origin, Y) \leq 2\pi/3$ almost surely, but with an additional noise restriction, see Section~\ref{sec:discussion}.

Let $\Ltwonorm{f} = ( \int_\Xspace \Ynorm{f(x)}^2 d\theLaw_X(x) )^{1/2}$ denote the $\LtwoXTS$ norm of a measurable function $f: \Xspace \to T_\origin \sphere$. 
The following result plays a central role in the proofs of our convergence rates.
\begin{Theorem} \label{theorem:risk_is_convex_new}
    Under Assumption~\ref{assumption:support_of_Y_for_theory_new}, provided $\sup_x k(x,x) < \infty$, there is an $\epsilon> 0$ such that 
    \begin{enumerate}
        \item For all $f \in \Uset$, $\risk(f) \geq \risk(f_\origin) + \frac{\epsilon}{2} \Ltwonorm{f-f_\origin}^2$. 
        \item For $\lambda_n > 0$, $\risk_n(\cdot, \lambda_n)$ is almost surely $2\lambda_n^2$-strongly convex on $\Uset$, and it almost surely admits a unique minimiser $\hat f_n$ over $\Uset$.
    \end{enumerate}
    
\end{Theorem}

Throughout the rest of Section~\ref{sec:convergence_rates}, $\hat f_n$ is a minimiser of $\risk_n(f, \lambda_n)$ over all $f \in \Uset$.

\subsubsection{Rates of convergence under minimal assumptions}
\label{sec:rates_minimal_assumptions}
The following result gives rates of convergence under very weak assumptions.
\begin{Theorem}[smoothness-free convergence rates] \label{theorem:rates_for_f_weakest_assumptions_corrected}
    Assume  $\sup_{x \in \Xspace} \sqrt{k(x,x)} = c < \infty$, and let $\hat \mu_n = \Exp_\origin \circ \hat f_n$.
    Under Assumption~\ref{assumption:support_of_Y_for_theory_new}, for $\lambda_n^2 = 8\pi^2 c^2 n^{-1/2}$,
    each of the following bounds holds with probability $1-\delta$, for any $\delta \in (0, 1)$:
\[
\Ltwonorm{ \hat f_n - f_\origin } \leq \tilde c \delta^{-1/2} n^{-1/4},  \qquad  \int_\Xspace \sphereDist^2(\mu(x), \hat \mu_n(x)) \ud\theLaw_X(x) \leq \tilde c^2 \delta^{-1} n^{-1/2},
\]
    where $\tilde c$ is a constant depending only on $\Hnorm{f_\origin}, c$ and $\epsilon$ defined in Theorem~\ref{theorem:risk_is_convex_new}.
    In particular,
    \[
        \Ltwonorm{\hat f_n - f_\origin} = O_{\theLaw^*}(n^{-1/4}) \quad \text{and} \quad \int_\Xspace \sphereDist^2(\mu(x), \hat \mu_n(x)) \ud\theLaw_X(x) = O_{\theLaw^*}(n^{-1/2}).
    \]
\end{Theorem}

\begin{Remark}
\label{remark:low_rates}
Theorem~\ref{theorem:rates_for_f_weakest_assumptions_corrected}
holds for arbitrary Polish covariate spaces and spheres of arbitrary dimension,
and the $n^{-1/4}$ rate in $L^2(\theLaw_X)$ does not depend on the response dimension, which can be infinite. 
It holds without assumptions on kernel smoothness or eigenvalue decay.
Beyond $f_\origin \in \Hspace$, no source conditions are required.
Under similar assumptions, the same $n^{-1/4}$ rate is obtained for kernel ridge regression \citep{caponnettoOptimalRatesRegularized2007}.
A more thorough discussion is given in Remark~\ref{remark:detailed_low_rates} of the \thesupplement.
\end{Remark}

\subsubsection{Rates of convergence depending on VVRKHS smoothness}
\label{sec:rates_depending_on_smoothness}

Under stronger assumptions, we can derive faster rates of convergence that depend on the smoothness of the VVRKHS functions.

Let $\Hspace_k$ be the scalar RKHS associated with $k$. 
For SMO kernels, every scalar projection $x\mapsto\YinnerProd{f(x),y}$, with $f\in\Hspace$ and $y\in T_\origin\sphere$, belongs to $\Hspace_k$ (Lemma~\ref{Lemma:projection_against_y_is_in_scalar_RKHS} in the \thesupplement).
The smoothness of $\Hspace$ functions is measured directly by the smoothness of the $\Hspace_k$ functions via the integral operator $\KIntegralOperator: \LtwoX \to \LtwoX$, 
\begin{equation} \label{eq:KIntegralOperator_defn}
    (\KIntegralOperator f)(x') = \int k(x',x)f(x) \ud\theLaw_X(x), \quad f \in \LtwoX,
\end{equation}
where $\LtwoX$ is the space of measurable functions $f:\Xspace \to \Real$ with $\int f^2(x) \ud\theLaw_X(x) < \infty$.
If $\Expected k(X,X) < \infty$, the operator $\KIntegralOperator$ is trace-class and self-adjoint  \citep[][Lemma~2.3]{steinwartMercerTheoremGeneral2012} and admits a spectral decomposition. Let
$(\sigma_l)_{l \geq 1}$ denote its non-negative eigenvalues, listed in non-increasing order. These eigenvalues form a summable sequence.
The rates of decay of the eigenvalues $\{\sigma_l\}$ encode the complexity of the RKHS $\Hspace_k$ and its interaction with the probability measure $\theLaw_X$. They impact the convergence rates of our estimators. 
We can now state a result on smoothness-dependent rates of convergence of our estimator.
\begin{Theorem}[smoothness-dependent convergence rates]
\label{theorem:rates_for_f_depending_on_smoothness_new}
Suppose that Assumption~\ref{assumption:support_of_Y_for_theory_new} holds,
$\dim(\Yspace)<\infty$, and $\sup_x k(x,x)<\infty$.
Let $\hat f_n$ be the minimiser over $\Uset$ from
Theorem~\ref{theorem:risk_is_convex_new}, and set
$\hat\mu_n=\Exp_\origin\circ\hat f_n$.
Then
\[
    \Ltwonorm{\hat f_n-f_\origin}^2 =O_{\theLaw^*}(r_n)
    \quad \text{and} \quad
    \int_\Xspace \sphereDist^2(\mu(x),\hat\mu_n(x))\,\ud\theLaw_X(x) =O_{\theLaw^*}(r_n),
\]
where $r_n$ depends on the decay rates of the eigenvalues of $\KIntegralOperator$, as follows:
\begin{center}
\begin{tabular}{@{}llll@{}}
\toprule
Decay name & Eigenvalue condition & $\lambda_n\asymp$ & $r_n$ \\
\midrule
Finite-rank & $\sigma_l=0$ for $l>L$, $L<\infty$ & $e^{-n}$ & $n^{-1}$ \\
Polynomial decay & $\sigma_l\lesssim l^{-p}$, $p>1$ & $n^{-p/[2(p+1)]}$ & $n^{-p/(p+1)}$ \\
Stretched exponential decay & $\sigma_l\lesssim e^{-\alpha l^{1/q}}$, $\alpha,q>0$
    & $n^{-1/2}$ & $(\log n)^q/n$ \\
\bottomrule
\end{tabular}
\end{center}
\end{Theorem}

\begin{Remark}[Decay rates for concrete kernels]
  \label{remark:eigenvalue_decay_rates}
Obtaining the eigenvalue decay rates for general distributions $\theLaw_X$ and general kernels $k$ is non-trivial. A well-understood case is that of translation-invariant kernels $k(x,x') = \varphi(x-x')$ on $\Xspace = \Real^s$, when $\theLaw_X$ has bounded density with respect to Lebesgue measure and bounded support. This includes the Mat\'ern kernel of smoothness $\nu$, for which the eigenvalue decay rate is polynomial with $p = 1 + 2\nu/s$ (the decay is faster for greater smoothness and slower for higher dimensions),
and the Gaussian kernel $k(x,x') = \exp(- \norm{x-x'}^2/{2\sigma^2} )$, $\sigma > 0$, for which the eigenvalue decay is a stretched exponential with parameter $q=s$.
\end{Remark}

\begin{Remark}[Comparison with existing rates]
The rates in Theorem~\ref{theorem:rates_for_f_depending_on_smoothness_new} are comparable to kernel ridge regression rates without additional source conditions. 
For finite-dimensional responses, \citet{caponnettoOptimalRatesRegularized2007} obtain the same finite-rank rate and comparable polynomial decay rates, up to logarithmic factors. 
The stretched-exponential decay rate also agrees with kernel ridge regression rates \citep{bakEffectDimensionalityConvergence2025}. 
The finite-rank case yields the parametric $n^{-1/2}$ rate in $\Ltwonorm{\cdot}$, consistent with parametric regression on finite-dimensional manifolds \citep{corneaRegressionModelsRiemannian2017,choiHighdimensionalHilbertSchmidt2025}.

These comparisons concern rates of integrated prediction error under the respective assumptions.
The pointwise rates in Fr\'echet regression \citep{petersenFrechetRegressionRandom2019,bhattacharjeeNonlinearGlobalFrechet2025} concern different error criteria and assumptions from ours.
A more thorough comparison is given in Remark~\ref{remark:detailed_smoothness_rates} of the \thesupplement.
\end{Remark}

\subsection{Strong consistency}
\label{sec:strong_consistency}

In this section, we show the outer almost sure consistency, in $\Ltwonorm{\cdot}$ and uniformly on the support of $\theLaw_X$, of the empirical risk minimiser over the set
$   \widetilde \Uset  = \{f \in \Uset \mid \Hnorm{f} \leq C \},$
where $C> \Hnorm{f_\origin}$ is some arbitrary constant. 
We shall need the following assumptions.
\begin{Assumption}\label{assumption:for_almost_sure_consistency}
         $\Xspace$ is a compact metric space,  $k: \Xspace \times \Xspace \to \Real$ is continuous, and  $\Yspace$ is a finite-dimensional Hilbert space (in particular, $\sphere$ is a finite-dimensional sphere).
\end{Assumption}

Let $\outerASconv$ denote outer almost sure convergence, let $\tilde f_n$ be the minimiser of $\risk_n(\cdot, \lambda_n)$ over $\widetilde \Uset$, and let $\tilde \mu_n = \Exp_\origin \circ \tilde f_n$.
\begin{Theorem}\label{theorem:strong_consistency_new} 
  Assume that Assumptions~\ref{assumption:support_of_Y_for_theory_new} and~\ref{assumption:for_almost_sure_consistency} hold.
    If $\lambda_n \downarrow 0$ as $n \to \infty$,
    \begin{enumerate}
        \item $\Ltwonorm{\tilde f_n - f_\origin} \outerASconv 0$ and  $\int_\Xspace \sphereDist^2(\tilde \mu_n(x), \mu(x)) d\theLaw_X(x) \outerASconv 0$  as $n \to \infty$;
        \item If $\KIntegralOperator \neq 0$, $\sup_{x \in \support(\theLaw_X)} \Ynorm{\tilde f_n(x) - f_\origin(x)} \outerASconv 0$ and 
        
        $\sup_{x \in \support(\theLaw_X)} \sphereDist(\tilde \mu_n(x), \mu(x)) \outerASconv 0$, as $n \to \infty$.
    \end{enumerate}
\end{Theorem}

\section{Overview of computational aspects}
\label{sec:computational_aspects}

The representer theorem (Theorem~\ref{theorem:existence_of_global_minimizer_unconstrained}) reduces model fitting to finite-dimensional optimisation.
The nonlinear spherical loss, however, requires an iterative algorithm to compute $\hat f_n$. 
To reduce computational cost and improve conditioning, we use low-rank approximations of the kernel and response Gram matrices.

We fix the intercept $\origin$ at the sample Fr\'{e}chet mean and initialise the algorithm with tangent-space ridge regression using the same kernel and penalty.
We then minimise the spherical objective with BFGS, which uses an analytic tangent-space gradient and approximates the inverse Hessian.
The penalty and, where applicable, the kernel bandwidth are selected by cross-validation.

Further computational details are provided in Section~\ref{sec:detailed_computational_aspects} of the \thesupplement{}, and an implementation is available in the \texttt{R} package \texttt{sphereg}.

\section{Simulation study: spherical functional responses}
\label{sec:simulation}
\subsection{Overview}

In this section, we investigate the finite-sample performance of our proposed estimator. We focus on response variables $Y$ belonging to the sphere $\sphere$ in $L^2([0,1])$, that is, $\int_0^1 Y^2(t) dt = 1$.

Although our simulation study draws inspiration from its application to probability density functions (Section~\ref{sec:application}), we extend the scope by allowing the function $Y$ to take negative values. This makes our approach more general than the density-estimation case, and the two are therefore not directly comparable. Furthermore, we consider scenarios that fall outside our consistency assumptions (Assumption~\ref{assumption:support_of_Y_for_theory_new}).  We demonstrate that our estimator exhibits robust performance even beyond the scope of these theoretical guarantees.

We benchmark our model against tangent-space regression. This contrast effectively illustrates the benefit of utilising an intrinsic geometry over a linear approximation. The performance gain is sizeable when data points are sufficiently spread apart on the sphere.

We draw covariates uniformly from $\Xspace = [0,1]^2$ and construct a conditional mean function $\mu(x)$ whose shape and location vary with $x$, using a mixture of two Gaussian densities as a template. 
The values of $\mu(x)$ reach spherical distances of $4\pi/5$ from their Fr\'{e}chet mean. 
Responses are generated by adding smooth Gaussian perturbations in the tangent space at $\mu(x)$ and mapping them onto the sphere with the exponential map. 
The perturbations use a $20$-term Fourier expansion with coefficient variances proportional to $j^{-4}$, transported from a common tangent space. 
We consider sample sizes $n\in\{20,50,100,200\}$ and noise levels $\tau_\varepsilon\in\{\pi/10,2\pi/10,4\pi/10\}$, with $100$ independent Monte Carlo replicates per setting. 
Figure~\ref{fig:simulation:illustration} illustrates the mean function and the effect of increasing noise.

The competitor, tangent-space regression, maps responses to the tangent space at their sample Fr\'{e}chet mean, fits vector-valued kernel ridge regression there, and maps predictions back to the sphere.

Both methods use this same reference point and a Gaussian radial basis kernel with fixed bandwidth. 
Their penalties are selected by five-fold cross-validation. 
The comparison therefore assesses the benefit of fitting with spherical rather than tangent-space squared distances.

Performance is measured by the mean squared spherical distance between the estimated and true conditional means, $\widehat{\operatorname{MSE}}=100^{-1}\sum_{b=1}^{100} \Expected_X[\sphereDist^2(\hat\mu_b(X),\mu(X))]$, where $\hat\mu_b$ is the estimate from Monte Carlo replicate $b$ and the expectation over $X$ is approximated on a grid. 
Figure~\ref{fig:sim_performance} reports these errors with $95\%$ confidence intervals. 
The methods perform similarly at $n=20$, whereas spherical kernel ridge regression achieves smaller errors at larger sample sizes and higher noise levels, with an increasing advantage over tangent-space regression.
Note that these results illustrate the performance of our method beyond the assumptions of our consistency theory. 
Section~\ref{sec:details_simulation_study} of the \thesupplement{} provides the full design of the simulations and further details, and  the construction of $\mu$ is given in Section~\ref{sec:simulation_mu}.

\begin{figure}
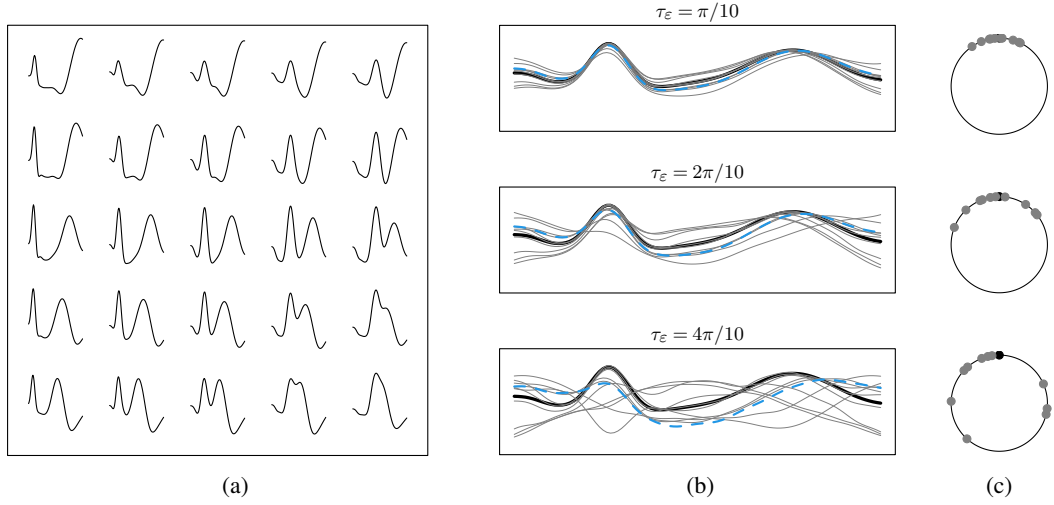

    \centering
    \begin{tikzpicture}
    	\node[anchor=south west] (mu)
    	at (0,0) {\scalebox{1.021}{\input{supplementary-files/figs/simulation/simulation_illustration-mu.tex}}};
    	\node[anchor=south west] (y)
    	at ([xshift=-2pt]mu.south east)
    	{\scalebox{1.02}{\input{supplementary-files/figs/simulation/simulation_illustration-y.tex}}};
    	\node[anchor=south west] (sphere)
    	at ([xshift=-15pt]y.south east)
    	{\input{supplementary-files/figs/simulation/simulation_illustration-sphere.tex}};

    	\node[anchor=center] at (mu.south)     {(a)};
    	\node[anchor=center] at (y.south)      {(b)};
    	\node[anchor=center] at (sphere.south) {(c)};
    \end{tikzpicture}
    \caption{Panel (a) shows values of $\mu(x_1, x_2)$ for different locations $(x_1, x_2)$ on the unit square. Panel (b) shows $\mu(0.5, 0.5)$ (thick black line) along with the fitted $\hat \mu(0.5, 0.5)$ (thick blue dashed line) for each of the three noise levels considered in the simulation study. Ten realisations of $Y | X = (0.5, 0.5)$ are also plotted to illustrate the sampling variability for each given noise level (thin grey lines).
      Panel (c) shows the spherical distances between ten realisations of $Y | X = (0.5, 0.5)$ (grey dots) and $\mu(0.5, 0.5)$ (black dot) for each noise level shown in the corresponding row of panel (b).
    }
    \label{fig:simulation:illustration}
\end{figure}

\begin{figure}[ht]
    \centering
    \input{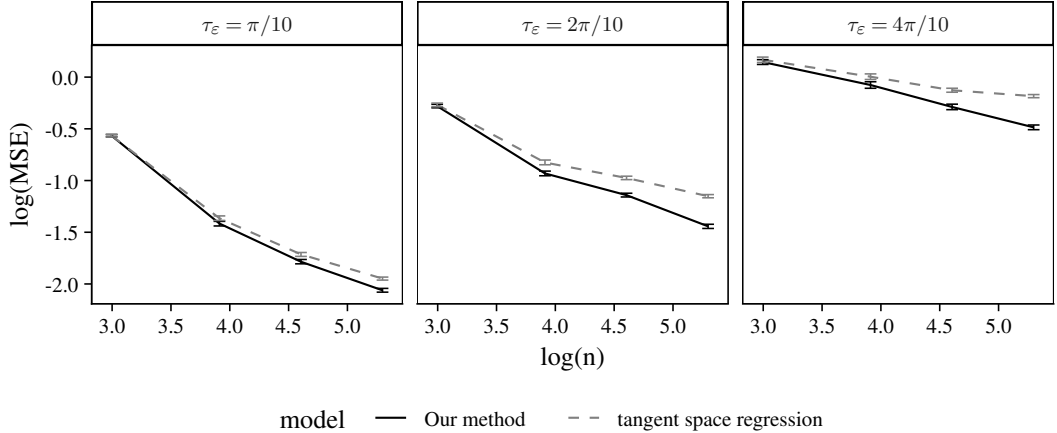}
    \caption{
      Estimates of the MSE in our simulation study on a log--log scale, with increasing noise variance $\tau_\varepsilon^2$ from left to right.
      In all plots, the thick black lines represent our proposed method (spherical kernel ridge regression), and the dashed grey lines represent tangent-space regression.
      The small boxplots indicate 95\% confidence intervals for the MSE.
    }
    \label{fig:sim_performance}
\end{figure}

\section{Spherical kernel ridge regression for densities: analysis of urban pollution}
\label{sec:application}

\subsection{Overview}
\label{sec:application_overview}

We analyse the aerosol particle light absorption coefficient (AAC), which quantifies how strongly aerosol particles absorb light per unit path length and thus contributes to atmospheric warming and radiative forcing. We use our proposed method (spherical kernel ridge regression) to model square-root densities of the AAC probability distribution at different locations in the Netherlands.

The presented data problem is representative of a common setting in geophysical and environmental applications, where automated measurement devices record a variable $Z \in \mathcal Z$ at high temporal frequency but only at a small number of geographical locations $X$. 
This yields few covariate values $x_1,\dots,x_n$ and many samples $z_{i1},\dots,z_{in_i}$ from $Z \mid X = x_i$. 
Such data are common, for instance, in oceanographic monitoring \citep{ARGO} or radiosonde observations \citep{IGRA}.
In practice, modelling all raw observations $(z_{ij})$ is computationally costly and requires strong distributional and dependence assumptions, and $Z \mid X$ is therefore often summarised by spatio-temporal aggregates such as means 
or similar summaries \citep{gruzieva2024monitoringresolution}, which suppress within-location variability and change the interpretation of the response. 
We adopt an intermediate strategy: for each location $x_i$, we estimate nonparametrically the conditional density of $Z \mid X = x_i$, obtaining estimated densities $h_i$ and setting the responses $y_i = \sqrt{h_i}$, which belong to the sphere in $L^2(\mathcal Z)$. We then regress the responses $y_1,\dots,y_n$ on the covariates $x_1,\ldots, x_n$.
Indeed, the particle light absorption coefficient is largely
driven by black carbon aerosols from traffic and residential combustion, and black carbon
mass concentrations are typically inferred from absorption via an assumed mass
absorption cross section. Variability and uncertainty in absorption therefore propagate
directly to black carbon estimates \citep{moosmuller2009aerosol}. This further motivates our
distributional approach, rather than a purely mean-based characterisation of absorption.

We use hourly absorption measurements from $14$ air-quality monitoring stations in the Netherlands and neighbouring regions (Belgium, Germany), obtained from the ACTRIS data portal \citep{ACTRIS}. Stations outside the Netherlands are included mainly to mitigate boundary effects, while inference focuses on the Dutch domain. The data we use are from 2019 for all stations with sufficiently complete records. For the station in B\"osel (Germany), where no 2019 data are available, we substitute data from the last chronologically available year (2014). Each station $i$ is represented by its geographic coordinates $x_i =$ (latitude, longitude), which serve as spatial covariates. Over the small and low-relief study region, we treat these coordinates as points in $\Real^2$.

Due to the data-collection mechanism and measurement errors, occasional small negative values of the absorption coefficient arise at low concentrations: these reflect uncertainty rather than physical absorption, and are retained in the analysis. Importantly, the presence of occasional negative observations does not pose difficulties for the proposed regression framework, whereas methods relying on log-transformations of the raw data \citep[such as][]{alas2022pedestrian} would require ad hoc adjustments or data truncation. 

For each station, the observations aggregated over one year are treated as a sample from an underlying distribution that is summarised by a kernel density estimator restricted to an interval $\mathcal{Z} = [z_{\min}, z_{\max}]$, where $z_{\min}$ and $z_{\max}$ are taken as the minimum and maximum, respectively, of all observed values across all stations in the dataset. The estimated densities are mapped to the unit sphere in $L^2(\mathcal Z)$ via the square-root transformation, yielding responses on the positive orthant of a Hilbert sphere.

Figure~\ref{fig:aac_map} shows the fit of the spherical kernel ridge regression model, with a Gaussian kernel defined on station latitude-longitude coordinates. We performed leave-one-station-out cross-validation over a grid of hyperparameters to select the best shape parameter ($\sigma =15.2$) and regularisation parameter ($\lambda = 1/4$). The optimal model---fitted with the best parameters---can then be used in many ways, such as for predicting the absorption coefficient distributions at locations without monitoring stations (e.g., Amsterdam). Furthermore, estimates of different  characteristics of the distribution,
such as quantiles of the absorption coefficient, can be obtained from the predicted distributions. Figure~\ref{fig:aac_map} shows the $90\%$ quantile and reveals higher pollution levels in the central-western part of the Netherlands. Direct access to higher quantiles, rather than only the mean, is particularly useful for environmental decision-making in this context, as it highlights regions where absorption may persistently exceed high pollution thresholds and may warrant targeted mitigation.

\begin{figure}
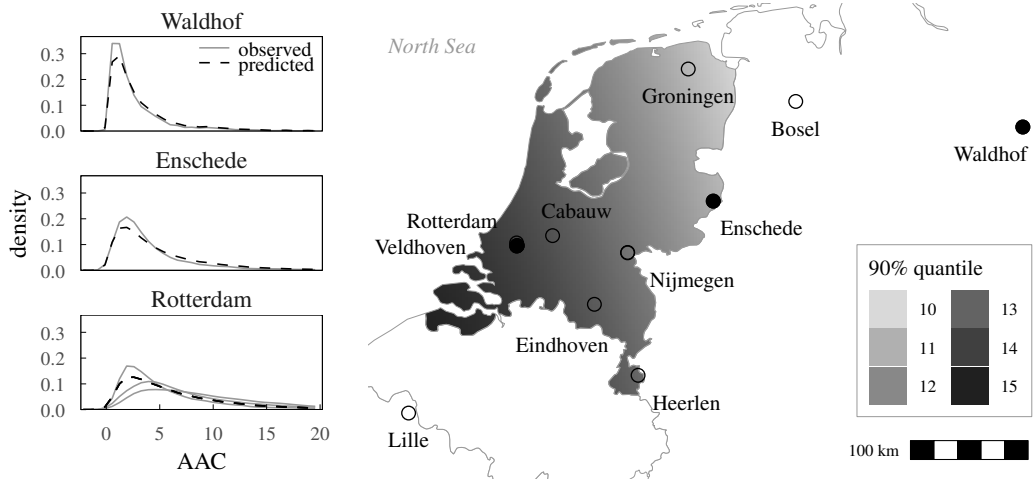

    { \scriptsize
	\input{supplementary-files/figs/application/aac_model_plot-densities.tex} 
	\input{supplementary-files/figs/application/aac_model_plot-map.tex}
    }
	\caption{
    Density comparison across geographical sites. The left panel displays the observed versus fitted densities (using our method) for three study locations. Note that Rotterdam contributes three distinct density observations due to its multiple measuring stations. The right panel provides a geospatial map of the dataset locations, showing the 90\% fitted quantile as determined by our method. Locations displayed in the left panel are marked by filled circles on the map. The other locations are marked by open circles.
  }
		\label{fig:aac_map}
\end{figure}

\subsection{Comparison of density regression geometries}
\label{sec:comparison_of_methods_in_application}

We compare four approaches to modelling the aerosol absorption coefficient densities:
(i) Fréchet regression in the $2$-Wasserstein space \citep{petersenFrechetRegressionRandom2019} (hereafter \frechetModel), 
(ii) the approach of \citet{maier2025additive} to regression in the Bayes Hilbert space of densities (hereafter \bayesModel) and (iii) our approach (hereafter \sphericalModel) described in Section~\ref{sec:application_overview}, performing regression on the Hilbert sphere of square-root densities \citep[which corresponds to an analysis based on the Fisher--Rao metric of the original density functions;][]{rao1945,srivastava2007densities}. In addition to these approaches for density responses, we also add (iv) the lognormal-normal convolution
model (hereafter \logNNModel) proposed by \citet{alas2022pedestrian} for estimating pollution concentration distributions directly on the raw AAC samples $\{ (x_i, z_{ij}) \}$, implemented in the \texttt{R} package \texttt{bamlss}.

We evaluate out-of-sample performance using a leave-one-station-out procedure, and assess the performance of the methods under three different metrics: the spherical distance, the Bayes-space distance (which \bayesModel\ optimises for) and the 2-Wasserstein distance (which \frechetModel\ optimises for).
The boxplots in Figure~\ref{fig:boxplot_errors} (a,b,c) illustrate the leave-one-station-out prediction errors across our metrics. Our proposed method, \sphericalModel\ (highlighted in dark grey), is expected to perform strongly under spherical distance since it is inherently based on this geometry, and does indeed outperform other methods under the spherical metric. Furthermore, \sphericalModel\ demonstrates a strong capability by also excelling in the Bayes-space distance, though \bayesModel\ remains the second-best method. Regarding the $2$-Wasserstein distance $W_2$, while \sphericalModel\ is slightly surpassed by both \frechetModel\ and \bayesModel, its performance significantly exceeds that of \logNNModel: this sub-optimal result under $W_2$ is unsurprising given that our method was not optimised using this metric.

In terms of computational cost (Figure~\ref{fig:boxplot_errors} (d)), our method exhibits superior stability and efficiency. Its total runtime is consistently below the median runtime observed for \frechetModel, and critically, it is several orders of magnitude faster than both \bayesModel\ and \logNNModel.

Further details are given in Section~\ref{sec:application_further_details} of the \thesupplement.

\begin{figure}
\centering
{\fontsize{10}{12}\selectfont
\begin{tikzpicture}
	\node[anchor=south west, inner sep=0pt] (mu)
	at (0,0) {\input{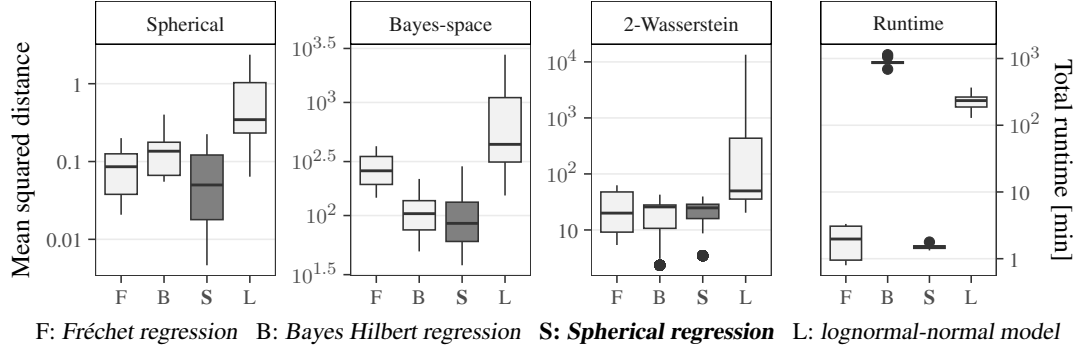}};
	\node[anchor=south west, inner sep=0pt] (y)
	at ([xshift=6pt]mu.south east)
	{
\begin{tikzpicture}[x=1pt,y=1pt]
\definecolor{fillColor}{RGB}{255,255,255}
\path[use as bounding box,fill=fillColor,fill opacity=0.00] (0,0) rectangle ( 97.56,119.25);
\begin{scope}
\path[clip] (  0.00,  0.00) rectangle ( 97.56,119.25);
\definecolor{drawColor}{RGB}{255,255,255}
\definecolor{fillColor}{RGB}{255,255,255}

\path[draw=drawColor,line width= 0.6pt,line join=round,line cap=round,fill=fillColor] ( -0.00,  0.00) rectangle ( 97.56,119.25);
\end{scope}
\begin{scope}
\path[clip] (  0.00, 13.72) rectangle ( 67.40,102.67);
\definecolor{fillColor}{RGB}{255,255,255}

\path[fill=fillColor] ( -0.00, 13.72) rectangle ( 67.40,102.67);
\definecolor{drawColor}{gray}{0.92}

\path[draw=drawColor,line width= 0.6pt,line join=round] ( -0.00, 20.28) --
	( 67.40, 20.28);

\path[draw=drawColor,line width= 0.6pt,line join=round] ( -0.00, 45.91) --
	( 67.40, 45.91);

\path[draw=drawColor,line width= 0.6pt,line join=round] ( -0.00, 71.54) --
	( 67.40, 71.54);

\path[draw=drawColor,line width= 0.6pt,line join=round] ( -0.00, 97.17) --
	( 67.40, 97.17);
\definecolor{drawColor}{gray}{0.20}

\path[draw=drawColor,line width= 0.6pt,line join=round] (  9.63, 32.77) -- (  9.63, 33.63);

\path[draw=drawColor,line width= 0.6pt,line join=round] (  9.63, 19.73) -- (  9.63, 17.76);
\definecolor{fillColor}{gray}{0.95}

\path[draw=drawColor,line width= 0.6pt,fill=fillColor] (  3.61, 32.77) --
	(  3.61, 19.73) --
	( 15.65, 19.73) --
	( 15.65, 32.77) --
	(  3.61, 32.77) --
	cycle;

\path[draw=drawColor,line width= 1.1pt] (  3.61, 27.84) -- ( 15.65, 27.84);
\definecolor{fillColor}{gray}{0.20}

\path[draw=drawColor,line width= 0.4pt,line join=round,line cap=round,fill=fillColor] ( 25.68, 93.06) circle (  1.96);

\path[draw=drawColor,line width= 0.4pt,line join=round,line cap=round,fill=fillColor] ( 25.68, 98.63) circle (  1.96);

\path[draw=drawColor,line width= 0.4pt,line join=round,line cap=round,fill=fillColor] ( 25.68, 97.78) circle (  1.96);

\path[draw=drawColor,line width= 0.6pt,line join=round] ( 25.68, 95.90) -- ( 25.68, 96.01);

\path[draw=drawColor,line width= 0.6pt,line join=round] ( 25.68, 95.34) -- ( 25.68, 95.24);
\definecolor{fillColor}{gray}{0.95}

\path[draw=drawColor,line width= 0.6pt,fill=fillColor] ( 19.66, 95.90) --
	( 19.66, 95.34) --
	( 31.70, 95.34) --
	( 31.70, 95.90) --
	( 19.66, 95.90) --
	cycle;

\path[draw=drawColor,line width= 1.1pt] ( 19.66, 95.60) -- ( 31.70, 95.60);
\definecolor{fillColor}{gray}{0.20}

\path[draw=drawColor,line width= 0.4pt,line join=round,line cap=round,fill=fillColor] ( 41.73, 26.68) circle (  1.96);

\path[draw=drawColor,line width= 0.6pt,line join=round] ( 41.73, 25.21) -- ( 41.73, 26.51);

\path[draw=drawColor,line width= 0.6pt,line join=round] ( 41.73, 24.26) -- ( 41.73, 23.45);
\definecolor{fillColor}{gray}{0.50}

\path[draw=drawColor,line width= 0.6pt,fill=fillColor] ( 35.71, 25.21) --
	( 35.71, 24.26) --
	( 47.74, 24.26) --
	( 47.74, 25.21) --
	( 35.71, 25.21) --
	cycle;

\path[draw=drawColor,line width= 1.1pt] ( 35.71, 24.88) -- ( 47.74, 24.88);

\path[draw=drawColor,line width= 0.6pt,line join=round] ( 57.77, 82.37) -- ( 57.77, 86.01);

\path[draw=drawColor,line width= 0.6pt,line join=round] ( 57.77, 78.54) -- ( 57.77, 74.35);
\definecolor{fillColor}{gray}{0.95}

\path[draw=drawColor,line width= 0.6pt,fill=fillColor] ( 51.76, 82.37) --
	( 51.76, 78.54) --
	( 63.79, 78.54) --
	( 63.79, 82.37) --
	( 51.76, 82.37) --
	cycle;

\path[draw=drawColor,line width= 1.1pt] ( 51.76, 81.00) -- ( 63.79, 81.00);

\path[draw=drawColor,line width= 0.6pt,line join=round,line cap=round] ( -0.00, 13.72) rectangle ( 67.40,102.67);
\end{scope}
\begin{scope}
\path[clip] (  0.00,102.67) rectangle ( 67.40,119.25);
\definecolor{drawColor}{RGB}{0,0,0}
\definecolor{fillColor}{RGB}{255,255,255}

\path[draw=drawColor,line width= 0.6pt,line join=round,line cap=round,fill=fillColor] ( -0.00,102.67) rectangle ( 67.40,119.25);
\definecolor{drawColor}{gray}{0.10}

\node[text=drawColor,anchor=base,inner sep=0pt, outer sep=0pt, scale=  0.88] at ( 33.70,107.93) {Runtime};
\end{scope}
\begin{scope}
\path[clip] (  0.00,  0.00) rectangle ( 97.56,119.25);
\definecolor{drawColor}{gray}{0.20}

\path[draw=drawColor,line width= 0.6pt,line join=round] (  9.63, 10.97) --
	(  9.63, 13.72);

\path[draw=drawColor,line width= 0.6pt,line join=round] ( 25.68, 10.97) --
	( 25.68, 13.72);

\path[draw=drawColor,line width= 0.6pt,line join=round] ( 41.73, 10.97) --
	( 41.73, 13.72);

\path[draw=drawColor,line width= 0.6pt,line join=round] ( 57.77, 10.97) --
	( 57.77, 13.72);
\end{scope}
\begin{scope}
\path[clip] (  0.00,  0.00) rectangle ( 97.56,119.25);
\definecolor{drawColor}{gray}{0.30}

\node[text=drawColor,anchor=base,inner sep=0pt, outer sep=0pt, scale=  0.88] at (  9.63,  2.71) {F};

\node[text=drawColor,anchor=base,inner sep=0pt, outer sep=0pt, scale=  0.88] at ( 25.68,  2.71) {B};

\node[text=drawColor,anchor=base,inner sep=0pt, outer sep=0pt, scale=  0.88] at ( 41.73,  2.71) {\bfseries S};

\node[text=drawColor,anchor=base,inner sep=0pt, outer sep=0pt, scale=  0.88] at ( 57.77,  2.71) {L};
\end{scope}
\begin{scope}
\path[clip] (  0.00,  0.00) rectangle ( 97.56,119.25);
\definecolor{drawColor}{gray}{0.20}

\path[draw=drawColor,line width= 0.6pt,line join=round] ( 70.15, 20.28) --
	( 67.40, 20.28);

\path[draw=drawColor,line width= 0.6pt,line join=round] ( 70.15, 45.91) --
	( 67.40, 45.91);

\path[draw=drawColor,line width= 0.6pt,line join=round] ( 70.15, 71.54) --
	( 67.40, 71.54);

\path[draw=drawColor,line width= 0.6pt,line join=round] ( 70.15, 97.17) --
	( 67.40, 97.17);
\end{scope}
\begin{scope}
\path[clip] (  0.00,  0.00) rectangle ( 97.56,119.25);
\definecolor{drawColor}{gray}{0.30}

\node[text=drawColor,anchor=base west,inner sep=0pt, outer sep=0pt, scale=  0.88] at ( 72.35, 17.25) {1};

\node[text=drawColor,anchor=base west,inner sep=0pt, outer sep=0pt, scale=  0.88] at ( 72.35, 42.88) {10};

\node[text=drawColor,anchor=base west,inner sep=0pt, outer sep=0pt, scale=  0.88] at ( 72.35, 68.51) {$10^{2}$};

\node[text=drawColor,anchor=base west,inner sep=0pt, outer sep=0pt, scale=  0.88] at ( 72.35, 94.14) {$10^{3}$};
\end{scope}
\begin{scope}
\path[clip] (  0.00,  0.00) rectangle ( 97.56,119.25);
\definecolor{drawColor}{RGB}{0,0,0}

\node[text=drawColor,rotate=-90.00,anchor=base,inner sep=0pt, outer sep=0pt, scale=  1.10] at ( 89.99, 58.20) {Total runtime [min]};
\end{scope}
\end{tikzpicture}};
	\node[anchor=north, inner sep=0pt]
	at ([xshift=67.665pt,yshift=-2pt]mu.south west) {(a)};
	\node[anchor=north, inner sep=0pt]
	at ([xshift=165.675pt,yshift=-2pt]mu.south west) {(b)};
	\node[anchor=north, inner sep=0pt]
	at ([xshift=258.085pt,yshift=-2pt]mu.south west) {(c)};
	\node[anchor=north, inner sep=0pt]
	at ([xshift=33.70pt,yshift=-2pt]y.south west) {(d)};
	\node[anchor=north, inner sep=0pt] (legend)
	at ([yshift=-4pt]current bounding box.south) {%
		\begin{tabular}{@{}l@{\hspace{2em}}l@{}}
			F: \frechetModel & B: \bayesModel \\
			\bfseries S: \sphericalModel & L: \logNNModel
		\end{tabular}};
\end{tikzpicture}}
\caption{Boxplots of the leave-one-station-out prediction errors (on a logarithmic scale) for the four regression methods compared in our application, as well as the runtime of each method, see details in Section~\ref{sec:comparison_of_methods_in_application}. 
Our proposed method is highlighted in dark grey. 
Each of the first three panels corresponds to a metric under which the average leave-one-station-out squared prediction error is computed. 
The last panel gives the total runtime in minutes of each method. 
The time required for density estimation, an additional preprocessing step for the first three approaches, is negligible (less than half a second).
} 
\label{fig:boxplot_errors}
\end{figure}

\section{Discussion}
\label{sec:discussion}

We have introduced a novel method for performing regression where the responses lie on a sphere. Our approach decomposes the spherical non-linearity of the response into two components: a base intercept $\origin \in \sphere$, and a function $f: \Xspace \to T_\origin \sphere$ which belongs to a VVRKHS.

The point $\origin$ plays the role of an intercept in our method. 
It could be a suitable template response. 
If such a response is not available, we suggest setting $\origin$ to a the intrinsic mean of $Y_1,\ldots, Y_n$ in practice.
For theoretical considerations, we can take $\origin$ to be the intrinsic mean of a fixed number of data points: the asymptotic results would then hold conditional on those points, and over the randomness of the rest of the data.
Although the parametrisation of $f_\origin$ depends on the choice of $\origin$, we illustrate in Section~\ref{sec:choice_of_intercept} in some Sobolev VVRKHS settings that the model class is locally invariant with respect to this choice, see Proposition~\ref{proposition:choice_of_origin_is_irrelevant_to_the_model} in the \thesupplement.

Our theoretical results rely on Assumption~\ref{assumption:support_of_Y_for_theory_new}: the bounded support of $Y$, $\sphereDist(\origin, Y) \leq \pi/4$ almost surely. This assumption defines the set $\Uset$ on which empirical risk is minimised.
The bounded support assumption can be substantially relaxed by using numerical analysis. For instance, we can relax the assumption to 
\[
  \sphereDist(\origin, Y) \leq r \quad \mathrm{a.s.} \qquad \text{and} \qquad \Ynorm{ \Log_\origin(\mu(X)) - \Log_\origin(Y) } \leq \rho \quad \mathrm{a.s.},
\]
for pairs of values $(r,\rho)$. The intuition here is that $r$ is the radius of the observations, whereas $\rho$ is the noise level. We show in Section~\ref{sec:results_numerical_proofs} of the \thesupplement\ that for $(r,\rho)=(\pi/2, \pi/2.1)$ and $(r, \rho) = (2\pi/3, \pi/6.6)$ the results of Sections~\ref{sec:rates_minimal_assumptions}, \ref{sec:rates_depending_on_smoothness} and \ref{sec:strong_consistency} carry over after suitably modifying the set $\Uset$ on which $\risk_n$ is optimised. The case $r=\pi/2$ allows the response $Y$ to lie on an entire hemisphere, provided the pull-back noise magnitude, $\Ynorm{ \Log_\origin(\mu(X)) - \Log_\origin(Y) }$, is not larger than $\pi/2.1$. We can increase the radius on which the response lies to $r=2\pi/3$ but then the pull-back noise magnitude needs to be less than $\pi/6.6$. Generally, increasing the support ($r$) requires decreasing the noise ($\rho$). The noise assumption can also be made intrinsic: requiring $\sphereDist(\mu(X), Y) \leq \pi/4$ and using $r = \pi/2$ yields similar results, but the proofs are more intricate because they require verifying that tangent-space pull-backs of balls of the sphere are convex sets, 
which is highly challenging.

Another approach to weakening the bounded support assumption on $Y$ would be to partition $\Xspace$ into disjoint subsets $A_1,\ldots, A_s \subset \Xspace$ covering $\Xspace$, and assume that there exist $\origin_1, \ldots, \origin_s \in \sphere$ such that $\theLaw(\sphereDist(\origin_j, Y)\leq \pi/4  \mid X \in A_j ) = 1$ for $j=1,\ldots, s$. This would then lead to a separate estimate $\hat \mu_j$ on each set $A_j$, and the global estimator would be obtained by combining them, i.e.,
\(
  \hat \mu(x)  = \sum_{j=1}^s \hat \mu_j(x) \indicator_{A_j}(x).
\)

Note that without parametric assumptions on the model, bounded support assumptions are crucially needed. Indeed, the existence and uniqueness of the unconditional intrinsic mean of $Y$, which is a simpler problem, have so far only been shown under bounded support assumptions \citep{afsari2011riemannian, yokotaConvexFunctionsBarycenter2017}.

Our method could be extended to responses on more general and potentially infinite-dimensional manifolds $\Mspace$. There are several technical challenges for such an extension: (1) existence (and uniqueness) of the conditional mean function $\mu$ is not guaranteed \citep[to the best of our knowledge, existing results for general manifolds are only available for finite-dimensional manifolds, see][]{afsari2011riemannian}, and (2) a key technical difficulty for the theoretical analysis of the method would be to obtain local quadratic growth of the population risk, such as the one in Theorem~\ref{theorem:risk_is_convex_new}, which boils down to studying the convexity of the map $v \in T_\origin \Mspace \mapsto \ud^2_\Mspace( \Exp_\origin v, y)$.

\section*{Acknowledgements}

We thank Honey D.\ Alas, Yun Ho, Victor-Emmanuel Brunel, Stephan Huckemann, Victor Panaretos, Xavier Pennec, William Underwood, and Sven Wang for helpful discussions.
AS acknowledges funding from SNSF Grant 200020\_207367.

\bibliographystyle{abbrvnat} 
\bibliography{sphere-bib-shahin,sphere-bib-beatrice}

\ifShowsupplement

\makeatletter
\renewcommand{\thesection}{S\arabic{section}}
\renewcommand{\theAssumption}{S\arabic{Assumption}}
\renewcommand{\theDefinition}{S\arabic{Definition}}
\renewcommand{\thefigure}{S\arabic{figure}}
\renewcommand{\thetable}{S\arabic{table}}
\makeatother
\setcounter{section}{0}

\section{Notation}

All references to Sections or results of this supplement start with the letter "S". References not starting with an "S" pertain to the main paper.

Let $\Yspace$ be a real separable Hilbert space with inner-product $\YinnerProd{\cdot, \cdot}$ and induced norm $\Ynorm{\cdot}$. We denote by $\bounded(\Yspace)$ the space of bounded linear operators on $\Yspace$, by $T^\adjoint$ the adjoint operator of $T \in \bounded(\Yspace)$, and by $\orthogonal(\Yspace) \subset \bounded(\Yspace)$ the group of orthogonal linear operators on $\Yspace$, i.e. operators $R$ such that $R R^\adjoint = R^\adjoint R=\id$. We denote by $\operatornorm{\cdot}$ the operator norm.
Let $\mathbb{N}$ denote the strictly positive natural numbers and $\bar{\Real} = \Real \cup \{\pm \infty\}$ denote the extended real line.
We write $a_n \sim b_n$ if $\lim_{n \to \infty} a_n / b_n = 1$.
For a metric space $(\mathcal{M}, d_\mathcal{M})$, the closed ball with centre $o \in \mathcal{M}$ and radius $r > 0$ is denoted by $\ball_\mathcal{M} (o, r) = \{ x \in \mathcal{M} \mid d_\mathcal{M}(o, x) \leq r\}$.
We write $Z \dsim F$ to mean that $Z$ follows the distribution $F$. 
We write $\Expected^*$ for the outer expectation.
We write $\inftynorm{f}=\sup_{x\in\Xspace}\Ynorm{f(x)}$ for vector-valued $f$. 

\subsection{Differential, gradient and Hessian}

For a function $f: \Mspace \to \Mspace'$ between two manifolds $\Mspace, \Mspace'$, the differential at $x \in \Mspace$ is the function $Df(x): T_x \Mspace \to T_{f(x)} \Mspace'$ defined by 
\[
    Df(x)[u] = \frac{\ud}{\ud t}f(c(t))_{|t=0}, \quad u \in T_x \Mspace,
\]
where $c$ is a smooth curve in $\Mspace$ with $c(0)=x, c'(0)=u$. If $f(x,y)$ is defined on a product of manifolds, we define $D_x f(x,y)[u]$ as the differential with respect to variable $x$, i.e., $D_x f(x,y)[u] = (D f(\cdot, y) )(x)[u]$. 

Let $H$ be a Hilbert space and $f:H \to \Real$ be Fr\'echet differentiable. We denote the gradient of $f$ at $x$ by $\nabla f(x) \in H$. 
If $\nabla f: H \to H$ is Fr\'echet differentiable, we denote by $\Hess_f(x):H \times H \to \Real$ the Hessian of $f$ at $x$. For each $x \in H$, $\Hess_f(x)$ is a bilinear map on $H \times H$.
If $f: H \times H' \to \Real$ is a function $f=f(x,y)$, where $H,H'$ are two Hilbert spaces, we denote by $\nabla_x f (x,y) \in H$ the gradient of the function $x \mapsto f(x,y)$.

\section{Additional discussions}

\begin{Remark}[Comparison of our setting to the literature] \label{remark:comparison_to_literature}\mbox{}
\begin{enumerate}
\item For $\Xspace = [0,1]$, our estimator generalises, as minimiser of \eqref{eq:penalized_empirical_risk}, smoothing splines for curve fitting on a sphere.
 Different motivations have led to different notions of generalised splines on non-linear spaces, including cubic splines for interpolation minimising the Riemannian curvature tensor \citep{Noakes1989CubicSplinesCurvedSpaces}, splines based on generalized B{\'e}zier curves \citep[e.g.,][]{adouani2024regression}, or smoothing splines combined with unrolling or unwrapping techniques \citep{jupp1987fitting}. 
 An advantage of our approach is that it immediately generalises beyond a single scalar covariate. 

 \item Choosing the linear scalar kernel  $k(x,\check x) = x^\transpose \check x$, our model corresponds exactly to geodesic regression \citep[e.g.,][]{thomasfletcherGeodesicRegressionTheory2013} for $\Xspace = \Real$, and translates to a penalised GLM-type intrinsic regression \citep[e.g.,][]{corneaRegressionModelsRiemannian2017} for $\Xspace = \Real^q$, with fixed intercept $\origin$. 
 However, our approach also extends to infinite dimensions, with $\sphere$ being an infinite-dimensional (Hilbert) sphere, or $\Xspace$ being a Hilbert space. 
 If $\sphere$ is the sphere in $L^2([0,1])$ and $\Xspace = L^2([0,1])$, choosing $k(x,\check{x}) = \int x(t) \check{x}(t)\, \ud t$ induces the usual class of linear Hilbert-Schmidt operators as in standard VVRKHS formulations \citep{grunewalder2013smooth}, while the spherical constraint on the response yields spherical function-on-function regression $\mu(x): s \mapsto [\Exp_\origin(\int \beta(\cdot,t) x(t)\, \ud t) ](s)$. 
 For non-linear effects of a functional covariate $X$ in a Hilbert space $\Xspace$, an additive kernel of the form $k_m(x, \check x)=\sum_{j=1}^m k_j(\langle x, e_j\rangle, \langle e_j, \check{x}\rangle)$ employs kernels $k_j$ on the $m$ scores, say, with respect to the first principal components $e_1,\dots,e_m\in\Xspace$ of $X$, in the spirit of \citet{mullerFunctionalAdditiveModels2008,zhuStructuredFunctionalAdditive2014}.
 \item Constructions of additive and interaction smoothing splines with scalar kernel $k$ \citep[e.g.,][Chapter~10.2]{wahbaSplineModelsObservational1990} directly carry over to the corresponding vector-valued kernel. Hence, our framework also covers generalisation from GLM-type intrinsic regression to Generalised Additive Model-type regression on the sphere.
 \citet{stockerFunctionalAdditiveModels2023} discuss a semi-parametric approach to such models for form and shape spaces.
\item As mentioned in Remark~\ref{remark:types_data_sphere}, our model encompasses regression with densities as responses, under the Fisher--Rao metric \citep{srivastava2007densities}. Other geometries and methods have been proposed for such regression. 
These include Bayes Hilbert space approaches \citep{van2014bayes,menafoglio2014kriging,maier2025additive}, transformation-based methods \citep{petersenFunctionalDataAnalysis2016,hanAdditiveFunctionalRegression2020}, and Fréchet regression in Wasserstein spaces \citep{petersenFrechetRegressionRandom2019,chen2023sliced}---see also Section~\ref{sec:application}.

\item The global nonlinear Fréchet regression of \citet{bhattacharjeeNonlinearGlobalFrechet2025} also combines an RKHS with the intrinsic squared distance of the response space. Their construction nevertheless differs from ours. They employ the RKHS structure on the \emph{predictor} side: the conditional expectation of $\ud^2(Y,y)$ given $X$ is replaced by the projection of $\ud^2(Y,y)$ onto an RKHS of functions of $X$, and the estimator is a weighted Fréchet mean $\arg\min_y \sum_i w_{in}(x) \ud^2(Y_i,y)$ with kernel weights, computed separately at each $x$. We instead model $\mu = \Exp_\origin \circ f_\origin$ with $f_\origin$ in a \emph{vector-valued} RKHS of $T_\origin\sphere$-valued functions, and fit a single function $\hat f_n$ by penalised minimisation of the intrinsic loss. Their theoretical guarantees do not cover our setting because existence and uniqueness of their target are established for response metrics admitting an isometric embedding into a Hilbert space with convex image, or for spaces of nonpositive curvature, neither of which holds for the sphere with the geodesic distance, and their rates further require a finite entropy integral of the response space, which fails for the sphere of an infinite-dimensional Hilbert space.
\end{enumerate}
\end{Remark}

\begin{Remark} \label{remark:detailed_low_rates} \mbox{}

    \begin{enumerate}
        \item Typical rates for kernel ridge regression depend on the smoothness of the RKHS space, which is quantified by the rates of decay of the eigenvalues associated with the integral operator induced by the kernel (defined in \eqref{eq:KIntegralOperator_defn}), and source conditions. In the worst case (roughest possible functions and no source conditions) the same $n^{-1/4}$ as in Theorem~\ref{theorem:rates_for_f_weakest_assumptions_corrected} are obtained, see \citet{caponnettoOptimalRatesRegularized2007}. The difference is that our problem is highly non-linear in the response space (the sphere $\sphere$), whereas in kernel ridge regression the response space is linear.
        \item Our result holds for the sphere $\sphere$ of arbitrary dimension, in particular the infinite-dimensional sphere, and the rates do not depend on the dimension of the sphere. 
\item Rates for Fréchet regression \citep{petersenFrechetRegressionRandom2019,bhattacharjeeNonlinearGlobalFrechet2025} are pointwise (or local) in $x$ and rely on pointwise (or local) assumptions, such as quadratic growth of the population objective at each $x$; they are therefore not directly comparable with our single global rates for the entire estimator $\hat f_n$.
    \end{enumerate}
\end{Remark}

\begin{Remark}[Detailed comparison with existing rates]
    \label{remark:detailed_smoothness_rates}
The rates obtained in Theorem~\ref{theorem:rates_for_f_depending_on_smoothness_new} can be compared to the closest similar regression settings for which theory is available: functional regression with RKHS, linear vector-valued kernel ridge regression, and manifold regression.
    \begin{enumerate}
        \item In functional regression with RKHS techniques, the problem is to estimate $\beta$ in the model $Y = \int \beta(t) X(t)dt + \vep$ by assuming $\beta$ belongs to an RKHS. In this problem, for instance under the assumption that $\KIntegralOperator$ and the covariance operator of $X$ are simultaneously diagonalisable, and under polynomial decay of their eigenvalues, \citet{yuanReproducingKernelHilbert2010} obtain a minimax optimal rate similar to our rate.

            \item For linear vector-valued kernel ridge regression, the rates depend on two quantities: (1) source conditions and (2) rates of decay of the eigenvalues of $\KIntegralOperator$. 
Source conditions encode how much smoother $f_\origin$ is compared to typical functions in $\Hspace$. This is crystallised into the assumption that there is a $c \in [1,2]$ such that for all $y \in T_\origin \sphere$, the function $\YinnerProd{f_\origin(\cdot), y} \in \image( \embeddingCov^{(c-1)/2})$, where $\embeddingCov = \Expected( k(X, \cdot) \otimes k(X, \cdot)) \in \bounded(\Hspace_k)$, and has the same non-zero eigenvalues as $\KIntegralOperator$. The case $c=1$ corresponds to no additional source condition beyond $f_\origin \in \Hspace$. In the linear setup, the exact form of $\hat f_n$ is known, and such source conditions allow us to control the bias term in the overall error. In our setting the closed form expression of $\hat f_n$ is unknown and such source conditions are not directly useful. We therefore compare our rates to the linear rates with $c=1$.
\citet{caponnettoOptimalRatesRegularized2007} provide rates for the least squares problem $\arg \inf_{f \in \Hspace} \Expected \Ynorm{Y - f(X)}^2$ under the assumption that $\trace K(x,x) < \infty$  (which implies in our SMO kernel setting that $\dim \Yspace < \infty$). The rates they obtain with no source conditions (case $c=1$) match our finite-rank rates exactly, and our polynomial decay rates up to a $\log(n)$ factor. 
We cannot compare to the rates of \citet{liOptimalSobolevNorm2024}  since they consider the misspecified regression problem.
Our rates for the stretched exponential decay also match the kernel ridge regression rates, see, e.g., \citet{bakEffectDimensionalityConvergence2025}.

\item We can also compare our rates with the manifold regression literature. \citet{corneaRegressionModelsRiemannian2017} considers finite-dimensional manifolds as responses and vector covariates. They operate under a parametric setting, and obtain the usual parametric $\sqrt{n}$ rates for the parameter estimators. The equivalent in our setting would be a finite-rank kernel $k$, such as the polynomial kernel $k(x,\check x) = (1+ x^\transpose \check x)^p$, $p \in \mathbb N$ and $\Xspace = \Real^s$, which would yield the same rates. \citet{choiHighdimensionalHilbertSchmidt2025} consider regression between manifolds by mapping the manifold points onto tangent spaces and then performing linear regression, which in the general setting boils down to estimating a Hilbert--Schmidt operator between the tangent spaces. In the case of finite-dimensional response and covariate, they obtain the parametric $\sqrt{n}$ rate. Their linear setting can be compared to an inner-product kernel $k(x,\check x) = x^\transpose \check x$ under a bounded domain in our setting, for which we obtain the same rate.

\end{enumerate}
\end{Remark}

\section{Detailed computational aspects}
\label{sec:detailed_computational_aspects}

Unlike kernel regression with Euclidean responses, the minimiser of our
empirical risk does not admit a closed-form expression, due to the spherical loss.
Consequently, we rely on gradient-based
optimisation of the empirical loss function~\eqref{eq:penalized_empirical_risk}.

By Theorem~\ref{theorem:existence_of_global_minimizer_unconstrained}, and under
the assumption that a unique minimiser of $\risk_n(\cdot,\lambda_n)$ exists, we
may search for the optimal $f$ by optimising over coefficients
$\xi_1,\ldots,\xi_n \in \vspan(\Log_\origin(Y_1),\ldots,\Log_\origin(Y_n))$.
Letting $\boldsymbol{\xi} = (\xi_1,\ldots,\xi_n)$, we define the reparametrised
empirical risk $R_n(\boldsymbol{\xi}) := \risk_n(f_{\boldsymbol{\xi}},\lambda_n)$,
where $f_{\boldsymbol{\xi}}(\cdot) = \sum_{i=1}^n k(X_i,\cdot)\,\xi_i$. At each
iteration, we compute the Euclidean gradient of a smooth ambient extension
$\overline{R}_n : \Yspace^n \rightarrow \Real$ of $R_n$ with respect to
$\boldsymbol{\xi}$, using the closed-form expression derived in
the \thesupplement, Section~\ref{subsec:grad_computation}, and project this ambient gradient
orthogonally onto the tangent space at $\origin$, yielding the Riemannian
gradient within $T_\origin\sphere$.

In our implementation, we fix $\origin$ to the Fr\'{e}chet mean of
$Y_1,\dots,Y_n$ (see Section~\ref{sec:discussion}). We initialise
$\boldsymbol{\xi}$ at tangent-space regression estimates,
obtained under the same model specifications (kernel $k$ and penalty $\lambda_n$)
but replacing the spherical distances in $R_n(\boldsymbol{\xi})$ by the linear
distances $\|f_{\boldsymbol{\xi}}(X_i) - \Log_{\origin}(Y_i)\|_{\Yspace}$, which
yields the familiar closed-form solution \citep{kadriOperatorvaluedKernelsLearning2016}.

Standard gradient descent from this initialisation can be of limited practical
use: in our numerical experiments (not shown here) it made negligible progress beyond the initial
tangent-space fit. A further challenge is
the dimensionality of $\boldsymbol{\xi}$: for a $D$-dimensional sphere, the
coefficient vector is of dimension $\min(D,n) \times n$, which becomes
computationally prohibitive and numerically ill-conditioned as $n$ grows. We
address both issues simultaneously by combining a low-rank approximation with
BFGS optimisation \citep[][]{broyden1970convergence, fletcher1970new, goldfarb1970family, shanno1970conditioning}. BFGS retains the same
gradient but builds an approximation of the inverse Hessian, accelerating
convergence without requiring explicit Hessian computation or inversion. We use the BFGS implementation of \texttt{stats::optim} with its default convergence criterion, which stops when the relative decrease of the objective falls below $\sqrt{\varepsilon_{\mathrm{machine}}} \approx 1.5 \cdot 10^{-8}$, with at most $100$ iterations.

More precisely, we use the leading $m_k$ eigenvectors of the kernel Gram matrix
$\mathbf{K} = [k(X_i, X_j)]_{i,j=1}^{n}$ to obtain a basis
$g_1, \dots, g_{m_k}$ for an optimally approximating $m_k$-dimensional subspace
of $\vspan\{k_{X_1}, \dots, k_{X_n}\}$ \citep[for details, see][]{wood2003lowrank}, where $k_x := k(\cdot, x)$.
Analogously, we use the leading $m_y$ eigenvectors of the response Gram matrix
$\mathbf{G} = [\langle Y_i, Y_j \rangle_{\Yspace}]_{i,j=1}^{n}$ to obtain a
low-dimensional representation $\mathbf{z}_i \in \Real^{m_y}$ for each $Y_i$,
such that $[\mathbf{z}_i^\top \mathbf{z}_j]_{i,j=1}^{n} \approx \mathbf{G}$.
Together, these reductions lower the total number of parameters 
down to $m_k \times m_y$ (instead of $O(n^2)$). 
The two dimensions are chosen by a relative eigenvalue-mass criterion: $m_k$ is the smallest number of leading eigenvalues of $\mathbf K$ whose cumulative share of the eigenvalue mass exceeds $1-\tau$, and analogously $m_y$ for $\mathbf G$, with $\tau = 10^{-4}$ throughout. In practice, we observed that even a near-lossless low-rank approximation typically requires far fewer parameters, while simultaneously improving numerical conditioning and empirical fit. The complete procedure, including the selection of $\lambda_n$ and of the kernel bandwidth, is summarised in Section~\ref{subsec:implementation} of the \thesupplement. An implementation is provided in the \texttt{R} package \texttt{sphereg}.

\section{Details on the simulation study}
\label{sec:details_simulation_study}

\subsection{Simulation setup} \label{subsec:simulation}

The core of the data generating process involves two components: a conditional mean function $\mu: [0,1]^2 \rightarrow \sphere$ and a perturbation process $\varepsilon$, which is generated in the tangent space $T_\origin\sphere$.
The random response variable $Y$ is then generated conditional on the covariate $X \dsim \operatorname{Unif}([0,1]^2)$, by adding a perturbation process through the Riemannian exponential map. Given $X=x$,
\begin{equation}
  \label{eq:sim:model}
  Y  = \Exp_{\mu(x)} \PT_{\origin\rightarrow\mu(x)} \varepsilon,
\end{equation}
where $\origin$ is the Fréchet mean of $\mu(X)$,  $\PT_{\origin \rightarrow \mu(x)} :T_\origin \sphere \to T_{\mu(x)} \sphere$ denotes the parallel transport along the geodesic connecting $\origin$ and $\mu(x)$, and $\varepsilon$ is a perturbation/noise term. The use of parallel transport is needed because the perturbation process $\varepsilon$ is generated on $T_\origin \sphere$ and not on $T_{\mu(x)} \sphere$.
In the following, we describe how we determine the conditional mean function $\mu$ and the perturbation term $\epsilon$.

\subsubsection*{Defining the Conditional Mean Function ($\mu$)}

The construction of the conditional mean function $\mu(\cdot)$ is designed to mimic a 2-mixture density with modes and weights depending on the covariate $x \in [0,1]^2$. Its construction is given in the \thesupplement, Section~\ref{sec:simulation_mu}. A plot of $\mu$ is given in Figure~\ref{fig:simulation:illustration} (a). Note that with this construction, $\sup_{x \in [0,1]^2} \sphereDist(\origin, \mu(x)) = 4\pi/5$, and the setting of the numerical simulations does not satisfy our consistency assumption.

\subsubsection*{Constructing the Perturbation Process ($\varepsilon$)}

To model perturbations on the Riemannian manifold, we utilise a smooth perturbation process defined via a Fourier basis expansion in the template tangent space $T_\origin \sphere$. The perturbation is then transported to the tangent space at the conditional mean $\mu(x)$, 
where it is subsequently mapped via the exponential map to produce $Y$, as described in \eqref{eq:sim:model}. 

Consider the basis functions $e_j(t) = \sqrt{2} \cos(j \pi t),\ j=1, \dots, m:=20$ on the interval $[0,1]$.
For each data point $i = 1, \dots , n$, we sample independent coefficients $\vartheta_{ij} \dsim N(0, j^{-4})$ for $e_j$, with total variance $\tilde{\tau}^2 = \sum_{j=1}^{m} j^{-4}$. 
This choice, corresponding to a quadratic penalty on the second-order derivatives, ensures that the resulting process is smooth and possesses a covariance operator whose eigenvalues decay polynomially. Since the $e_j$ are orthogonal to the constant $q: t \mapsto 1$, they are elements of $T_{q} \sphere$, and we define
\[
  {\varepsilon}_i= \PT_{q \rightarrow \origin} \Big(\sum_{j=1}^m \vartheta_{ij} e_j\Big),
\]
parallel transporting them to $\origin$ and rescaling the residual variance to the desired $\tau_\varepsilon^2=\Expected [\|\varepsilon_{i}\|^2]$. 
Responses $Y$ given $x_i$ are then obtained using \eqref{eq:sim:model}.
While parallel transport preserves the variance, wrapping for $\|\varepsilon_{i}\|\geq\pi$ leads to a slightly reduced intrinsic variance  $\tau_Y^2=\Expected [\sphereDist^2(Y, \mu(X)) \mid X ]$, with Monte-Carlo ratio estimate $\tau_Y^2/\tau_\varepsilon^2 \approx 0.97$ ($\geq 0.968$ one-sided 99\% confidence interval) for $\tau_\varepsilon = 4\pi/10$, increasing for higher concentrations, with $\tau_\varepsilon = 2\pi/10$ and $\tau_\varepsilon = \pi/10$, to more than $\tau_Y^2/\tau_\varepsilon^2>0.99$.

\subsubsection*{Tangent-space regression}

We compare our results against tangent-space regression implemented using a VVRKHS framework. This approach involves projecting the observed data onto the tangent space at a point $\origin$ via $Z_i := \Log_\origin Y_i$. Subsequently, we perform a VVRKHS regression \citep{kadriOperatorvaluedKernelsLearning2016} of these projected variables ($Z_i$) as a function of the covariate $X_i$, utilising a SMO kernel with Gaussian radial basis function (GRBF) scalar kernel. The output is a fitted tangent vector function $\hat g: \Xspace \to T_\origin \sphere$. Finally, we reconstruct the conditional mean estimator on the sphere using the exponential map, $\hat{\mu}(x) := \Exp_\origin(\hat g(x))$.

\subsubsection*{Simulation parameters}

The simulation study was designed to evaluate the model's performance under various sample sizes ($n$) and observational noise ($\tau_\epsilon$). We performed $B=100$ repetitions for each parameter combination.
The simulations were structured across four different sample sizes $n \in \{20, 50, 100, 200\}$ and three distinct residual variance levels $\tau_\epsilon \in \{ \pi/10, 2\pi/10, 4\pi/10 \}$. These $\tau_\epsilon$ values allow us to investigate the model's performance across varying noise levels, ranging from low (corresponding to low data diameters) to high.
Figure~\ref{fig:simulation:illustration} (b, c) show the effect of increasing $\tau_\varepsilon$ on the generated responses $Y| X=x$, the estimated $\hat \mu(x)$, as well as the spherical distance between $Y|X=x$ and $\mu(x)$.
Our model and the tangent-space regression are fitted with GRBF  kernel. 
For each simulation repetition, the point $\origin$ in our model and in tangent-space regression is taken to be the Fréchet mean of $y_1, \dots, y_n$.
The regularisation parameter $\lambda$ is optimised via 5-fold cross-validation (CV).
The kernel bandwidth $\sigma$ is held fixed to a constant value throughout the simulation in order to ensure a consistent kernel across all replications, thereby enabling meaningful comparisons; varying $\sigma$ would effectively change the kernel itself and confound the assessment. 

After running the simulations, we obtain $B=100$ estimators of the conditional mean function, $\{\hat{\mu}_j\}_{j=1}^B$, for each regression method.

\subsubsection*{Measure of performance}

For each setting of $n$ and $\tau_\epsilon$, we measure the performance of the estimators using the mean squared error (MSE), defined as
\begin{equation*}
    \widehat{\operatorname{MSE}}= 
    \frac{1}{B} \sum_{j=1}^B
      \Expected_X[\sphereDist^2\!\left(\hat \mu_j(X), \mu(X)\right)],
\end{equation*}  
where the expectation is computed on a grid. 

\subsection{Simulation results}

The results are shown in Figure~\ref{fig:sim_performance}, and can be reproduced using the code available in the \thesupplement.  
When compared with tangent-space regression, our method demonstrates superior performance; however, in the low sample size case ($n=20$), it is on par with tangent-space regression. This improvement becomes more significant as the sample size increases.

\section{Further details regarding the density regression application}
\label{sec:application_further_details}

\subsection{Details for \frechetModel}

\frechetModel\ is implemented in \texttt{R} in the
package \texttt{frechet}. We employ the local version of Fréchet regression, as a natural competitor for our nonparametric estimator, utilising an anisotropic Gaussian kernel with two bandwidth parameters (one for each covariate). Here, the distributions are fitted with respect to the 2-Wasserstein metric, as described by \cite{petersenFrechetRegressionRandom2019} and implemented in the function \texttt{frechet::LocDenReg}.
 
\subsection{Details for \bayesModel}

Following \citet{maier2025additive}, we use component-wise $L^2$-Boosting for \bayesModel\ which is implemented in the \texttt{R} package \texttt{FDboost}.
Densities in a Bayes Hilbert space are required to be almost-everywhere positive, which is not necessarily the case for our density estimates $h_j$. 
We thus slightly shift them away from zero by adding a small constant $c_j>0$, a common approach for zero counts in compositional data, working with 
$\tilde{h}_j = (h_j + c_j)[1+c_j(z_{\max} - z_{\min})]^{-1}$ instead of $h_j$. 
Here, we use a fixed small constant $c_j=10^{-5}$ which we found to produce good reliable results, better than $n_j$-dependent shifts inspired by  \citet{martin2015bayesian}, such as $c_j=n_j^{-1/2}$, $c_j=n_j^{-1}$ or $c_j=n_j^{-1}(z_{\text{max}}-z_{\text{min}})^{-1}$.  
In practice, computations are carried out using the isometric centred log-ratio transforms $\operatorname{clr} \tilde{h}_j = \log \tilde{h}_j - \frac{1}{z_{\text{max}}-z_{\text{min}}} \int_{z_{\text{min}}}^{z_{\text{max}}} \log \tilde{h}_j(z)\, dz$ which maps the densities 
into $L^2(\mathcal Z)$. We fit the $\operatorname{clr} \tilde{h}_j$ with tensor-product B-splines, which are decomposed into marginal effects and interaction effects, as suggested, e.g., by \citet{brockhaus2020boosting} for functional boosting models.

\subsection{Details for \logNNModel}

In the \logNNModel, the AAC measurements are modelled as $Z = \tilde{Z} + \epsilon$ with $\log\tilde{Z} \mid X = x_i \dsim N(\theta_1(x_i), \theta_2^2(x_i))$ and independent noise $\epsilon \dsim N(0, \theta_3^2)$.  
The distribution parameters $\theta_1(x) = g_1(x)$ and $\theta_2(x) = \exp(g_2(x))$ are estimated using low-rank Gaussian processes on $g_1$ and $g_2$ with Gaussian kernel \citep{woodGeneralizedAdditiveModels2017a}, with location-independent measurement noise variance $\theta_3^2$, in a Bayesian framework for distributional regression \citep{umlauf2018bamlss}. To mitigate the substantially longer computation times of this model compared to the other methods, we fit it on data subsampled to $\tilde{n}_j = 100$ measurements per location. Increasing the per-location sample size to $\tilde{n}_j = 1000$ did not improve performance.

\subsection{Details on the cross validation}

To evaluate out-of-sample performance, we perform a leave-one-station-out procedure.
For each station $j=1,\dots,14$, we remove its density from the training set, fit each regression model to the remaining $13$ stations, and predict the density at the held-out location. 

For \bayesModel\ and \sphericalModel, the hyperparameters are selected by {leave-one-station-out} cross-validation, minimising the respective intrinsic prediction error over a grid of bandwidth values. For \frechetModel,  cross-validation for the bandwidth selection is embedded in the R package \texttt{frechet}. 
Due to the long computation times of \logNNModel, we do not perform an inner cross-validation loop here. Instead, we select the range parameter directly on the outer cross-validation error from a grid of values, which gives this method a slight advantage over the others.

Recall that the empirical density at station $j$ is $h_j$ and let us denote by $\hat h^{(m)}_j$ the corresponding leave-one-station-out prediction density for each method $m$ considered. Note that if the method $m$ predicts a point on the sphere $\hat \mu^{(-j)}(x_j)$, then the corresponding predicted density is $\hat h_j^{(m)}(z) = (\hat \mu^{(-j)}(x_j)(z))^2$.
We quantify discrepancies under three metrics:
\begin{enumerate}
  \item The squared spherical distance between the square-root densities 
    $\sqrt{h_j}$ and $ \sqrt{\hat h_j^{(m)}}$.
  \item The squared Bayes-space distance
    \(
    \ud_{B^2}^2(h_j,\hat h^{(m)}_j)
    \),
    defined as the squared $L^2$ norm of the difference between centred log-ratio transforms,
    \[
      \ud_{B^2}^2(h_j,\hat h_j^{(m)})
      \;=\;
      \int_{\mathcal{Z}} 
      \bigl(\operatorname{clr}(h_j)(z) - \operatorname{clr}(\hat h_j^{(m)})(z)\bigr)^2
      \,\mathrm{d}z.
    \]

  \item The squared $2$-Wasserstein distance
    \(
    W_{2}^2(h_j,\hat h^{(m)}_j)
    \)
    based on the $L^2$ distance between quantile functions $Q_j$ and $\hat Q_j^{(m)}$ of $h_j$ and $\hat h_j^{(m)}$,
    \[
      W_2^2(h_j,\hat h_j^{(m)})
      \;=\;
      \int_0^1 \bigl(Q_j(u) - \hat Q_j^{(m)}(u)\bigr)^2 \, \mathrm{d}u.
    \]
\end{enumerate}

\section{Two technical results of independent interest}
\label{sec:technical_results_of_independent_interest}

The following results draw inspiration from \citet{pillonettoSolutionsNonlinearControl2008}. Recall that $\ball_\Hspace(0, r) \subset \Hspace$ is the closed ball of radius $r$ with centre $0$. 
Recall that the space $C(\Xspace, \Yspace)$  of bounded continuous functions $g: \Xspace \to \Yspace$, equipped with the norm
\[
    \inftynorm{g} = \sup_{x \in \Xspace} \Ynorm{g(x)},
\]
is a Banach space. Let $\operatornorm{\cdot}$ denote the operator norm of operators.
The two following results hold for general vector-valued operator kernels.
\begin{Lemma} \label{Lemma:Bclosedsupnorm}
    Let $\Xspace$ be a metric space, $\Yspace$ be a finite-dimensional Hilbert space and $\Hspace$ be a Vector-Valued RKHS of functions  $f :\Xspace \to \Yspace$, with continuous operator-kernel $K:\Xspace \times \Xspace \rightarrow \bounded(\Yspace)$ satisfying $\sup_{x \in \Xspace} \operatornorm{K(x,x)} < \infty$. Then, for every $r > 0$, $B := \ball_\Hspace(0, r)$ is closed in $C(\Xspace, \Yspace)$.
\end{Lemma}

\begin{Lemma} \label{Lemma:ball_in_VVRKHS_compact_in_sup_norm}
    Let $\Xspace$ be a compact metric space with metric $d_X$, $\Yspace$ be a finite-dimensional Hilbert space and $\Hspace$ a Vector-Valued RKHS of continuous functions from $\Xspace$ to $\Yspace$, with continuous and bounded kernel $K:\Xspace \times \Xspace \rightarrow \bounded(\Yspace)$. Then, $B := \ball_\Hspace(0, r)$, $r >0$, is a compact subset of $C(\Xspace, \Yspace)$.
\end{Lemma}

\section{Technical results}
\label{sec:technical}

\subsection{Empirical process theory}

The following is an extension of \citet[][corollary 1]{maurerVectorContractionInequalityRademacher2016}.
\begin{Proposition}
    \label{proposition:vector-contraction-with-absolute-value}
Let \(\mathcal{X}\) be any set, \((x_{1},\dots,x_{n}) \in \mathcal{X}^{n}\), let \(\mathcal{F}\) be a class of functions
\(f : \mathcal{X} \to \Yspace \), where $\Yspace$ is a separable Hilbert space. Let \(h_{i} : \Yspace \to \Real\) be Lipschitz functions with Lipschitz constant bounded by \(L < \infty \), for $i=1,\ldots,n$.
Assume there is a $\tilde f \in \mathcal F$ such that $h_i(\tilde f(x_i))=0$ for all $i=1,\ldots,n$. 
Then
\[
\Expected \sup_{f \in \mathcal{F}} \left| \sum_{i} \varepsilon'_{i} h_{i}\bigl(f(x_{i})\bigr) \right|
    \le  2 \sqrt{2} L \, \Expected \sup_{f \in \mathcal{F}} \sum_{i=1}^n \sum_{j \in J} \varepsilon_{ij} \YinnerProd{f(x_{i}), e_j},
\]
where $\{\varepsilon'_i\}$ are independent Rademacher variables, \(\{ \varepsilon_{ij} \}\) is an independent doubly indexed Rademacher sequence, and $(e_j)_{j \in J}$ is an orthonormal basis of $\Yspace$.
\begin{proof}
    Define
    \[
        Z(f) := \sum_{i=1}^n \varepsilon_i h_i\bigl(f(x_i)\bigr).
    \]
    Then pointwise (for any realisation of the $\varepsilon_i$),
    \[
        \sup_{f \in \mathcal{F}} |Z(f)| 
        = \max\Bigl\{ \sup_{f \in \mathcal{F}} Z(f),\ \sup_{f \in \mathcal{F}} \bigl(-Z(f)\bigr) \Bigr\}.
    \]
    By assumption each term in the max is positive or zero, and $\max(A,B) \le A + B$ for $A,B \geq 0$ yields
    \[
        \Expected \sup_{f \in \mathcal{F}} |Z(f)|
        \le \Expected \sup_{f \in \mathcal{F}} Z(f) 
        + \Expected \sup_{f \in \mathcal{F}} \bigl(-Z(f)\bigr).
    \]
    By symmetry of the Rademacher variables, $(\varepsilon_i) \stackrel{d}{=} (-\varepsilon_i)$, so
    \[
        \Expected\sup_{f \in \mathcal{F}} \bigl(-Z(f)\bigr)
        = \Expected \sup_{f \in \mathcal{F}} \sum_i (-\varepsilon_i) h_i\bigl(f(x_i)\bigr)
        = \Expected \sup_{f \in \mathcal{F}} \sum_i \varepsilon_i h_i\bigl(f(x_i)\bigr)
        = \Expected \sup_{f \in \mathcal{F}} Z(f).
    \]
    Therefore, by \citet[][corollary 1]{maurerVectorContractionInequalityRademacher2016},
    \begin{align*}
        \Expected \sup_{f \in \mathcal{F}} \left| \sum_i \varepsilon_i h_i\bigl(f(x_i)\bigr) \right|
 & \le 2 \, \Expected \sup_{f \in \mathcal{F}} \sum_i \varepsilon_i h_i\bigl(f(x_i)\bigr).
 \\ &
 \le 2 \sqrt{2} L \, \Expected \sup_{f \in \mathcal{F}} \sum_{i,j} \varepsilon_{ij} \YinnerProd{f(x_{i}), e_j},
    \end{align*}
\end{proof}
\end{Proposition}

\begin{Definition} \label{definition:covering}
Let $(\mathcal M,d)$ be a metric space and let $\mathcal F \subseteq \mathcal M$. For $\varepsilon > 0$, the \emph{covering number} of $\mathcal F$ at scale $\varepsilon$ is
\[
\covering(\varepsilon, \mathcal F, d)
:= \min\Bigl\{\, N \in \mathbb{N} : \exists m_1,\dots,m_N \in \mathcal M \text{ such that }
\mathcal F \subseteq \bigcup_{i=1}^N B_d(m_i,\varepsilon) \Bigr\},
\]
where $B_d(m,\varepsilon) := \{ m' \in \mathcal M : d(m,m') \le \varepsilon \}$.
\end{Definition}

Recall that if $(\Xspace,\mathcal{A},\theLaw)$ is a probability space, then
$ L^p(\theLaw)$ with $p \geq 1$ is the set of (equivalence classes of) measurable functions $f: \Xspace \to \Real$ such that $\norm{f}_p < \infty$, where
\begin{equation} \label{eq:Lp_norm}
    \norm{f}_p := (\int |f|^p  d \theLaw)^{1/p}.
\end{equation}

For a function $f: \Xspace \to \Real$ where $\Xspace$ is a set, define
\begin{equation} \label{eq:sup_norm_over_set}
    \norm{f}_\Xspace = \sup_{x \in \Xspace} |f(x)|
\end{equation}

\begin{Definition}[Bracketing number in $L^p$]
    \label{definition:bracketing}
    Let $\Xspace$ be a measurable space, let $\theLaw$ be a probability measure on $\Xspace$, and let $p \ge 1$. 
    Further, let $\mathcal{F} \subseteq L^p(\theLaw)$.
    \begin{itemize}
        \item For $\varepsilon > 0$, the \emph{bracketing number} of $\mathcal{F}$ at scale $\varepsilon$ with respect to $L^p(\theLaw)$ is
            \begin{align*}
                \bracketing\bigl(\varepsilon, \mathcal{F}, L^p(\theLaw)\bigr) &:= \min\biggl\{ N \in \mathbb{N} : \exists\, l_1,u_1,\dots,l_N,u_N \in L^p(\theLaw) \text{ such that}\\
                &\quad l_i \le u_i \ \theLaw\text{-a.e.}, \quad \norm{u_i - l_i}_p \le \varepsilon \quad (i=1,\dots,N), \\
                &\quad \text{and } \mathcal{F} \subseteq \bigcup_{i=1}^N [l_i, u_i] \biggr\},
            \end{align*}
            where $[l,u] := \{ f \in L^p(\theLaw) : l \le f \le u \ \theLaw\text{-a.e.} \}$.
    \end{itemize}
\end{Definition}

The following result is essentially \citet[][Theorem~2.7.17]{vandervaartWeakConvergenceEmpirical2023a}.
\begin{Lemma}
    \label{Lemma:bracketing_and_covering_for_Lipschitz}
    \mbox{}

    Let $(\mathcal{M},d)$ be a metric space, $(\mathcal{Z},\mathcal{A})$ a measurable space, 
    and assume that $\theLaw$ is a probability measure  on $(\mathcal{Z},\mathcal{A})$.
    Let $\psi : \mathcal{M} \times \mathcal{Z} \to \Real$ be such that $\psi(m,\cdot)$ 
    is $\mathcal{A}$-measurable for every $m \in \mathcal{M}$. 
    Assume the following Lipschitz condition holds:  there exists $\alpha > 0$ such that
    \[
        \norm{\psi(m,\cdot) - \psi(m',\cdot)}_{\mathcal{Z}}
        \;\le\; \alpha \, d(m,m') \qquad \text{for all } m,m' \in \mathcal{M}.
    \]
Define the class
    \[
        \mathcal{F} \;:=\; \bigl\{ f_m : m \in \mathcal{M} \bigr\}, 
        \qquad f_m(z) := \psi(m,z),
    \]
Assume furthermore that $\mathcal{F} \subseteq L^p(\theLaw)$.

    Then, for every $\varepsilon > 0$, and every $p \geq 1$,
    \[
        \bracketing\bigl(\varepsilon, \mathcal{F}, L^p(\theLaw)\bigr)
        \;\le\;
        \covering \left(\frac{\varepsilon}{2\alpha}, \mathcal{M}, d \right),
    \]
    where $\covering$ and $\bracketing$ are defined in Definitions~\ref{definition:covering} and~\ref{definition:bracketing}.

     \begin{proof}
         Fix $\varepsilon > 0$. 
        Assume without loss of generality that $N = \covering \left(\frac{\varepsilon}{2\alpha}, \mathcal{M}, d \right) < \infty$, and let 
         \[
             \{m_1,\dots,m_N\} \subset \mathcal{M}
         \]
         be a $\frac{\varepsilon}{2\alpha}$-cover of $\mathcal{M}$ with respect to $d$, 
         so that for every $m \in \mathcal{M}$ there exists $j \in \{1,\dots,N\}$ with
         \[
             d(m,m_j) \;\le\; \frac{\varepsilon}{2\alpha},
         \]
         and $N = \covering \left(\frac{\varepsilon}{2\alpha}, \mathcal{M}, d \right)$.

         By the Lipschitz condition, for all $m \in \mathcal{M}$, $j \in \{1,\dots,N\}$ and 
         all $z \in \mathcal{Z}$,
         \[
             d(m,m_j) \le \frac{\varepsilon}{2\alpha}
             \;\Longrightarrow\;
             \bigl|\psi(m,z) - \psi(m_j,z)\bigr|
             \;\le\;
             \alpha \, d(m,m_j)
             \;\le\;
             \frac{\varepsilon}{2}.
         \]
         For each $j = 1,\dots,N$, define functions
         \[
             \ell_j(z) := \psi(m_j,z) - \frac{\varepsilon}{2},
             \qquad
             u_j(z) := \psi(m_j,z) + \frac{\varepsilon}{2},
             \qquad z \in \mathcal{Z}.
         \]
         Then, for every $m \in \mathcal{M}$, choosing $j$ with 
         $d(m,m_j) \le \frac{\varepsilon}{2\alpha}$ yields
         \[
             \ell_j(z) \;\le\; \psi(m,z) \;\le\; u_j(z)
             \qquad \text{for all } z \in \mathcal{Z},
         \]
         so the bracket $[\ell_j,u_j]$ contains $\psi(m,\cdot)$.

         Moreover, the $L^p(\theLaw)$-size of each bracket is
         \[
             \left(\int_{\mathcal{Z}} \bigl|u_j(z) - \ell_j(z)\bigr|^p \, \ud\theLaw(z)\right)^{1/p}
             \;=\;
             \left(\int_{\mathcal{Z}} \varepsilon^p \, \ud\theLaw(z)\right)^{1/p}
             \;=\;
             \varepsilon.
         \]
         Hence each $[\ell_j,u_j]$ is an $\varepsilon$-bracket in $L^p(\theLaw)$, and
         the collection $\{[\ell_j,u_j] : j=1,\dots,N\}$ covers $\mathcal{F}$.
         By the definition of the bracketing number,
         \[
             \bracketing \bigl(\varepsilon, \mathcal{F}, L^p(\theLaw)\bigr)
             \;\le\; N
             \;=\;
             \covering \left(\frac{\varepsilon}{2\alpha}, \mathcal{M}, d\right).
         \]
         This completes the proof.
     \end{proof}
\end{Lemma}

\begin{Lemma} \label{Lemma:bracketing_for_strong_consistency}
    Let 
    $\psi_f(x,y) = \sphereDist^2( \Exp_\origin( f(x) ), y)$
    and
    \[
        \mathcal{F}(\eta)
        = \{\psi_f \mid f \in \ball_\Hspace(0,\eta) \}.
    \]
    Then,
        \[
            \bracketing (\vep, \mathcal F (\eta), L^1(\theLaw_{X,Y})) \leq \covering( \vep/(2\alpha), \ball_\Hspace(0,\eta), \inftynorm{.}),
        \]
        where $\alpha = 2\pi$.
    \begin{proof}
         Let $f,g \in \Hspace$. Since $\psi_f(x,y) = \ell_y(f(x))$, Lemma~\ref{Lemma:ell_y_Lipschitz} gives
                \[
                    \bigl|\psi_f(x,y) - \psi_{g}(x,y)\bigr| \le 2\pi\, \Ynorm{f(x) - g(x)}.
                \]
                We directly get $\bigl|\psi_f(x,y) - \psi_{g}(x,y)\bigr| \leq 2 \pi \inftynorm{f-g}$. 
                The claim follows from an application of Lemma~\ref{Lemma:bracketing_and_covering_for_Lipschitz}.
    \end{proof}
\end{Lemma}

\subsection{Technical results for measure theory}

\begin{Lemma} \label{Lemma:positive_integral}
    Let $(X, \mathcal A, \mu)$ be a measure space, $\varphi: X \to [0, \infty)$ be measurable, and let $A = \{ x \in X \mid \varphi(x) > 0 \}$.
    If $\mu(A) > 0$ then $\int_X \varphi \ud\mu > 0$.
    \begin{proof}
        For each $n\in\mathbb{N}$, set
        \[
            A_n := \{x\in X : \varphi(x) > 1/n\}.
        \]
        Then $A_n \subseteq A$ for all $n$ and
        \[
            A = \bigcup_{n=1}^{\infty} A_n.
        \]
        Moreover, for every $n$,
        \[
            \varphi(x) \ge \frac{1}{n}\,\mathbf{1}_{A_n}(x) \quad \text{for all } x\in X,
        \]
        so
        \[
            \int_X \varphi \, \ud\mu \;\ge\; \frac{1}{n}\,\mu(A_n).
        \]
        Assume, for contradiction, that $\int_X \varphi \, \ud\mu = 0$. Then the above inequality implies $\mu(A_n) = 0$ for all $n\in\mathbb{N}$. Hence, by continuity from below of measures,
        \[
            \mu(A)
            = \mu\!\left(\bigcup_{n=1}^{\infty} A_n\right)
            = \lim_{n\to\infty} \mu(A_n)
            = 0,
        \]
        which contradicts $\mu(A) > 0$. Therefore $\int_X \varphi \, \ud\mu > 0$.

    \end{proof}
\end{Lemma}

\begin{Lemma}\label{Lemma:essential_sup}
Let $(\mathcal Z, \mathcal A)$ be a measurable space, $(\Omega, \mathcal F, \theLaw)$ be a probability space, and $Z: \Omega \to \mathcal Z$ be measurable. Let $\mathcal{M}$ be a metric space with metric $d$.
Assume that $\Psi: \mathcal{Z} \times \mathcal{M} \to \Real$ is measurable in the first variable, i.e., $\Psi(\cdot, m)$ is measurable for all $m \in \mathcal M$, and that
\[
\left|\Psi(z,m) - \Psi(z,m')\right| \leq \eta(d(m,m')), \quad \forall z \in \mathcal Z,
\]
for some modulus of continuity $\eta: [0,\infty) \to [0,\infty)$, with $\eta(t)\to 0$ as $t \to 0$.

Then, for any $c \in \Real$,
\[
A := \{ m \in \mathcal{M} \mid \Psi(Z,m) \leq c,\, a.s. \}
\]
is a closed subset of $\mathcal M$, and
\[
    B := \{ m \in \mathcal{M} \mid \esssup \Psi(Z,m) < c \}
\]
is an open subset of $\mathcal M$.

\begin{proof}
Let $g(m) = \esssup \,\Psi(Z,m)$.
If $g(m) = +\infty$ for one $m \in \mathcal M$ then $g \equiv + \infty$ by the continuity modulus assumptions, and the conclusions trivially hold. 
We therefore consider the case $g(m) \in \Real$ for all $m \in \mathcal M$. 
We can rewrite
$A = \{ m \in \mathcal M \mid g(m) \leq c \} = g^{-1}( (-\infty, c])$,  and
$B =  g^{-1}( (-\infty, c))$.

We will show that $g$ is continuous, which will conclude the proof. By the modulus of continuity,
\[
\Psi(z,m') - \eta(d(m,m')) \leq \Psi(z,m) \leq \Psi(z,m') + \eta(d(m,m')), \quad \forall z \in \mathcal{Z},
\]
hence
\[
g(m') - \eta(d(m,m')) \leq g(m) \leq g(m') + \eta(d(m,m')).
\]
Thus $g(m)$ is continuous and the claims follow. 
\end{proof}
\end{Lemma}

\subsection{Results for convex functions on Hilbert spaces}

The following results are for convex functions on convex sets, and hold in particular for points at the boundary of the convex set.
\begin{Lemma}\label{Lemma:convex_gradient}
    Let $H$ be a Hilbert space with inner-product $\innerProd{\cdot, \cdot}$ and induced norm $\norm{\cdot}$. Let $A \subseteq H$ be open, and $\mathcal{C} \subseteq A$ be convex.
    \begin{enumerate}
        \item If $\Phi: A \to \Real$ is Fr\'echet differentiable, then 
    \begin{enumerate}
        \item $\Phi$ is convex on $\mathcal{C}$ if and only if $\forall u,v \in \mathcal{C}, \; \Phi(v) \geq \Phi(u) + \innerProd{\nabla \Phi(u), v-u}$ 
        \item If $\Phi$ is convex on $\mathcal{C}$ and if there exists $u_0 \in \mathcal{C}$ such that $\Phi(u) \geq \Phi(u_0)$ for all $u \in \mathcal{C}$, then $\innerProd{\nabla \Phi(u_0), u - u_0} \geq 0$ for all $u \in \mathcal{C}$.
    \end{enumerate}
        \item If $\Phi$ is twice Fr\'echet differentiable on $A$, and the minimal eigenvalue of its Hessian $H_\Phi(u)$ is at least $c \geq 0$ for all $u \in \mathcal C$, then for all $u,v \in \mathcal{C}$,
    \[
        \Phi(v) \geq \Phi(u) + \innerProd{\nabla \Phi(u), v-u} + \frac c2 \norm{v-u}^2.
    \]
    In particular, $\Phi$ is convex on $\mathcal C$, and strictly convex if $c > 0$, and at a minimum $u_0 \in \mathcal C$ of $\Phi$, 
    \[
        \Phi(v) \geq \Phi(u_0) + \frac c2 \norm{v-u_0}^2, \quad \forall v \in \mathcal{C}.
\]
    \end{enumerate}
    \begin{proof}
        \begin{enumerate}
            \item \mbox{}
                \begin{enumerate}
                    \item By convexity,
                        \[
                            {\Phi}\big(u + t(v-u)\big) \leq t\, {\Phi}(v) + (1-t)\, {\Phi}(u) = {\Phi}(u) + t\big({\Phi}(v) - {\Phi}(u)\big)
                        \]
                        This implies 
                        \[
                            \frac{{\Phi}\big(u + t(v-u)\big) - {\Phi}(u)}{t} \leq {\Phi}(v) - {\Phi}(u)
                        \]

                        Taking the limit $t \to 0^+$ yields
                        \[
                            \langle \nabla {\Phi}(u),\, v-u \rangle + {\Phi}(u) \leq {\Phi}(v)
                        \]
                        Note that this is even valid for $u \in \partial \mathcal{C}$.
                     
                        For the converse, let $u,v \in \mathcal{C}$ and $t \in [0,1]$. 
                        Since $\mathcal{C}$ is convex, the point
                        \[
                            w := t u + (1-t)v \in \mathcal{C}.
                        \]
                        By assumption, applied at $w$ with $u$ and $v$ respectively, we have
                        \begin{align}
                            \Phi(u) &\ge \Phi(w) + \innerProd{\nabla \Phi(w),\, u-w}, \label{eq:lemma_grad_convexity1}\\
                            \Phi(v) &\ge \Phi(w) + \innerProd{\nabla \Phi(w),\, v-w}. \label{eq:lemma_grad_convexity2}
                        \end{align}
                        Multiply \eqref{eq:lemma_grad_convexity1} by $t$ and \eqref{eq:lemma_grad_convexity2} by $(1-t)$ and add:
                        \begin{align*}
                            t\,\Phi(u) + (1-t)\,\Phi(v)
                            &\ge t\,\Phi(w) + (1-t)\,\Phi(w)
                               + \innerProd{\nabla \Phi(w),\, t(u-w) + (1-t)(v-w)} \\
                            &= \Phi(w) 
                               + \innerProd{\nabla \Phi(w),\, t u + (1-t)v - w}.
                        \end{align*}
                        Using $w = t u + (1-t)v$, the inner product term vanishes, so
                        \[
                            t\,\Phi(u) + (1-t)\,\Phi(v) \ge \Phi(w) = \Phi\big(tu + (1-t)v\big),
                        \]
                        which is exactly the convexity of $\Phi$ on $\mathcal{C}$.
                    \item If \( u_0 \) is in the interior of $\mathcal{C}$, the statement holds since \( \nabla {\Phi}(u_0) = 0 \) because $u_0$ is a minimiser.

                        In general, for $t > 0$, $\Phi(u_0 + t(u - u_0)) \geq \Phi(u_0)$ thus
                        \[
                            \frac{ \Phi(u_0 + t(u - u_0)) - \Phi(u_0) }{t} \geq 0.
                        \]
                        Taking the limit $t \to 0^+$ yields the result.
                \end{enumerate}
            \item For $t \in [0,1]$ define $g(t) = \Phi(u + t(v-u))$. By Taylor's theorem, for each $t \in (0,1]$  there is an $\alpha \in (0,t)$ such that 
                \begin{align*}
                    g(t) &= g(0) + g'(0)t + g''(\alpha)t^2/2 \\
                         &= \Phi(u) + t\innerProd{\nabla \Phi(u), v-u} + \frac{t^2}2 H_\Phi(u + \alpha(v-u))[v-u, v-u] \\
                         &\geq \Phi(u) + t\innerProd{\nabla \Phi(u), v-u} + \frac {ct^2}2 \norm{v-u}^2.
                \end{align*}
                For $t = 1$, the first statement follows by noticing that $\Phi(v) = g(1)$. Then we get $\Phi(v) \geq \Phi(u) + \innerProd{\nabla \Phi(u), v-u}$ for all $u,v \in \mathcal C$. Take $\theta \in [0,1]$ and set $z = \theta u + (1-\theta)v$. Apply the latter inequality to $u,z$ and to $v,z$ to get 
                \[
                    \Phi(u) \geq \Phi(z) + \innerProd{ \nabla \Phi(z), u-z}, \quad 
                    \Phi(v) \geq \Phi(z) + \innerProd{ \nabla \Phi(z), v-z}.
                \]
                Multiplying these inequalities by $\theta$ and $1-\theta$, respectively, and summing up implies that $\Phi(z) \leq \theta \Phi(u) + (1-\theta) \phi(v)$, thus $\Phi$ is convex on $\mathcal C$.
                The last statement follows by applying part 1.
        \end{enumerate}
\end{proof}
\end{Lemma}
Properties of convex functions defined on Hilbert spaces, and the link between strong and weak semi-continuity, are available in the literature for functions that are convex over the entire Hilbert space. We need results for functions that are only convex on a closed convex subset.
\begin{Lemma}\label{Lemma:convex-closed-semi_continuity}
    Let $H$ be a Hilbert space, $C \subseteq H$ be convex and closed. Assume $\phi:H \to \mathbb R$ is a function such that $\phi_{|C}$ is lower semi-continuous and convex. Then $\phi_{|C}$ is weakly lower semi-continuous. Furthermore, if $C$ is bounded or $\phi$ is coercive, the infimum $\inf_{x \in C} \phi(x)$ is achieved for an $x^* \in C$.
    \begin{proof}
        Let $\psi = \phi_{|C}$.
        For the first claim, by \citet[][Proposition 2.3]{peypouquetConvexOptimizationNormed2015}, $\mathrm{Epi}(\psi) \subseteq C \times \Real$ is closed, and it is easily checked that it is a convex set. Since $C$ is closed in $H$ then $\mathrm{Epi}(\psi)$ is also closed in $H \times \Real$. By \citet[][Proposition 1.23]{peypouquetConvexOptimizationNormed2015} it is therefore weakly closed. Applying again \citet[][Proposition 2.3]{peypouquetConvexOptimizationNormed2015}, $\psi$ is weakly lower semi-continuous.

        For the claim about the infimum, take $\{x_n\} \subset C$ a sequence achieving the infimum, i.e., $I := \inf_{x \in C} \phi(x) = \lim_{n \to \infty} \phi(x_n)$. We already know that $C$ is weakly closed since it is closed and convex. Furthermore since either $C$ is bounded or $\phi$ is coercive, the sequence $\{x_n\}$ is contained in a certain closed ball $B \subset H$, which is weakly compact \citep[][V.4.2 Theorem]{conwayCourseFunctionalAnalysis1997} and we can extract a subsequence from $\{x_n\}$, which we again denote by $\{x_n\}$ such that $x_n$ converges weakly to $x^* \in B$ for a certain $x^* \in B$, and since we have already established that $C$ is weakly closed, $x^* \in C$. By weak lower semi-continuity,
        \[
            \phi(x^*) \leq \lim_{n \to \infty} \phi(x_n) = I
        \]
        but since $x^* \in C$, $I \leq \phi(x^*)$ and thus $I = \phi(x^*)$.
    \end{proof}
\end{Lemma}

\subsection{Technical results for VVRKHS} \label{sec:technical_vvrkhs}

    Throughout this section, and unless stated otherwise, $\Yspace$ shall denote separable Hilbert space, $\Xspace$ a Polish space, and $\Hspace$ a VVRKHS of functions $f: \Xspace \to \Yspace$, with SMO kernel $K(x,x') = k(x,x') \id$, where $k$ is a scalar positive semi-definite kernel with RKHS $\mathcal{H}_k$.

\begin{Lemma}
    \label{Lemma:projection_against_y_is_in_scalar_RKHS}
    
    For any $f \in \Hspace$ and $y \in \Yspace$, the function $f' : x \mapsto \YinnerProd{ f(x), y }$ belongs to $\mathcal{H}_k$, and 
    $f' = \Psi_y f$ where $\Psi_y: \Hspace \to \mathcal H_k$ has operator norm $\Ynorm{y}$.

    Furthermore, if $J = \dim(\Yspace) < \infty$ and $\{e_j\}_{j=1}^J$ is an orthonormal basis of $\Yspace$, then
    \[
    \Hnorm{f}^2 = \sum_{j=1}^J \norm{\YinnerProd{f(\cdot), e_j}}_{\Hspace_k}^2, \quad f \in \Hspace.
    \]
    \begin{proof}
        Define the bounded linear map $\Phi_y : \mathcal{H}_k \to \Hspace$ by $\Phi_y(g) := g \cdot y$, where $(g \cdot y)(x) = g(x)y$. Since $\Hspace \cong \mathcal H_k \otimes \Yspace$, $\operatornorm{\Phi_y} = \Ynorm{y}$.

        For any $x \in \Xspace$, the reproducing property of $\Hspace$ and the identity
        $K(\cdot,x)\,y = k(\cdot,x)\,y = \Phi_y(k_x)$ give
        \[
            f'(x)
            = \YinnerProd{ f(x),\, y }
            = \HinnerProd{ f,\, k(\cdot,x)\,y }
            = \HinnerProd{ f,\, \Phi_y(k_x) }
            = \innerProd{ \Phi_y^\adjoint f,\, k_x }_{\mathcal{H}_k},
        \]
        so $f' = \Phi_y^\adjoint f \in \mathcal{H}_k$, with 
        \[
            \norm{f'}_{\mathcal{H}_k} \leq \operatornorm{\Phi_y^\adjoint} \Hnorm{f} = \Ynorm{y} \Hnorm{f}.
        \]
        Setting $\Psi_y := \Phi_y^\adjoint$ finishes the proof.
        
        For the second statement, it is enough to consider finite sums of the form $g = \sum_{i=1}^n k(\cdot, x_i) \xi_i$ with $x_i \in \Xspace$ and $\xi_i \in \Yspace$. We get
        \begin{align*}
        \Hnorm{g}^2 &= 
        \sum_{i,i'=1}^n k(x_i,x_{i'}) \YinnerProd{\xi_i, \xi_{i'}} 
        \intertext{and expanding the $\xi$s in the $e_j$ basis gives}
        & = \sum_{j=1}^J \sum_{i,i'} k(x_i,x_{i'}) \YinnerProd{\xi_i, e_j} \YinnerProd{\xi_{i'}, e_j},
        \\ & = \sum_{j=1}^J \norm{ \YinnerProd{g(\cdot), e_j}}_{\Hspace_k}^2.
        \end{align*}
        The proof is finished since $\Hspace$ is the completion of such $g$s and $g \mapsto \YinnerProd{g(\cdot), e_j}$ is continuous.
    \end{proof}
\end{Lemma}

\begin{Lemma}\label{Lemma:bounding_fX}
    \[
        \Ynorm{f(x)} \leq \Hnorm{f} \sqrt{k(x,x)},
    \]
    for all $x \in \Xspace$ and $f \in \Hspace$.
    More generally, if $K: \Xspace \times \Xspace \to \bounded(\Yspace)$ is not necessarily a SMO kernel, then
    \[
        \Ynorm{f(x)} \leq \Hnorm{f} \operatornorm{K(x,x)}^{1/2},
    \]

    \begin{proof}
See the proof of \citet[][Proposition 2]{carmeliVectorValuedReproducing2010}.
    \end{proof}
\end{Lemma}

\begin{Lemma} \label{Lemma:operator_on_Y_induces_operator_on_H}
    For $B \in \bounded(\Yspace)$ and $f \in \Hspace$, define $\tilde B f: \Xspace \to \Yspace$ by 
    \[
        (\tilde B f)(x) = B (f(x)), \quad \forall x \in \Xspace.
    \]
    Then $\tilde B \in \bounded(\Hspace)$ and $\operatornorm{\tilde B} = \operatornorm{B}$.
    In particular, if $\proj: \Yspace \to \Yspace$ is an orthogonal projection,
    \[
        \Hnorm{\tilde \proj f} \leq \Hnorm{f}, \quad f \in \Hspace.
    \]
    \begin{proof}
        The proof uses \citet[][Example 5]{carmeliVectorValuedReproducing2010} and properties of tensor products of operators on Hilbert spaces. 
    \end{proof}
\end{Lemma}

The following result tells us that orthogonal operators on $\Yspace$ induce orthogonal operators on $\Hspace$.  
\begin{Lemma} \label{Lemma:orthogonal_operator_Hspace}
    For $R \in \mathcal{O}(\Yspace)$ and $f \in \Hspace$ define $\tilde R f: \Xspace \to \Yspace$ by
    \[
        (\tilde R f)(x) = R f(x), \quad x \in \Xspace.
    \]
    Then $\tilde R\in \orthogonal(\Hspace)$, and hence $\tilde R f \in \Hspace,  \Hnorm{\tilde R f} = \Hnorm{f}$ and $\operatornorm{\tilde R} = 1$.
    \begin{proof}
        By Lemma~\ref{Lemma:operator_on_Y_induces_operator_on_H}, $\tilde R \in \bounded(\Hspace)$.
       Recall that $\Hspace$ is the completion of linear combinations of elements of the form $K_x y$ with $x \in \Xspace, y \in \Yspace$. From the definition of $\tilde R$,
        \[
            \tilde R K_x y = K_x Ry,
        \]
        which implies 
        \begin{align*}
            \HinnerProd{\tilde R K_x y , \tilde R K_{x'} y' } & = \HinnerProd{ K_x R y ,  K_{x'}  Ry' } \\
                                                              & = k(x', x) \YinnerProd{ R y ,  Ry' } \\
                                                              & = k(x', x) \YinnerProd{ y ,  y' } \\
                                                              & = \HinnerProd{ K_x y ,  K_{x'} y' } 
        \end{align*}
        By linear extension, $\tilde R \in \orthogonal(\Hspace)$ is an isometry on $\Hspace$.
    \end{proof}
\end{Lemma}

The following Lemma helps construct orthogonal operators that map specific subspaces to other subspaces while acting as the identity on some other subspaces.
\begin{Lemma} \label{Lemma:rotations}
    Let $\Yspace$ be a separable Hilbert space. Assume $U,V,W \subseteq \Yspace$ are finite-dimensional subspaces.
    If $U \subseteq V \cap W$ and $\dim(V) \leq \dim(W)$ then there exists $R \in \mathcal O(\Yspace)$ such that $R_{|U} = \id_U$ and $Rv \in W$ for all $v \in V$.
    \begin{proof}
        Set
        \[
            m_1 := \dim(U^{\perp}\cap V), \qquad
            m_2 := \dim(U^{\perp}\cap W).
        \]
        Since $U\subset V,W$ and the spaces are finite-dimensional, we have
        \[
            \dim V = \dim U + m_1, \qquad
            \dim W = \dim U + m_2,
        \]
        hence $m_1 \le m_2$.
        \medskip
        \noindent\emph{Step 1: case $m_1 = m_2$.}
        Assume $m_1 = m_2 =: m$.
        Choose orthonormal bases
        \[
            \{u_1,\dots,u_n\} \text{ of } U, \qquad
            \{v_1,\dots,v_m\} \text{ of } U^{\perp}\cap V, \qquad
            \{w_1,\dots,w_m\} \text{ of } U^{\perp}\cap W .
        \]
        Extend $\{u_1,\dots,u_n,v_1,\dots,v_m\}$ to an orthonormal basis of $\Yspace$ by adding
        a family $\{\tilde v_j\}_{j\in J}\subset V^{\perp}$.
        Similarly, extend $\{u_1,\dots,u_n,w_1,\dots,w_m\}$ to an orthonormal basis of $\Yspace$
        by adding a family $\{\tilde w_j\}_{j\in J}\subset W^{\perp}$ indexed by the same set $J$.

        Define $R:\Yspace\to \Yspace$ on these bases by
        \[
            Ru_i = u_i,\quad
            Rv_i = w_i \ (1\le i\le m),\quad
            R\tilde v_j = \tilde w_j \ (j\in J),
        \]
        and extend $R$ linearly and continuously.
        Then $R$ sends an orthonormal basis of $\Yspace$ onto another orthonormal basis of $\Yspace$,
        so $R$ is an orthogonal operator.

        For $u\in U$ we have $Ru=u$, hence $R_{\mid U}=\id_U$.
        If $v\in V$, write $v=u+v'$ with $u\in U$ and $v'\in U^{\perp}\cap V=\vspan\{v_1,\dots,v_m\}$.
        Then
        \[
            Rv = Ru + Rv' = u + Rv' \in U \oplus \vspan\{w_1,\dots,w_m\} = W,
        \]
        so $Rv \in  W$ in this case.

        \medskip
        \noindent\emph{Step 2: general case $m_1 < m_2$.}
        Choose any orthonormal vectors $v_{m_1+1},\dots,v_{m_2}$ in $V^{\perp}$.
        Since $V^{\perp}\subseteq U^{\perp}$, the family
        $\{v_1,\dots,v_{m_1},v_{m_1+1},\dots,v_{m_2}\}$ is an orthonormal subset of $U^{\perp}$.
        Define
        \[
            \widehat V := U \oplus \vspan\{v_1,\dots,v_{m_2}\}.
        \]
        Note that $V \subseteq \widehat V$.
        Then $U\subseteq \widehat V \cap W$, and $\dim \widehat V = \dim U + m_2 = \dim W$. 
        Applying Step~1 to the triple $(U,\widehat V,W)$ yields an orthogonal operator
        $R\in\mathcal O(\Yspace)$ such that $R_{\mid U}=\id_U$ and
        $R \hat v \in W$ for all $\hat v \in \widehat V$, thus $Rv \in W$ for all $v \in V$.
    \end{proof}

\end{Lemma}

Recall the definition of $\Uset$ from Section~\ref{sec:local_convexity}.

\begin{Lemma} \label{Lemma:Uset_is_closed_and_nonempty_new}
  Let
  \[
    \check \Uset
    :=
    \left\{
      f \in \Hspace \,\middle|\,
      \begin{array}{l}
        \operatorname*{ess\,sup}_X
        \Ynorm{f(\cdot)}
        < \pi 
      \end{array}
    \right\}.
  \]
  Provided $\sup_{x \in \Xspace} k(x,x) < \infty$,
  $\Uset$ is non-empty and closed, $\check \Uset$ is open, and 
  $\Uset \subseteq \check \Uset \subseteq \Hspace$.
  \begin{proof}
    The statements follow from Lemma~\ref{Lemma:bounding_fX}, and
    Lemma~\ref{Lemma:essential_sup}.
  \end{proof}
\end{Lemma}

Although a proof of this next result was given in \citet{micchelliLearningVectorValuedFunctions2005}, we include a corrected argument here, as the original argument contains a gap.
\begin{Lemma} \label{Lemma:minimum_interpolant_vvrkhs}
Let $\Xspace$ be a set, $\hy$ a Hilbert space, and $\Hspace$ the VVRKHS of functions $\Xspace \to \hy$ with SMO kernel $K(x,x') = k(x,x')\id_{\hy}$. Let $x_1,\ldots,x_n \in \Xspace$ and $f \in \Hspace$, and set $E := \vspan\{f(x_1),\ldots,f(x_n)\}$. Then $f = \tilde f + g$, where
\[
    \tilde f = \sum_{i=1}^n K_{x_i}\xi_i \quad \text{for some } \xi_1,\ldots,\xi_n \in E,
    \qquad g(x_i) = 0, \; i = 1,\ldots,n,
    \qquad \HinnerProd{\tilde f, g} = 0 .
\]
In particular, $\tilde f(x_i) = f(x_i)$ for all $i$, $\vspan\{\tilde f(x_1),\ldots,\tilde f(x_n)\} \subseteq E$, and $\Hnorm{\tilde f} \leq \Hnorm{f}$.
\end{Lemma}
\begin{proof}
Let $f \in \Hspace$ be fixed throughout.
Let
\[
    S := \Big\{ \sum_{i=1}^n K_{x_i}\eta_i \;:\; \eta_1,\ldots,\eta_n \in E \Big\} \subseteq \Hspace .
\]
Since $\dim E \leq n$, $S$ is a finite-dimensional, hence closed, subspace of $\Hspace$. Let $\tilde f := \proj_S f$ be the orthogonal projection of $f$ onto $S$, and $g := f - \tilde f \in S^\perp$. By definition of $S$, $\tilde f = \sum_{i=1}^n K_{x_i}\xi_i$ for some $\xi_i \in E$, and $\HinnerProd{\tilde f, g} = 0$ since $\tilde f \in S$.

It remains to show $g(x_i) = 0$. Fix $i$. On the one hand, $g(x_i) \in E$ since $f(x_i) \in E$ and $\tilde f(x_i) = \sum_{j=1}^n k(x_i,x_j)\xi_j \in E$. On the other hand, for every $\eta \in E$ we have $K_{x_i}\eta \in S$, so by the reproducing property
\[
    \YinnerProd{g(x_i), \eta} = \HinnerProd{g, K_{x_i}\eta} = 0 .
\]
Thus $g(x_i) \in E^\perp$ and therefore $g(x_i) \in E \cap E^\perp = \{0\}$.

Consequently $\tilde f(x_i) = f(x_i)$ for all $i$, which gives $\vspan\{\tilde f(x_i)\}_{i=1}^n = E$, and by orthogonality of $\tilde f$ and $g$, $\Hnorm{f}^2 = \Hnorm{\tilde f}^2 + \Hnorm{g}^2 \geq \Hnorm{\tilde f}^2$.
\end{proof}

\subsection{Technical results for effective dimensions}

In this section we assume that $(\sigma_j)_{j\geq 1}$ is a non-increasing non-negative and summable sequence. For $\lambda >0$ let
\[
    N(\lambda)\ :=\ \sum_{j\ge 1}\frac{\sigma_j}{\sigma_j+\lambda}.
\]

\begin{Lemma}
    \label{Lemma:effective_dimension_polynomial_decay}
    Assume $(\sigma_j)_{j\ge 1}$ is non-increasing and there exist constants $C>0$ and $p>1$ such that
    \[
        \sigma_j \le C j^{-p} \qquad \text{for all } j\ge 1.
    \]
    Then 
    \[
        N(\lambda) = O(\lambda^{-1/p}), \quad \text{for } \lambda \downarrow 0
    \]
    \begin{proof}
        Let $J \geq 1$ be an integer. Then
        \[
            N(\lambda)\ =\ \sum_{j=1}^{J}\frac{\sigma_j}{\sigma_j+\lambda}\ +\ \sum_{j>J}\frac{\sigma_j}{\sigma_j+\lambda}
            \ \le\ J\ +\ \sum_{j>J}\frac{\sigma_j}{\lambda}.
        \]
        Using $\sigma_j\le C j^{-p}$, we get
        \[
            \sum_{j>J}\frac{\sigma_j}{\lambda}\ \le\ \frac{C}{\lambda}\sum_{j>J} j^{-p}.
        \]
        Since $p>1$,
        \[
            \sum_{j>J} j^{-p}\ \le\ \int_{J}^{\infty} x^{-p}\,dx\ =\ \frac{J^{1-p}}{p-1}.
        \]
 Therefore
        \[
            N(\lambda)\ \le\ J\ +\ \frac{C}{\lambda}\cdot \frac{J^{1-p}}{p-1}.
        \]
        For $\lambda\in(0,C]$  set
        \[
            J\ :=\ \left\lceil (C/\lambda)^{1/p}\right\rceil.
        \]
                Then $(C/\lambda)^{1/p} \leq J\le 1+(C/\lambda)^{1/p} \leq 2 (C/\lambda)^{1/p}$. Moreover, since $1-p < 0$ and $J \geq (C/\lambda)^{1/p}$, we have $J^{1-p} \leq (C/\lambda)^{(1-p)/p}$, hence
        \[
            \frac{C}{\lambda}\cdot \frac{J^{1-p}}{p-1} \leq \frac{(C/\lambda)^{1 + (1-p)/p}}{p-1} = \frac{(C/\lambda)^{1/p}}{p-1}.
        \]
        Therefore $N(\lambda) \leq \frac{2p-1}{p-1}\, C^{1/p}\, \lambda^{-1/p}$ for all $\lambda \in (0, C]$, i.e.\ $N(\lambda) = O(\lambda^{-1/p})$ as $\lambda \downarrow 0$.
    \end{proof}
\end{Lemma}

\begin{Lemma}\label{Lemma:effective_dimension_stretched_exponential}
Assume $\sigma_j \le C\exp(-\alpha j^{1/q})$ for all $j \geq 1$, for some $C,\alpha>0$ and $q > 0$.
Then $N(\lambda) = O( (\log(1/\lambda))^{q})$ as $\lambda \downarrow 0$.
\end{Lemma}
\begin{proof}
Let $S := \sum_{j \geq 1} e^{-(\alpha/2) j^{1/q}}$. Fix any integer $k > q$ (e.g.\ $k = \lfloor q \rfloor + 1$). For $x > 0$, all terms of the series $e^x = \sum_{n \geq 0} x^n/n!$ are non-negative, so $e^x \geq x^k/k!$, i.e.\ $e^{-x} \leq k!\, x^{-k}$. With $x = \frac\alpha2 j^{1/q}$,
\[
    e^{-\frac\alpha2 j^{1/q}} \leq k! \Big(\frac2\alpha\Big)^k j^{-k/q},
\]
and since $k/q > 1$, the series $\sum_{j \geq 1} j^{-k/q}$ converges. Hence $S \leq k!\,(2/\alpha)^k \sum_{j\geq 1} j^{-k/q} < \infty$. Note that $S$ depends only on $\alpha$ and $q$, not on $\lambda$ or $J$.

Let $J \geq 1$ be an integer. Since $\sigma_j/(\sigma_j+\lambda) \leq \min\{1, \sigma_j/\lambda\}$,
\[
    N(\lambda) \le J + \frac{C}{\lambda}\sum_{j>J} e^{-\alpha j^{1/q}} .
\]
For $j > J$, $e^{-\alpha j^{1/q}} = e^{-(\alpha/2) j^{1/q}}\, e^{-(\alpha/2) j^{1/q}} \leq e^{-(\alpha/2) J^{1/q}}\, e^{-(\alpha/2) j^{1/q}}$, hence
\[
    \sum_{j>J} e^{-\alpha j^{1/q}} \leq e^{-(\alpha/2) J^{1/q}}\, S .
\]
Now let $\lambda \in (0, C)$, put $L := \log(C/\lambda) > 0$, and choose
\[
    J := \big\lceil (2L/\alpha)^q \big\rceil \geq 1 .
\]
Then $(\alpha/2) J^{1/q} \geq L$, so $C e^{-(\alpha/2)J^{1/q}} \leq C e^{-L} = \lambda$, and therefore
\[
    N(\lambda) \leq J + \frac{C e^{-(\alpha/2)J^{1/q}}}{\lambda}\, S \leq J + S \leq 1 + \Big(\frac{2}{\alpha}\Big)^q L^q + S .
\]
Since $L = \log C + \log(1/\lambda)$, we have $L = O(\log(1/\lambda))$ as $\lambda \downarrow 0$, and hence $N(\lambda) = O((\log(1/\lambda))^q)$.
\end{proof}

\subsection{Technical results for the sphere}
\label{sec:technical_sphere}

This section contains some technical results for the sphere in a Hilbert space. These are mostly known and available in the literature for the finite-dimensional sphere, but are harder to find for the infinite-dimensional sphere.
In the following,
we will write $\sinc(x) = \sin(x)/x$ if $x \neq 0$ and $\sinc(0)=1$. By using a Taylor expansion of $\sin(x)$, we notice that $\sinc$ is $\continuous^\infty$ on $\Real$.

\begin{Lemma} \label{Lemma:rotations_and_sphere}
    For any orthogonal transformation $R \in \orthogonal(\Yspace)$, $\sphereDist(p, \check p) = \sphereDist(Rp, R\check p)$ and $R \Exp_p(v) = \Exp_{Rp}(Rv)$, for all $p \in \sphere, v \in T_p \sphere$. 
    \begin{proof}
        Note that $Rv \in T_{Rp} \sphere$ since $R$ is orthogonal.
        The proof follows directly from the formulas for $\sphereDist$ and $\Exp_p(v)$. 
    \end{proof}
\end{Lemma}

\begin{Lemma} \label{Lemma:operator_norm_exp}

For all $p \in \sphere, v, w \in T_p \sphere$,
\begin{equation} \label{eq:differential_of_Exp_in_v}
    D \Exp_p(v)[w] = -\sinc \norm{v}  \innerProd{v,w} p + \frac{ \cos \norm{v} - \sinc \norm{v} }{\norm{v}^2} \innerProd{v,w} v + \sinc \norm{v} w,
\end{equation}
where the formula is interpreted continuously at $v=0$.
In particular,
$\operatornorm{ D \Exp_p(v) } = 1$.
\begin{proof}
Direct calculations yield \eqref{eq:differential_of_Exp_in_v}.
In particular, 
\[
    D \Exp_p(v)[v] = -\sin( \norm{v} ) \norm{v} p + \cos \norm{v} v
\]
and since $\innerProd{p,v}=0$,  $\norm{D \Exp_p(v)[v]}^2 =  \norm{v}^2$ and thus $\operatornorm{D \Exp_p(v)} \geq 1$ for $v \neq 0$ and $\operatornorm{D \Exp_p(v)} = 1$ if $\dim \sphere = 1$.
For $v=0, \; D \Exp_p(0) = \id$ and hence  $\operatornorm{D \Exp_p(0)} = 1$.

Let us assume $v \neq 0$ and $\dim \sphere \geq 2$ from now on.
  If $w \in T_p \sphere$ and $\innerProd{w,v} = 0$ then 
\[
    D \Exp_p(v)[w] = \sinc( \norm{v} ) w
\]
and hence $\norm{D \Exp_p(v)[w]}^2 \leq  \norm{w}^2$. Note also that $\innerProd{D \Exp_p(v)[v], D \Exp_p(v)[w]} = 0$.
    For $u \in T_p \sphere$, writing $u = \alpha v + w$ where $\alpha \in \Real, \innerProd{w,v} = 0$, we get
\[
D \Exp_p(v)[u] = \alpha D \Exp_p(v)[v] + D \Exp_p(v)[w],
\]
and 
\[
    \norm{D \Exp_p(v)[u]}^2 = \alpha^2 \norm{ D \Exp_p(v)[v]}^2 + \norm{D \Exp_p(v)[w]}^2 \leq \alpha^2 \norm{v}^2 + \norm{w}^2 = \norm{u}^2.
\]
Thus $\operatornorm{D \Exp_p(v)} = 1$ if $v \neq 0$. This concludes the proof.
\end{proof}
\end{Lemma}

\begin{Lemma} \label{Lemma:distance_and_exp_fixed_p}
 Let $p\in \sphere$, and let $u, u' \in T_p \sphere$. Then,
    \[
        \sphereDist( \Exp_p(u),\Exp_p(u') ) \leq \norm{u-u'}.
    \]
    \begin{proof}
        Define
        \[
            \gamma(t) = u + t(u'-u), \quad t\in[0,1].
        \]
        For all $t \in [0,1]$,
        let
        \[
            \sigma(t) = \Exp_p(\gamma(t)), \quad t\in[0,1].
        \]
        Then $\sigma$ is a $\continuous^1$ curve joining $\Exp_p(u)$ to $\Exp_p(u')$, and its length is
        \[
            L(\sigma)
            = \int_0^1 \|\dot{\sigma}(t)\|\,dt
            = \int_0^1 \norm{ D\Exp_p(\gamma(t))[\dot{\gamma}(t)]} \,dt.
        \]
        Using the operator norm, 
        \[
            L(\sigma)
            \le \int_0^1 \operatornorm{D\Exp_p(\gamma(t))}
            \,\norm{\dot{\gamma}(t)}\,dt.
        \]
        Since $\dot{\gamma}(t) = u'-u$, Lemma~\ref{Lemma:operator_norm_exp} implies
        \[
            L(\sigma) \le \int_0^1 \norm{u'-u} \,dt = \norm{u - u'}.
        \]
        By definition of the Riemannian distance,
        \[
            \sphereDist(\Exp_p(u),\Exp_p(u'))
            \le L(\sigma)
            \le \norm{u-u'}
        \]
    \end{proof}
\end{Lemma}

\begin{Lemma} \label{Lemma:gradient_of_squared_distance}
    Let $y \in \sphere$ and $\phi: \sphere \to \Real$ be defined by $\phi(x) = \sphereDist^2(x,y)$. Then for $x \neq -y$,
    $\nabla \phi (x) = -2\Log_x(y)$.
    \begin{proof}
        The proof can be done by explicitly calculating the gradient of $\phi$ from the formula of $\sphereDist$. 
    \end{proof}
\end{Lemma}

Recall the definition of $\ell_y$ from \eqref{eq:ell_function}.
\begin{Lemma} \label{Lemma:gradient_l_v}
    Provided  $\Exp_{\origin}(v) \neq -y$,
    \[
         \nabla \ell_y(v) = -2 ( D \Exp_\origin(v) )^\adjoint \Log_{\Exp_\origin(v)} (y) \in T_\origin \sphere \quad \text{and} \quad \norm{\nabla \ell_y(v)} \leq 2\pi.
    \]
    \begin{proof}
        The proof follows by using the chain rule, and using Lemmas~\ref{Lemma:operator_norm_exp} and~\ref{Lemma:gradient_of_squared_distance}. 
    \end{proof}
\end{Lemma}

\begin{Lemma} \label{Lemma:ell_y_Lipschitz}
    For all $y \in \sphere$, the function $\ell_y$ is $2\pi$-Lipschitz on $T_\origin \sphere$, that is,
    \[
        |\ell_y(u) - \ell_y(u')| \leq 2\pi \Ynorm{u - u'}, \quad u, u' \in T_\origin \sphere.
    \]
    \begin{proof}
        By the reverse triangle inequality, the fact that the diameter of $\sphere$ is $\pi$, and Lemma~\ref{Lemma:distance_and_exp_fixed_p},
        \begin{align*}
            |\ell_y(u) - \ell_y(u')|
            &= \left| \sphereDist^2\bigl(\Exp_\origin(u),y\bigr) - \sphereDist^2\bigl(\Exp_\origin(u'),y\bigr) \right| \\
            &\leq \bigl[ \sphereDist\bigl(\Exp_\origin(u),y\bigr) + \sphereDist\bigl(\Exp_\origin(u'),y\bigr) \bigr] \sphereDist\bigl(\Exp_\origin(u),\Exp_\origin(u')\bigr) \\
            &\leq 2\pi \Ynorm{u - u'}.
        \end{align*}
    \end{proof}
\end{Lemma}

\subsection{Gradient computation}\label{subsec:grad_computation}
Let $(x_i,y_i)$ be realised observations of $(X_i,Y_i)$. 

To compute our estimator, the goal is to minimise the empirical risk $\risk_n(f, \lambda_n)$ with respect to $f$. By Theorem~\ref{theorem:existence_of_global_minimizer_unconstrained}, assuming that a unique minimiser of $\risk_n(\cdot, \lambda_n)$ exists, we can search for the optimal $f$ by finding the optimal coefficients $\xi_1,\ldots, \xi_n \in \vspan( \Log_\origin(Y_1), \ldots, \Log_\origin(Y_n) )$.
Letting $\bxi = (\xi_1, \ldots, \xi_n)$, the reparametrised empirical risk is $R_n( \bxi) := \risk_n( f_\bxi, \lambda_n)$ where $f_\bxi(\cdot) = \sum_{i=1}^n K(x_i, \cdot) \xi_i$, so
\begin{align*}
    R_n(\boldsymbol{\xi})&= \frac{1}{n}\sum_{i=1}^n (\arccos(\langle\Exp_\origin f_\bxi(x_i), y_i \rangle_{\Yspace}) )^2 + \lambda_n^2 \Hnorm{f_\bxi}^2.
\end{align*}
The strategy to compute the gradient of $ R_n(\boldsymbol{\xi})$ with respect to one $\xi_l \in T_\origin\sphere$, $l \in \{1, \dots, n\}$, is to first consider the smooth extension $ \bar R_n(\boldsymbol{\xi})$ defined on the ambient space $(\Yspace)^n$, and then project the gradient onto $T_\origin \sphere$. In the following, we denote $\mu_i = \Exp_\origin (f_\bxi(x_i))$ and $k_{il} = k(x_i, x_l)$, so that $K(x_i, x_l) = k_{il}\, \id$ and $f_\bxi(x_i) = \sum_{j=1}^n k_{ij} \xi_j$. 

For fixed \(i\), consider the function
$\mu \mapsto \big(\arccos\langle\mu,y_i\rangle_{\Yspace}\big)^2,$ $ \mu\in\Yspace.$
Its Fréchet derivative at \(\mu_i\) in direction \(v\in\Yspace\) is
\begin{align*}
    D_{\mu_i} \left[ (\arccos(\langle\mu_i , y_i \rangle_{\Yspace}) )^2 \right]  [v]= - 2 \frac{\arccos \langle \mu_i, y_i \rangle}{\sqrt{1-(\langle \mu_i, y_i \rangle)^2}} \langle v, y_i \rangle
\end{align*}
Using the explicit expression of \(\Exp_\origin\) on the sphere, one obtains
for the ambient differential (computed as a map between open subsets of
\(\Yspace\)):
\begin{align*}
D_{\xi_l} \left[ \Exp_\origin \left(\sum_{j=1}^n k_{ij}\xi_j\right) \right] [v] &= D\Exp_\origin\big(f_\bxi(x_i)\big)[k_{il} v] \\
&= - \sin\| f_\bxi(x_i)\| \frac{\langle f_\bxi(x_i), k_{il} v \rangle }{\|f_\bxi(x_i)\|}  \origin \\
&\quad + \cos\| f_\bxi(x_i)\| \frac{\langle f_\bxi(x_i), k_{il} v \rangle }{\|f_\bxi(x_i)\|^2} f_\bxi(x_i) \\
&\quad + \frac{\sin\| f_\bxi(x_i)\| }{\| f_\bxi(x_i)\|} \left[k_{il} v- f_\bxi(x_i) \frac{\langle k_{il} v, f_\bxi(x_i)\rangle}{\| f_\bxi(x_i)\|^2}\right]
\end{align*}
Consider the smooth extension $\overline R_n : (\Yspace)^n \rightarrow \Real$ of $R_n: (T_\origin \sphere)^n \rightarrow \Real$ defined by 
\begin{align*}
    \overline R_n(\boldsymbol{\xi})&= \frac{1}{n}\sum_{i=1}^n (\arccos(\langle\Exp_\origin \sum_{j=1}^n k_{ij}\xi_j, y_i \rangle_{\Yspace}) )^2 + \lambda_n^2 \Hnorm{\sum_{j=1}^n K_{x_j}\xi_j}^2,
\end{align*}
Then, the differential of $\overline R_n$ with respect to $\xi_l$, in the direction $v \in \Yspace$, is given by 
\begin{align*}
D_{\xi_l}   \overline R_n (\bxi) [v] = - \frac{2}{n} \sum_{i = 1}^n  \frac{\arccos \langle \mu_i, y_i \rangle}{\sqrt{1-(\langle \mu_i, y_i \rangle)^2}} \langle D_{\xi_l} \mu_i(\bxi)[v] , y_i \rangle + 2\lambda_n^2 \langle f_\bxi(x_l), v \rangle 
\end{align*}
and the gradient of $R_n$ is given by 
\begin{align*}
    \nabla_{\xi_l} R_n (\bxi) = \proj_\origin(\nabla_{\xi_l} \overline R_n (\bxi) ).
\end{align*}
Finally, the gradient 
$\nabla R_{n} (\boldsymbol{\xi}) \in (T_\origin\sphere)^n$ is given by 
\begin{align} \label{eq:gradxi}
\nabla_{\xi_l} R_n (\bxi) = \frac{2}{n} \sum_{i=1}^n \left[-  \frac{k_{il} \arccos(\langle \mu_i, y_i\rangle_{\Yspace})}{\sqrt{1- \langle \mu_i, y_i\rangle_{\Yspace}^2}} \proj_\origin[ \gamma_i] \right] + 2\lambda_n^2 f_\bxi(x_l) ,
\end{align}
where 
\begin{align*}
\gamma_i
&=
\frac{\sin\|f_\bxi(x_i)\|_{\Yspace}}{\|f_\bxi(x_i)\|_{\Yspace}}
\!\left(
-\left(\langle \origin, y_i\rangle_{\Yspace}
+ \frac{\langle f_\bxi(x_i),y_i\rangle_{\Yspace}}{\|f_\bxi(x_i)\|^2}\right) f_\bxi(x_i)
+ y_i
\right)
+\\&+
\frac{\cos\|f_\bxi(x_i)\|_{\Yspace}}{\|f_\bxi(x_i)\|_{\Yspace}^2}
\langle f_\bxi(x_i),y_i\rangle_{\Yspace}\, f_\bxi(x_i).
\end{align*}
Note that 
\begin{align*}
    \proj_\origin[ \gamma_i] &= \frac{\sin\|f_\bxi(x_i)\|_{\Yspace}}{\|f_\bxi(x_i)\|_{\Yspace}} \Big(-\Big(\langle \origin, y_i\rangle_{\Yspace} + \frac{\langle f_\bxi(x_i), y_i\rangle_{\Yspace}}{\|f_\bxi(x_i)\|_{\Yspace}^2} \Big)f_\bxi(x_i) +\proj_\origin[ y_i] \Big)  +\\& +\frac{\cos\|f_\bxi(x_i)\|_{\Yspace}}{\|f_\bxi(x_i)\|_{\Yspace}^2} \langle f_\bxi(x_i), y_i\rangle_{\Yspace} f_\bxi(x_i).
\end{align*}

Another point of view would be to leverage Theorem~\ref{theorem:existence_of_global_minimizer_unconstrained} and write $\xi_l = \sum_{j=1}^n a_{lj} \Log_\origin(y_j)$. Letting $A$ be the $n \times n$ matrix with $(A)_{lj} = a_{lj}$, we get the empirical risk
\[
    \widetilde R_n(A) := \risk_n(f_A, \lambda_n)
\]
where $f_A = \sum_{l=1}^n K(x_l, \cdot) \sum_{j=1}^n a_{lj} \Log_\origin (y_j)$.
The partial derivative of $ \widetilde R_n$ with respect to the coefficient $ a_{l,j}$ is given by
\[
\frac{\partial \widetilde R_n}{\partial a_{l,j}}
\;=\;
\left\langle \nabla_{\xi_l}R_n ,\, \frac{\partial \xi_l}{\partial a_{l,j}} \right\rangle_{\Yspace}
\;=\;
\left\langle \nabla_{\xi_l}R_n ,\, \Log_\origin (y_j) \right\rangle_{\Yspace}.
\]
Using the explicit expression of $\nabla_{\xi_l}R_n$ in~\eqref{eq:gradxi}, we obtain
\begin{align*}
\frac{\partial\widetilde R_n}{\partial a_{l,j}}
&=\;\frac{2}{n}\sum_{i=1}^n 
\left\{
-\,\frac{k_{il}\,\arccos(\langle \mu_i , y_i \rangle_{\Yspace})}
{\sqrt{1-\langle \mu_i , y_i \rangle_{\Yspace}^{\,2}}}
\left\langle 
\proj_\origin\big[ \gamma_i \big],\,
\Log_\origin (y_j)
\right\rangle_{\Yspace}
\right\}
\\
&\quad +\;
2\lambda_n^2 \,\langle f_A (x_l),\Log_\origin (y_j)\rangle_{\Yspace}.
\end{align*}

Hence, for each $l\in\{1,\dots,n\}$, the gradient of $\widetilde R_n$ with respect to the coefficient vector 
$(a_{l,1},\dots,a_{l,n})$ is obtained by taking the inner products of $\nabla_{\xi_l}R_n$ with 
$\Log_\origin (y_1), \ldots, \Log_\origin (y_n)$.

\subsection{Implementation details}\label{subsec:implementation}

Algorithm~1 summarises the fitting procedure of Section~\ref{sec:computational_aspects}, as implemented in the \texttt{R} package \texttt{sphereg}. Its two low-rank reductions, of the response space and of the kernel, are instances of the same construction, which we briefly describe first.

\medskip
\noindent\textbf{Low-rank approximation of Gram matrices.}
Let $h_1, \dots, h_n$ be elements of a Hilbert space $\mathcal W$ with Gram matrix $\mathbf M = [\langle h_i, h_j \rangle_{\mathcal W}]_{i,j=1}^n$, and let $m \leq \operatorname{rank}(\mathbf M)$. 
We approximate $h_1, \dots, h_n$ from within an $m$-dimensional subspace $\mathcal W_m \subset \mathcal W$ that is chosen optimally in the sense that $\mathbf M$ is approximated as well as possible by the Gram matrix $\mathbf{M}_m = [\langle \proj_{\mathcal W_m} h_i, \proj_{\mathcal W_m} h_j \rangle_{\mathcal W}]_{i,j=1}^n$ of the orthogonal projections $\proj_{\mathcal W_m} h_i$ of $h_i$ onto $\mathcal W_m$, i.e., such that inner products are preserved as well as possible. 
By the Eckart--Young--Mirsky theorem \citep[e.g.,][]{HsingEubank2015}, this best rank-$m$ approximation, in Frobenius as well as in spectral norm, is obtained from the eigendecomposition of $\mathbf{M}$ truncated to its $m$ leading eigenpairs; it is realised by the span $\mathcal W_m$ of the corresponding (uncentred) principal components of $h_1, \dots, h_n$, with total squared projection error $\operatorname{tr}(\mathbf M - \mathbf M_m)$. 
We choose $m$ as the smallest integer with $\operatorname{tr}(\mathbf M - \mathbf M_m) \leq \tau \operatorname{tr}(\mathbf M)$ for a tolerance $\tau \in [0, 1)$, i.e., such that the $m$ leading eigenvalues retain a share of at least $1 - \tau$ of the eigenvalue mass (the $\tau$-rule). For $\tau = 0$, only numerically zero eigenvalues are discarded and the representation is exact with $m = \operatorname{rank}(\mathbf M) \leq n$.

For the kernel, $\mathcal W = \Hspace_k$ is the RKHS of $k$, $h_i = k(X_i, \cdot)$, $\mathbf M = \mathbf K$ and $m = m_k$. An orthonormal basis $g_1, \dots, g_{m_k}$ of $\mathcal W_{m_k}$, obtained from the eigendecomposition, replaces $k(X_1, \cdot), \dots, k(X_n, \cdot)$, i.e., we restrict the candidate functions to $f_{\boldsymbol\beta} = \sum_{l=1}^{m_k} g_l\, \beta_l$ with coefficients $\boldsymbol\beta = (\beta_1, \dots, \beta_{m_k}) \in (T_\origin \sphere)^{m_k}$. This is the truncation that \citet{wood2003lowrank} introduced to construct low-rank thin plate regression splines \citep[see also][]{woodGeneralizedAdditiveModels2017a}. 

For the responses, $\mathcal W = \Yspace$, $h_i = Y_i$, $\mathbf M = \mathbf G$ and $m = m_y$, and each projection is normalised to unit norm, $\tilde Y_i := \proj_{\mathcal W_{m_y}} Y_i / \Ynorm{\proj_{\mathcal W_{m_y}} Y_i}$, which is the point of $\sphere \cap \mathcal W_{m_y}$ closest to $Y_i$; Algorithm~1 then operates on $\tilde Y_1, \dots, \tilde Y_n$ in place of $Y_1, \dots, Y_n$.
Here, the low-rank approximation corresponds to a truncated (uncentred) principal component representation of the data \citep[see, e.g.,][]{ramsayFunctionalDataAnalysis2005}, with the principal components computed in their dual form through the $n \times n$ Gram matrix rather than through a covariance operator on $\Yspace$. 
Representing $\proj_{\mathcal W_{m_y}} Y_i$ by its coordinate vector $\mathbf z_i \in \Real^{m_y}$ in an orthonormal basis $e_1, \dots, e_{m_y}$ of $\mathcal W_{m_y}$ yields an isometric representation in $\Real^{m_y}$. All subsequent computations---normalisation, the Fréchet mean, $\Log_\origin$ and $\Exp_\origin$, spherical distances, and the optimisation---can be carried out on the unit sphere of $\Real^{m_y}$. 

\medskip
\noindent\fbox{\parbox{0.97\textwidth}{
\textbf{Algorithm 1} (Spherical kernel ridge regression).\\
\emph{Input:} data $(X_i,Y_i)_{i=1}^n$, scalar kernel $k$, penalty $\lambda$, truncation tolerance $\tau$.
\begin{enumerate}[nosep]
    \item \emph{Response coordinates.} Form $\mathbf G = [\YinnerProd{Y_i,Y_j}]_{i,j=1}^n$, apply the rank-$m_y$ approximation with $m_y$ chosen by the $\tau$-rule, and replace each $Y_i$ by $\tilde Y_i \in \sphere \cap \mathcal W_{m_y}$, represented by its coordinate vector $\mathbf z_i/\norm{\mathbf z_i} \in \Real^{m_y}$; steps 2, 4 and 5 are carried out in these coordinates.
    \item \emph{Intercept.} Compute $\origin$ as the Fréchet (Karcher) mean of $\tilde Y_1,\ldots,\tilde Y_n$ by the fixed-point iteration $\origin \leftarrow \Exp_\origin\big(\tfrac12 \bar v\big)$, $\bar v = \tfrac1n \sum_i \Log_\origin \tilde Y_i$, started at the normalised sample mean and stopped when $\Ynorm{\bar v} < 10^{-3}$ or after $100$ iterations; $\origin$ is held fixed afterwards.
    \item \emph{Kernel reduction.} Form $\mathbf K = [k(X_i,X_j)]_{i,j=1}^n$, apply the rank-$m_k$ approximation with $m_k$ chosen by the $\tau$-rule, and obtain an orthonormal basis $g_1,\ldots,g_{m_k}$ of $\mathcal W_{m_k} \subset \Hspace_k$; the candidate functions are $f_{\boldsymbol\beta} = \sum_{l=1}^{m_k} g_l\,\beta_l$ with $\boldsymbol\beta = (\beta_1,\ldots,\beta_{m_k}) \in (T_\origin\sphere)^{m_k}$ and $\Hnorm{f_{\boldsymbol\beta}}^2 = \sum_{l=1}^{m_k} \Ynorm{\beta_l}^2$.
    \item \emph{Initialisation.} Set $\boldsymbol\beta^{(0)}$ to the tangent-space Ridge solution, i.e., the minimiser of $\frac1n\sum_i \Ynorm{f_{\boldsymbol\beta}(X_i) - \Log_\origin \tilde Y_i}^2 + \lambda^2 \Hnorm{f_{\boldsymbol\beta}}^2$, given in closed form by $[\beta^{(0)}_1,\ldots,\beta^{(0)}_{m_k}] = [\Log_\origin \tilde Y_1,\ldots,\Log_\origin \tilde Y_n]\,\mathbf Q(\mathbf Q^\top\mathbf Q + n\lambda^2\mathbf I)^{-1}$ with $\mathbf Q = [g_l(X_i)]_{i,l}$, where $\mathbf Q^\top\mathbf Q$ is diagonal for the basis obtained from the eigenpairs of $\mathbf K$.
    \item \emph{Optimisation.} Minimise $\boldsymbol\beta \mapsto \risk_n(f_{\boldsymbol\beta}, \lambda)$, evaluated on $(X_i, \tilde Y_i)_{i=1}^n$, with BFGS (\texttt{stats::optim}), using the gradient of Section~\ref{subsec:grad_computation}, in which $y_i$ is replaced by $\tilde Y_i$, $k_{il}$ by $g_l(X_i)$ and the penalty term $2\lambda_n^2 f_{\bxi}(x_l)$ of \eqref{eq:gradxi} by $2\lambda^2 \beta_l$, with the default stopping rule (relative decrease of the objective below $\sqrt{\varepsilon_{\mathrm{machine}}} \approx 1.5\cdot10^{-8}$) and at most $100$ iterations.
\end{enumerate}
\emph{Output:} $\hat{\boldsymbol\beta}$, $\origin$, and $\hat\mu(x) = \Exp_\origin\big(\sum_{l=1}^{m_k} g_l(x)\hat\beta_l\big)$.
}}

\noindent\textbf{Selection of the tuning parameters.} The penalty $\lambda$ and, where applicable, the kernel bandwidth are selected by cross-validation on a product grid, minimising the average prediction error
\[
    \frac{1}{|\mathcal I|}\sum_{i \in \mathcal I} \sphereDist\big(\hat\mu^{(-\mathcal I)}(X_i), Y_i\big)
\]
over the held-out folds $\mathcal I$, where $\hat\mu^{(-\mathcal I)}$ is fitted on the remaining observations. In the simulation study of Section~\ref{sec:simulation} we use $5$ folds drawn at random, a fixed bandwidth (see Section~\ref{subsec:simulation}), $\lambda$ on the grid $\{2^{-7}, 2^{-4.75}, 2^{-2.5}, 2^{-0.25}, 2^{2}\}$, and $\tau = 10^{-4}$. In the application of Section~\ref{sec:application}, the out-of-sample evaluation is leave-one-station-out, and for each held-out station the tuning parameters are selected by a further leave-one-out cross-validation over the remaining $13$ stations, on the product of ten squared bandwidths equispaced in $[10^{-8}, 10^{-4}]$ and ten values of $\lambda$ equispaced in $[9,25]$. There the low-rank reductions are left at $\tau = 0$, i.e.\ no truncation beyond the numerical rank of $\mathbf K$ and $\mathbf G$.

\subsection{Results for identifying regions of local convexity}

\begin{Proposition}
	\label{proposition:ell_Hessian_computation}
	For $y, \origin \in \sphere$, the Hessian $\Hess_{\ell_y}(v)$ of $\ell_y(v)$ at $v\in T_\origin \sphere$ depends on $y$ and $v$ only via
	$\zeta = \sphereDist(y, \origin)$, $\langle v, y \rangle_\Yspace$ and $\|v\|_\Yspace$, with the explicit form given in the proof. 
\begin{proof}
	In this proof, we will write $\|\cdot\|$ for $\|\cdot\|_\Yspace$ and $\langle\cdot,\cdot\rangle$ for $\langle\cdot,\cdot\rangle_\Yspace$ to simplify the notation, since there is no risk of confusion. 
	
	To compute the Hessian of $\ell_y: T_\origin \sphere \rightarrow \Real$, we will first compute the Hessian of the smooth extension
		\begin{align*}
			\overline \ell_y& : \Yspace \rightarrow \Real, \quad \overline \ell_y(v) =\arccos^2 \langle \cos \|v\| \origin + \sin \|v\|  \, \frac{v}{\|v\| }, y \rangle.
		\end{align*}
  Since the domain of $\ell_y$, $T_\origin \sphere$, is a linear subspace of $\Yspace$, the Hessian $\Hess_{\ell_y}(v)$ corresponds to the orthogonal projection of the Hessian $\Hess_{\overline\ell_y} (v)$ on $T_\origin \sphere$, 
		\begin{align}\label{eq:project_Hessian}
			\Hess_{\ell_y}(v) &= \proj_\origin \Hess_{\overline\ell_y} (v)\proj_\origin,
		\end{align}
where $\proj_\origin = \id- \origin \otimes \origin $, with $\id$ the identity operator on $\Yspace$.
		Now we compute $\Hess_{\overline\ell_y} (v)$. Denote $\phi(v) = \langle\Exp_\origin v,y\rangle$
		\begin{align*}
			\nabla\phi (v) &= v \Big( - \frac{\sin(\nv)}{\nv} \py + \frac{\cos(\nv)}{\nv^2}  \vy - \frac{\sin(\nv)}{\nv^3} \vy \Big) + \frac{\sin(\nv)}{\nv} y. \\
			\nabla\overline \ell_y(v) &= -2 \underbrace{ \frac{\arccos(\langle\Exp_\origin v,y\rangle)}{\sqrt{1-(\langle\Exp_\origin v,y\rangle)^2}} }_{h(v)}   \nabla\phi(v)\\
			\nabla h(v) & = \Big(\frac{-1}{1-(\langle\Exp_\origin v,y\rangle)^2} + \frac{\langle\Exp_\origin v,y\rangle \arccos(\langle\Exp_\origin v,y\rangle) }{(1-(\langle\Exp_\origin v,y\rangle)^2)^{3/2}} \Big)    \nabla\phi(v)\\
			\Hess_{\phi} (v) & = \Big( - \frac{\sin(\nv)}{\nv} \py + \frac{\cos(\nv)}{\nv^2}  \vy - \frac{\sin(\nv)}{\nv^3} \vy \Big) \text{Id} \\ &+\nabla_v \Big( \underbrace{- \frac{\sin(\nv)}{\nv} \py + \frac{\cos(\nv)}{\nv^2}  \vy - \frac{\sin(\nv)}{\nv^3} \vy }_{g(v)}\Big)\otimes v \\ &+ \left(-\frac{\sin(\nv)}{\nv^3}  + \frac{\cos(\nv)}{\nv^2} \right)v\otimes y 
		\end{align*}
		\begin{align*}
			\nabla g(v)\otimes v  & = \biggl(  \vy\left(  -  \frac{3\cos\nv}{\nv^4} -\frac{\sin\nv}{\nv^3} + \frac{3\sin\nv}{\nv^5}  \right) \\ &\qquad -\py \Big( \frac{\cos(\nv)}{\nv^2}  -  \frac{\sin(\nv)}{\nv^3}\Big)  \biggr) v \otimes v \\ &+ \left(-\frac{\sin(\nv)}{\nv^3}  + \frac{\cos(\nv)}{\nv^2} \right)y \otimes v \\
			\Rightarrow \Hess_{\overline\ell_y} (v) & = 2\Big( \frac{1}{1-(\langle\Exp_\origin v,y\rangle)^2} - \frac{\langle\Exp_\origin v,y\rangle \arccos(\langle\Exp_\origin v,y\rangle) }{(1-(\langle\Exp_\origin v,y\rangle)^2)^{3/2}} \Big) \times\\ &\nabla\phi (v) \otimes  \nabla\phi (v)  - 2 \left(\frac{\arccos(\langle\Exp_\origin v,y\rangle)}{\sqrt{1-(\langle\Exp_\origin v,y\rangle)^2}}   \right) \Hess_{\phi}(v).
		\end{align*}
		Therefore 
		\begin{align*}
			\Hess_{\overline\ell_y} (v) &= c_1\hat v\otimes \hat v+ c_2(\hat v\otimes y + y\otimes \hat v)+ c_3 y\otimes y+ c_4 \; \id,
		\end{align*}
		where $\hat v = v/\nv$ and $c_i = c_i(\nv, \vy, \py)$ are defined as
		\begin{align*}
			& c_1(\nv, \vy, \py) =  2\nv^2 A \eta^2  - 2B \big(  \vy\left(  -  \frac{3\cos\nv}{\nv^2} -\frac{\sin\nv}{\nv} + \frac{3\sin\nv}{\nv^3}  \right)+\\ & \hspace{13cm}+\nv\py \gamma \big)\\
			&c_2(\nv, \vy, \py)= - 2 A \eta  \sin(\nv)+ 2 B  \gamma; \quad c_3 (\nv, \vy, \py) = 2  A  \left(\frac{\sin(\nv)}{\nv} \right)^2 \\
			&c_4(\nv, \vy, \py)=  2 B \eta
      \intertext{where}
			&A = A(\nv, \vy, \py)= \frac {1}{1-\psi ^2} - \frac{\psi \arccos(\psi )}{(1-\psi ^2)^{3/2}}, \\ 
      &B = B(\nv, \vy, \py) = \frac{\arccos(\psi )}{\sqrt{1-\psi ^2}}, \\
			&\psi = \psi(\nv, \vy, \py) =  \frac{\sin(\nv)}{\nv} \vy + \cos(\nv)  \py, \\
			& \eta = \eta(\nv, \vy, \py) =  \left( \frac{\sin(\nv)}{\nv^2} - \frac{\cos(\nv)}{\nv} 
            \right)\frac{\vy}{\nv} +\frac{\sin(\nv)}{\nv}  \py, \\
			& \gamma = \gamma(\nv, \vy, \py)  = \frac{\sin(\nv)}{\nv^2} - \frac{\cos(\nv)}{\nv}.  
		\end{align*}
		
    Later results will seek to establish the positive-definiteness of $\Hess_{\ell_y}$ as an operator on $T_\origin \sphere$, and using \eqref{eq:project_Hessian}, it is enough to show that $\Hess_{\bar{\ell}_y}$ is positive definite as an operator on $\Yspace$. To help with this task, we will now express $\Hess_{\bar{\ell}_y}$ in another basis.
		Let
		$e_1 = y$ and $e_2 = \frac {v-\nv\alpha_1  e_1} {\nv\alpha_2}$ 
		where $\nv\, \alpha_1 = \langle v, y \rangle$ and $ \nv\, \alpha_2 =  \sqrt{\|v\|^2-\langle v, y \rangle^2}$. Notice that $e_1, e_2$ are orthogonal and of unit norm unless $\nv=|\langle v,y\rangle|$. 
        We shall for the time being assume that $\nv\neq|\langle v,y\rangle|$. 
		Complete $e_1, e_2$ into an orthonormal basis $(e_j)_{j \in J}$ of $\Yspace$, where $J = \{1, \ldots, \dim(\Yspace)\}$ if $\dim(\Yspace) < \infty$ and $J = \mathbb N$ otherwise, and let $E: \ell_2(J) \to \Yspace$ be the orthogonal operator mapping the $j$th canonical basis vector of $\ell_2(J)$ to $e_j$. 
        
		We can rewrite the Hessian as
		\begin{align}\label{eq:Hessian}
			\Hess_{\overline\ell_y} (v) = E \left(c_1 \begin{bsmallmatrix} 
				\alpha_1^2 & \alpha_1 \alpha_2 & 0 &  \dots  \\
				\alpha_1\alpha_2 & \alpha_2 ^2 &  0 & \dots  \\
				0 & 0 & 0 & \dots  \\ 
				\vdots & \vdots & \vdots & \ddots 
			\end{bsmallmatrix} + c_2 \begin{bsmallmatrix} 
				2 \alpha_1 & \alpha_2 & 0 &  \dots  \\
				\alpha_2 & 0 &  0 & \dots  \\
				0 & 0 & 0 & \dots  \\
				\vdots & \vdots & \vdots & \ddots 
			\end{bsmallmatrix} + c_3 \begin{bsmallmatrix} 
				1 & 0 &  \dots  \\
				0 & 0 & \dots  \\
				\vdots & \vdots & \ddots
			\end{bsmallmatrix} + c_4 \; \text{Id} \right) E^\adjoint
		\end{align}
		smoothly extending at $\|v\|=0$, and provided $\nv\neq|\langle v,y\rangle|$. 
        If $\nv=|\langle v,y\rangle|$,
        the Hessian reduces to
		\[
		(c_1 \alpha_1^2 + 2 c_2 \alpha_1 + c_3) y \tensor y + c_4 \id.
		\]
		
		\smallskip
		
		Accordingly, the spectrum of \(\Hess_{\overline\ell_y}(v)\)---and hence, by \eqref{eq:project_Hessian}, also the spectrum of \(\Hess_{\ell_y}(v)\)---depends
		on \(y\) and \(v\) only through the three scalar quantities
		\[
		\zeta=\sphereDist(\origin,y)=\arccos\, \langle y, \origin\rangle,
		\qquad
		\|v\|,
		\qquad
		\langle v,y\rangle.
		\]
\end{proof}
\end{Proposition}

\begin{Corollary}
	\label{corollary:ell_Hessian_decomposition}
	The spectrum of the Hessian $\Hess_{\ell_y}(v)$ depends on $y$ and $v$ only via $\zeta = \sphereDist(y, \origin)$ and 
	\begin{align}
		\label{eq:Phi_map}
		\Phi_y(v) = \begin{bmatrix}
			\langle v, u\rangle_\Yspace \\
      \\
			( \|v\|_\Yspace^2 - \langle v, u\rangle_\Yspace^2 )^{1/2} 
		\end{bmatrix}								 
	\end{align}
	with $u = \frac{\Log_{\origin} y}{\|\Log_{\origin} y\|}$  if  $y \notin \{\origin, - \origin\}$ and $u = \origin$ otherwise.
	In particular, also the smallest eigenvalue $\nu_{y}(v) \in \bar{\Real} = \Real \cup \{-\infty, \infty\}$ of $\Hess_{\ell_y}(v)$ depends on $y$ and $v$ only via $\zeta$ and $\Phi_y(v)$, such that we can write $\nu_y(v) = g_\zeta \circ \Phi_y(v)$ for $g_\zeta : \Real^2 \to \bar{\Real}$, 
	yielding the decomposition
	\begin{equation} 
		\label{eq:psiR2}
		\begin{tikzcd}
			T_\origin\sphere  \arrow[rr, "\nu_y"] \arrow[dr, "\Phi_y"', " "] & & \bar{\Real} \\
			& \Real^2 \arrow[ur, "g_\zeta"'] &
		\end{tikzcd}
	\end{equation}
\end{Corollary}	
\begin{proof}
		\item Given $\Phi_y(v)$ and $\langle \origin ,y\rangle = \cos \zeta$, we can reconstruct $(\|v\|, \langle v,y\rangle)$ as
		\begin{align*}
			\langle v, y \rangle &= [\Phi_y(v)]_1 \sqrt{1-\langle y, \origin  \rangle^2}, \\
			\|v\|^2& = [\Phi_y(v)]_1^2 + [\Phi_y(v)]_2^2,
		\end{align*}
		using that $\langle v, u \rangle = \langle v, \frac{y-\langle \origin ,y\rangle\origin}{\|y-\langle \origin ,y\rangle\origin\|} \rangle = \frac{\langle v, y \rangle}{\sqrt{1-\langle y, \origin  \rangle^2}}$.
		Hence, \(\Hess_{\ell_y}(v)\) can also be expressed in terms of $\zeta$ and $\Phi_y(v)$, and 
		the functions $g_\zeta$ can be computed by plugging $\|v\|, \langle v,y\rangle$ and $\zeta$ into \eqref{eq:Hessian} in the proof of Proposition~\ref{proposition:ell_Hessian_computation}.
\end{proof}

\begin{Remark}[Interpretation of $\Phi_y$]
	The first component $[\Phi_y(v)]_1$ is the coordinate of $v$ in the direction $\Log_\origin y$. 
    The second component reflects the part of $v$ orthogonal to $\Log_\origin y$, preserving the norm. 
    For the standard sphere $\sphere \subset \Real^3$, the map $\Phi_y$ simply represents $T_\origin\sphere$ in coordinates.
\end{Remark}

\begin{Lemma}
	\label{lemma:Phi_map}
	The map $\Phi_y: T_\origin \sphere \rightarrow \Real^2$, defined in Equation \eqref{eq:Phi_map}, 
	preserves certain balls of radius $\rho>0$ as follows:
	\begin{enumerate}[label=\roman*)]
		\item $v \in \ball_{T_\origin \sphere}(0, \rho)$ if and only if $\Phi_y(v) \in \ball_{\Real^2}(0, \rho)$,
		\item $v \in \ball_{T_\origin \sphere}(\Log_\origin y, \rho)$ if and only if $\Phi_y(v) \in \ball_{\Real^2}(\Phi_y(\Log_\origin y), \rho)$,
	\end{enumerate}
	where $\Phi_y(0) = 0$ and $\Phi_y(\Log_\origin y) = ( \zeta, 0)$ with $\zeta = \sphereDist(y, \origin)$.
\end{Lemma}
\begin{proof}
	First, note that $\|\Phi_y(v)\|_{\Real^2}^2 = \|v\|_{\Yspace}^2$ by construction. This yields the Statement i).
	Moreover,
	\begin{align*}
		\| \Phi_y(v) - \Phi_y(\Log_\origin y) \|^2 &= (\langle v, u\rangle - \zeta)^2 + (\|v\|^2 - \langle v, u\rangle^2) \\
		&= \langle v, u\rangle^2 - 2 \zeta \langle v, u\rangle + \zeta^2 + \|v\|^2 - \langle v, u\rangle^2 \\
		&= \|v\|^2 - 2 \|\Log_{\origin} y\| \langle v, \Log_{\origin} y \rangle / \|\Log_{\origin} y\| + \|\Log_{\origin} y\|^2\\
		&= \|v - \Log_{\origin} y \|^2
	\end{align*}
	for $\zeta = \|\Log_\origin y\| = \langle u, \Log_\origin y \rangle \neq 0$, and again $\| \Phi_y(v)\|^2 = \|v\|^2$ as in i).	
\end{proof}

\begin{Proposition}
	\label{proposition:convexity_plot}
	For $y, \origin \in \sphere$ with $\sphereDist(y, \origin) = \zeta$, let $\nu_y(v) = g_\zeta \circ \Phi_y(v)$ be the decomposition of the smallest eigenvalue of $\Hess_{\ell_y} (v)$ as described in Corollary~\ref{corollary:ell_Hessian_decomposition}.
	Then we have that 
	\begin{enumerate}[label=\roman*)]
		\item If for some $\rho>0$, $g_\zeta(a,b) > 0$ for all $(a,b) \in \ball_{\Real^2}(0, \rho)$, then $\ell_y$ restricted to $\ball_{T_\origin \sphere}(0, \rho)$ is convex, and the smallest eigenvalue of its Hessian is $\nu_y(v) \geq \epsilon$ for some $\epsilon > 0$. 
		\item If for  some $\rho>0$, $g_\zeta(a,b) > 0$ for all $(a,b) \in \ball_{\Real^2}((\zeta,0), \rho)$, then $\ell_y$ restricted to $\ball_{T_\origin \sphere}(\Log_\origin y, \rho)$ is convex, and the smallest eigenvalue of its Hessian is $\nu_y(v) \geq \epsilon$ for some $\epsilon > 0$. 

	\end{enumerate}
\end{Proposition}
\begin{proof}
	First, note that, since $\ball_{\Real^2}((a_0, b_0), \rho)$ is compact and $g_\zeta$ is continuous, there is a minimiser $(a^*,b^*) \in \ball_{\Real^2}((a_0, b_0), \rho)$ of $g_\zeta$, such that $0 < \epsilon := g_\zeta(a^*,b^*) \leq g_\zeta(a,b)$ for all $(a,b) \in \ball_{\Real^2}((a_0, b_0), \rho)$, under assumption of positivity of $g_\zeta(a,b)$.
	This, in combination with Statement i) and ii) of Lemma~\ref{lemma:Phi_map} yields that the eigenvalues of the Hessian $\Hess_{\ell_y}(v)$ on $\ball_{T_\origin \sphere}(0, \rho)$ and $\ball_{T_\origin \sphere}(\Log_{\origin} y, \rho)$, respectively, are bounded away from $0$ by $\epsilon$. This yields, together with the convexity of the balls, Statement i) and ii).
	
\end{proof}

\begin{Proposition}
  \label{proposition:analytic_convexity}	
  There exists an $\epsilon > 0$ such that, for any
  $y, \origin \in \sphere$ satisfying $\sphereDist(y, \origin) \leq \pi/4$, the restriction to $\ball_{T_\origin \sphere}(0, \pi/4)$ of the loss function $\ell_y$ is $\epsilon$-strongly convex, i.e.
   $\langle u, \Hess_{\ell_y}(v)\, u\rangle_\Yspace \geq \epsilon \Ynorm{u}^2$ for all $v\in \ball_{T_\origin \sphere}(0, \pi/4)$ and $u\in T_\origin \sphere$.
  \begin{proof}
    As in the proof of Proposition~\ref{proposition:ell_Hessian_computation}, we simply write $\|\cdot\|$ for $\|\cdot\|_\Yspace$ and $\langle\cdot,\cdot\rangle$ for $\langle\cdot,\cdot\rangle_\Yspace$.
    Though we provide below an analytical proof, we refer the reader to Proposition~\ref{proposition:convexity_plot} and Figure~\ref{figure:local_convexity} for convincing themselves quickly that the smallest eigenvalue of $\Hess_{\ell_y}(v)$ is bounded away from zero by $\epsilon>0$, and hence $\langle u, \Hess_{\ell_y}(v)\, u\rangle_\Yspace \geq \epsilon \Ynorm{u}^2$ and $\ell_y$ is, in particular, $\epsilon$-strongly convex on $\ball_{T_\origin \sphere}(0, \pi/4)$ since the ball is a convex set.

    Let us begin our analytical proof. 
    Using Proposition~\ref{proposition:ell_Hessian_computation}, $\Hess_{\ell_y}(v)$ can be reparametrised by three scalars $\n = \|v\|$, $\yo = \langle y, \origin\rangle$, and $\vangle = \frac{\langle v, y \rangle}{\n \sqrt{1 - \yo^2}}$ for $\omega \neq 0$ and $\yo \neq 1$. If $\yo \neq 1, \omega =0$ then $\vangle = 0$. The case $\yo=1$ is trivial since in that case $y = \origin \Rightarrow \ell_y(v) = \norm{v}^2$, which is a strongly convex function.
    Note that $\arccos \vangle$ is the angle between $v$ and $\Log_\origin y$, since $\Log_\origin y$ is a positive multiple of $y - \yo\,\origin$, which has norm $\sqrt{1 - \yo^2}$, and $\langle v, y - \yo\,\origin\rangle = \langle v, y\rangle$ for $v \in T_\origin \sphere$. Hence $(\n, \vangle)$ correspond to polar coordinates of the representation $\Phi_y(v)$ defined in Corollary~\ref{corollary:ell_Hessian_decomposition}. By assumption, $\rho = \cos \sphereDist(y, \origin) \in [\cos(\pi/4), \cos(0)] = [1/\sqrt{2}, 1]$. 
    Since $v \in \ball_{T_\origin \sphere}(0, \pi/4)$, these imply that $(\n, \yo, \vangle)$ belongs to the compact set  $\dom := [0, \pi/4] \times [1/\sqrt{2}, 1] \times [-1,1]$.

    Using \eqref{eq:Hessian}, showing that $\Hess_{\ell_y}(v)$ is positive definite boils down to showing that the following $(\n, \yo, \vangle)$-dependent matrix is positive definite:
    \begin{equation}
      \label{eq:Hessian3d}
      \widetilde H(\n, \yo, \vangle) := \left(c_1 \begin{bmatrix} 
          \alpha_1^2 & \alpha_1 \alpha_2 & 0 \\
          \alpha_1\alpha_2 & \alpha_2 ^2 & 0 \\
          0 & 0 & 0
          \end{bmatrix} + c_2 \begin{bmatrix} 
          2 \alpha_1 & \alpha_2 & 0 \\
          \alpha_2 & 0          & 0 \\
          0 & 0 & 0
          \end{bmatrix} + c_3 \begin{bmatrix} 
          1 & 0 & 0 \\
          0 & 0 & 0 \\
          0 & 0 & 0
        \end{bmatrix}
      \right)
      + c_4 \begin{bmatrix} 
        1 & 0 & 0 \\
        0 & 1 & 0 \\
        0 & 0 & 1
      \end{bmatrix},
    \end{equation}
    where the $c_j$s and the $\alpha_j$s are defined in Proposition~\ref{proposition:ell_Hessian_computation} but have been reparametrised here to be functions of $(\n, \yo, \vangle)$.
    $\widetilde H$ is positive definite only if  $c_4 > 0$. If that's the case, it is then enough to show the positive definiteness of the matrix
    \begin{equation}
      \label{eq:Hessian2d}
      H(\n, \yo, \vangle) := c_1 \begin{bmatrix} 
        \alpha_1^2 & \alpha_1 \alpha_2 \\
        \alpha_1\alpha_2 & \alpha_2 ^2 \\
        \end{bmatrix} + c_2 \begin{bmatrix} 
        2 \alpha_1 & \alpha_2 \\
        \alpha_2 & 0 \\
        \end{bmatrix} + c_3 \begin{bmatrix} 
        1 & 0 \\
        0 & 0 \\
        \end{bmatrix} + c_4 \begin{bmatrix} 
        1 & 0 \\
        0 & 1 \\
      \end{bmatrix} 
    \end{equation}
    to show that the smallest eigenvalue of $\widetilde H$ is strictly positive pointwise.
    This can be done by showing $\trace H > 0$ and $\det H > 0$. Indeed, $H$ is a $2 \times 2$ matrix, and $\trace H > 0, \det H >0$ imply that the sum and product of the two eigenvalues of $H$ are strictly positive, and hence each eigenvalue is strictly positive. Then, since $\dom$ is compact and eigenvalues of $\widetilde H$ are continuous in $(\n, \yo, \vangle)$, we get that the smallest eigenvalue of $\widetilde H$ over $\dom$ is not smaller than some $\epsilon > 0$.

    The remainder of our proof is structured as follows.
    \begin{enumerate}
      \item[Part A.] We construct piece-wise polynomial lower bounds $\underline{c}_4(\n, \yo, \vangle)$, $\underline{\trace H}(\n, \yo, \vangle)$ and $\underline{\det H}(\n, \yo, \vangle)$ for $c_4$, $\trace H$ and $\det H$, respectively.
      \item[Part B.] We show that $\underline{c}_4(\n, \yo, \vangle)>0$, $\underline{\trace H}(\n, \yo, \vangle)>0$ and $\underline{\det H}(\n, \yo, \vangle)>0$ for $(\n, \yo, \vangle) \in \dom$.
    \end{enumerate}

    We now provide the details for each part.

    \emph{Part A.} Here, we construct lower bounds $\underline{c}_4$, $\underline{\trace H}$ and $\underline{\det H}$ in such a way that we can show in Part B that they are positive on $\dom$. The challenge is to keep the polynomial degree as small as possible to keep the minimisation in Part B computationally feasible.

    According to Representation~\eqref{eq:Hessian}, $c_4$, $\trace H$ and $\det H$ can be written as algebraic combinations of $\n, \yo, \vangle$ and the following univariate functions
    \begin{align*}
      \h(\n) &= \frac{\sin \n}{\n} \qquad \text{ for } \n \neq 0 \text{ and } \h(0) = 1,\\
      \hh(\n) &= \frac{\sin \n}{\n^2} - \frac{\cos \n}{\n} \qquad \text{ for } \n \neq 0 \text{ and } \hh(0) = 0,\\
      s(\yo) &= \sqrt{1 - \yo^2}, \\
      a(\psi) &= \frac{1}{1-\psi^2} - \frac{\psi\arccos\psi}{(1-\psi^2)^{3/2}} \qquad \text{ for } \psi \neq 1 \text{ and } a(1) = 1/3,\\
      b(\psi) &= \frac{\arccos\psi}{\sqrt{1-\psi^2}} \qquad \text{ for } \psi \neq 1 \text{ and } b(1) = 1.
    \end{align*}
    It can be individually verified that they are continuous and non-negative on their domains of interest, with $\h,s,a,b$ monotonically decreasing, and $\hh$ monotonically increasing. 

    For each of these functions, we define piece-wise polynomial upper and lower bounds on $(\n, \yo, \vangle) \in [0, 0.8] \times [0.7, 1] \times [-1,1] =: \domm \supset \dom$ 
    with the convention that an index $\sigma<0$---or simply ``$-$'' as index---denotes a lower bound and an index $\sigma\geq 0$---or simply ``$+$'' as index---denotes an upper bound. For $\h$, we use piecewise constant bounds
    \begin{align*}
      \h_{-}(\n) &= \begin{cases} 0.9735, & \n \le 0.4,\\ 0.9310, & 0.4<\n\le 0.65,\\ 0.8966, & 0.65<\n\le 0.8,\end{cases}
                 &
      \h_{+}(\n) &= \begin{cases} 1.0000, & \n \le 0.4,\\ 0.9736, & 0.4<\n\le 0.65,\\ 0.9311, & 0.65<\n\le 0.8.\end{cases}
    \end{align*}
    For $\hh$, we use two systems of bounds: affine linear bounds
    \begin{align*}
      \hh_-(\n) &= \frac{249}{800}\,\n, &
      \hh_+(\n) &= 0.225 + 0.31\,(\n - 0.7)
    \end{align*}
    and also piecewise constant bounds
    \begin{align*}
      \tilde{\hh}_-(\n) &= \begin{cases} 0.0000, & \n \le 0.4,\\ 0.1312, & 0.4<\n\le 0.65,\\ 0.2076, & 0.65<\n\le 0.8,\end{cases} &
      \tilde{\hh}_+(\n) &= \begin{cases} 0.1313, & \n \le 0.4,\\ 0.2077, & 0.4<\n\le 0.65,\\ 0.2500, & 0.65<\n\le 0.8.\end{cases}
    \end{align*}
    Hence, also bounds for functions depending on it will vary with the choice of the bounds on $\hh$. Yet, we largely suppress this dependence in the notation, making it explicit only when necessary.

    Based thereon, we derive bounds on $\cos \n $ as
    \begin{align*}
      \cos_-(\n) &= \h_-(\n) - \n\,\hh_+(\n), &
      \cos_+(\n) &= \h_+(\n) - \n\,\hh_-(\n).
    \end{align*}

    For $s$, we use piecewise linear bounds
    \begin{align*}
      s_{-}(\yo) &=
      \begin{cases}
        0.71 + (0.39-0.71)/(0.92-0.7) (\yo-0.7) & \yo < 0.92,\\
        0.39 + (0-0.39)/(1-0.92) (\yo-0.92)     & \yo \ge 0.92,
      \end{cases}
      \\[2pt]
      s_{+}(\yo) &=
      \begin{cases}
        0.73 + (0.50-0.73)/(0.92-0.7) (\yo-0.7)  & \yo < 0.92,\\
        0.50 + (0.10-0.50)/(1-0.92) (\yo-0.92)   & \yo \ge 0.92.
      \end{cases}
    \end{align*}

    With that, we obtain piece-wise polynomial bounds on $\psi$, $\eta$ and $\gamma$, defined
    in Proposition~\ref{proposition:ell_Hessian_computation}, as
    \begin{align*}
      \psi_{\sigma}(\n, \yo, \vangle)
      &= \h_{\theta\sigma}(\n)\,\n\,\vangle\,s_{\theta\sigma}(\yo)
      + \cos_{\sigma}(\n)\,\yo, \\
      \eta_{\sigma}(\n, \yo, \vangle)
      &= \hh_{\theta\sigma}(\n)\,\vangle\,s_{\theta\sigma}(\yo)
      + \h_{\sigma}(\n)\,\yo, \\
      \gamma_{\sigma}(\n, \yo, \vangle)
      &= \hh_{\sigma}(\n),
    \end{align*}
    for $\sigma\in\{-1,+1\}$ identifying lower bounds (negative sign of the index) and upper bounds (non-negative sign of the index).
    \smallskip

    Moreover, we individually bound $\mathrm{br}(\n) = 3 \hh(\n) - \n\, \h(\n)$ by
    \begin{align*}
      \mathrm{br}_{-}(\n) &= 0,
                          &
      \mathrm{br}_+(\n) &= \begin{cases} 0.0043, & \n \le 0.4,\\ 0.0178, & 0.4<\n\le 0.65,\\ 0.0327, & 0.65<\n\le 0.8.\end{cases}
    \end{align*}

    With $a$ and $b$ as defined above, we have
    $A(\n, \yo, \vangle) = a(\psi(\n, \yo, \vangle))$ and $B(\n, \yo, \vangle) = b(\psi(\n, \yo, \vangle))$.
    Depending on the sign of $\vangle$, we bound the function $a(\psi)$ by
    \begin{align*}
      a_{-}(\psi, \vangle)  &= 
      \begin{cases} 
        \tfrac{29}{50} - \tfrac{2667}{10000}\,\psi,  & \vangle \geq 0,\\
        \tfrac{20861}{25000} - \tfrac{416}{625}\,\psi, & \vangle < 0 
      \end{cases} \\
      a_{+}(\psi, \vangle)  &= 
      \begin{cases} 
        \tfrac{25847}{28700} - \tfrac{145}{287}\,\psi, & \vangle \geq 0,\\
        \tfrac{6697}{6450} - \tfrac{175}{258}\,\psi, & \vangle < 0. 
      \end{cases} 
    \end{align*}
    and the function $b(\psi)$ by
    \begin{align*}
      b_{-}(\psi, \vangle)  &= 
      \begin{cases}
        \tfrac{6567}{5000} - \tfrac{1667}{5000}\,\psi,  & \vangle \geq 0,\\
        \tfrac{74009}{50000} - \tfrac{2941}{5000}\,\psi, & \vangle < 0 
      \end{cases} \\
      b_{+}(\psi, \vangle)  &= 
      \begin{cases}
        \tfrac{1439}{1030} - \tfrac{153}{515}\,\psi,   & \vangle \geq 0,\\
        \tfrac{1703}{1060} - \tfrac{59}{106}\,\psi,   & \vangle < 0 
      \end{cases} \\
      \tilde{b}_{+}(\psi, \vangle)  &= 
      \begin{cases}
        \tfrac{1931}{1260} - \tfrac{109}{252}\,\psi,   & \vangle \geq 0,\\
        \tfrac{4109}{2580} - \tfrac{70}{129}\,\psi, & \vangle < 0. 
      \end{cases} 
    \end{align*}
    where an alternative upper bound $\tilde{b}_{+}(\psi, \vangle)$ is given for the case when $\hh$ is bounded by the piecewise constant $\tilde{\hh}_\sigma$.

    The validity of each of the univariate bounds can be individually verified for $(\n, \yo, \vangle) \in \domm$ and $\psi$ in the range of the respective bound  $\psi_{-\sigma}(\n, \yo, \vangle)$ for $a_\sigma$ and $b_\sigma$, $\sigma\in\{-1,+1\}$, where the change of signs is due to the fact that $a$, $b$ and their bounds are monotonically decreasing in $\psi$.
    Accordingly, we obtain bounds on $A$ and $B$ as
    \begin{align*}
      A_{\sigma}(\n, \yo, \vangle) &= a_{\sigma}\bigl(\psi_{-\sigma}(\n, \yo, \vangle),\, \vangle\bigr), &
      B_{\sigma}(\n, \yo, \vangle) &= b_{\sigma}\bigl(\psi_{-\sigma}(\n, \yo, \vangle),\, \vangle\bigr).
    \end{align*}

    Combining all previous bounds, we can bound the functions $c_1, \dots, c_4$ by 
    \begin{align*}
      c_{4,\sigma} &= 2\, B_\sigma\, \eta_\sigma,\\
      c_{3,\sigma} &= 2\, A_\sigma\, \h_{\sigma}^2,\\
      c_{2,\sigma}   &=  2\, B_\sigma\, \gamma_\sigma  -2\,\n\, A_{-\sigma}\, \eta_{-\sigma}\, \h_{-\sigma},\\
      c_{1,\sigma}   &= c_{1,\sigma}^+ + c_{1,\sigma}^-, \\ 
      c_{1,\sigma}^+ &= 2\n^2 A_\sigma\,\eta_\sigma^2 - 2\,\mathbf{1}_{\theta<0}\,\theta\,B_\sigma\, s_{\sigma}\,\mathrm{br}_\sigma, \\ 
      c_{1,\sigma}^- &= -2 B_{-\sigma}\,\rho\,\n\,\gamma_{-\sigma}
      - 2\,\mathbf{1}_{\theta>0}\,\theta\, B_{-\sigma}\, s_{-\sigma}\,\mathrm{br}_{-\sigma},
    \end{align*}
    where again the sign of $\sigma$ in the index indicates the direction of the bound.
    Here, $c_1$ is subdivided into a non-negative part $c_1^+$ and a non-positive part $c_1^-$, given by
    \begin{align*}
      c_1^+(\n,\rho,\theta) &=
      2\n^2 A\,\eta^2
      - 2\, \mathbf{1}_{\theta<0}\,\theta\,B\, s(\rho)\,\mathrm{br}(\n)  \ \ge 0,
      \\[2pt]
      c_1^-(\n,\rho,\theta) &=
      -2 B\,\rho\,\n\,\gamma
      - 2\, \mathbf{1}_{\theta>0}\,\theta\, B\, s(\rho)\,\mathrm{br}(\n) \ \le 0,
    \end{align*}
    such that $c_1 = c_1^+ + c_1^-$.

    With $\alpha_1 = \vangle\, s(\yo)$ and $\alpha_2^2 = 1 - \vangle^2 (1-\yo^2)$,
    expansion of \eqref{eq:Hessian2d} gives
    \begin{equation*}
      \trace H = c_1 + 2\vangle\, c_2 s(\yo) + c_3 + 2 c_4
    \end{equation*}
    and
    \begin{align*}
      \det H &= c_4\,(c_1 + 2\,\vangle\, c_2\, s(\yo) + c_3 + c_4)
      \;+\; \alpha_2^2\bigl(c_1 c_3 - c_2^2\bigr) \\
             &= c_4 \bigl(\trace H - c_4\bigr) + \alpha_2^2 \bigl(c_1 c_3 - c_2^2\bigr).
    \end{align*}

    Now, the desired lower bounds are obtained as 
    \begin{align*}
      \underline{c}_4 &= c_{4-} \\
      \underline{\trace H} &= c_{1,-} + 2\vangle\, c_{2,-\vangle}\, s_{-\vangle}(\yo) + c_{3,-} + 2 c_{4,-} \\
      \underline{\det H} &=  c_{4,-}\, \bigl(\underline{\trace H} - \underline{c}_4\bigr) + \alpha_2^2 \bigl(\tilde{c}_{1,-}^+ \tilde{c}_{3, -} + \tilde{c}_{1,-}^- \tilde{c}_{3, +} - \tilde{c}_{2,+}^2 - \iota^2\bigr)
    \end{align*}
    where \begin{itemize}
      \item the bound on $\underline{\det H}$ is only valid after showing $c_{4-} > 0$, 
      \item the ``$\tilde{c}$'' in the second term of $\underline{\det H}$ indicate that here the piecewise constant $\tilde{\hh}_\sigma$ are used to bound $\hh$, while elsewhere we use the affine linear bounds $\hh_\sigma$. 
      \item we use $\trace H - c_4 = c_1 + 2\vangle\, c_2 s(\yo) + c_3 + c_4$ to see that $\underline{\trace H} - \underline{c}_4$ is a lower bound (rather than only $\underline{\trace H} - c_{4,+}$),
      \item we set $\iota = \min_{\domm}(c_{2,-})$ to guarantee that $c_{2,+}^2 + \iota^2 \geq c_2^2$ irrespective of the sign, using that $c_{2,+}>c_{2,-}$.
        The last choice is generally valid, but motivated by the fact that we expect $\iota$ to be a small negative number.
    \end{itemize}
    \bigskip

    \emph{Part B.} We use the implementation of the cylindrical algebraic decomposition (CAD) algorithm for exact global optimisation in Mathematica~14.3 for minimisation/maximisation of the resulting piecewise polynomial lower bounds with rational coefficients \citep{Mathematica143}, computing 
    \begin{align*}
      \min_{\domm} \underline{c}_4 &> 0 \\
      \min_{\domm} \bigl( \underline{\trace H} - \underline{c}_4 \bigr) &> 0\\
      - \min_{\domm} \bigl( \tilde{c}_{1,-}^+ \tilde{c}_{3, -} + \tilde{c}_{1,-}^- \tilde{c}_{3, +} - \tilde{c}_{2,+}^2 - \iota^2 \bigr) &< (\min_{\domm} \underline{c}_4) \min_{\domm} \bigl( \underline{\trace H} - \underline{c}_4 \bigr)
    \end{align*}
    The first two inequalities show that $c_4>0$ and $\trace H>0$ on $\dom$.  
    Since $\alpha_2^2 \in [0,1]$, $\underline{c}_4 > 0$ and $\bigl( \underline{\trace H} - \underline{c}_4 \bigr) > 0$, the last inequality shows that
    \begin{align*}
      \det H &\geq \underline{\det H} \\
             &\geq (\min_{\domm} \underline{c}_4) \min_{\domm} \bigl( \underline{\trace H} - \underline{c}_4 \bigr) + \alpha_2^2 \min_{\domm}\bigl(\tilde{c}_{1,-}^+ \tilde{c}_{3, -} + \tilde{c}_{1,-}^- \tilde{c}_{3, +} - \tilde{c}_{2,+}^2 - \iota^2\bigr) \\
             &> 0
    \end{align*}
    which concludes the proof. The Mathematica code and output can be found in the \thesupplement.
  \end{proof}
\end{Proposition}

\section{Additional details on the simulation study}

\subsection{Construction of the conditional mean} \label{sec:simulation_mu}
To construct $\mu$ in a principled way, we proceed as follows: as an initial template for $\mu$, we use $\Psi: x \mapsto \psi_x^{1/2}$ given by a mixture density
    \begin{equation*}
        \psi_x(z) = w_1(x)\phi\left(\frac{z-m_1(x)}{s_1(x)}\right) + w_2(x) \phi\left(\frac{z-m_2(x)}{s_2(x)}\right) 
    \end{equation*}
    of Gaussians truncated to $z\in[0,1]$, where $\phi$ is the standard Gaussian probability density function, and  $m_1, m_2 :\Xspace \rightarrow \Real$, $s_1, s_2:\Xspace \rightarrow \Real_+$, $w_1, w_2:\Xspace \rightarrow \Real_+$ are suitable location, scale, and weight functions respectively. We choose the GRBF kernel with length-scale parameter $\sigma$, then we fit the spherical kernel ridge regression model to a set $\{\Psi(\boldsymbol{x_i})\}$, for covariates $\boldsymbol{x}_i \in [0,1]^2$, $i = 1,\dots, n$ on a regular grid, fixing $\origin$ to the spherical Fréchet mean of $\{\Psi(\boldsymbol{x_i})\}$. This gives a set of regression coefficients $\tilde \xi_1, \dots, \tilde \xi_n \in T_\origin \sphere$. To obtain a prescribed model variance level $\tau_\mu^2 = (2\pi/5)^2$, we rescale $\tilde{\xi}_i$ by their variance so that the resulting $\xi_i$ have this fixed model variance. This guarantees that the data are sufficiently spread out on the sphere, in a domain of radius larger than $\pi/2$ around $\origin$. The conditional mean function is then given by
    $$\mu(\cdot) = \Exp_\origin\left( \sum_{i=1}^n k (x_i,\cdot ) \xi_i \right).$$ 
    Figure \ref{fig:simulation:illustration} shows that the resulting $\mu(\boldsymbol{x}_i)$ exhibits spatially structured variation in both location and shape, originating from the Gaussian mixture.

\section{Sobolev spaces}
\label{sec:sobolev}
\subsection*{Introduction to the real-valued case}

Let $\Omega$ be an open set in $\Real^s$ and $\mathcal{C}^m(\Omega)$ the set of $m$ times continuously differentiable functions $\Omega \rightarrow \Real$. 
Let $|\cdot|$ denote the Euclidean norm on $\Real^s$. 
Let further $\mathcal{C}^m_c(\Omega) = \{\psi \in \mathcal{C}^m(\Omega): \psi(x) = 0, x \notin \Psi, \Psi \subset \Omega \text{ compact}\}$ be those with compact support.
For a nonnegative integer multi-index $\iota = (\iota_1,\dots,\iota_s)$, we write $|\iota| = \sum_{i = 1}^s \iota_i$ and $D^{\iota} \psi = \frac{\partial^{|\iota|} \psi}{\partial^{\iota_1} x_1 \dots \partial^{\iota_s} x_s}$ for $ \psi \in \mathcal{C}^m(\Omega)$.

Define the weak $\iota$th derivative of a function $f\in L^2(\Omega) := L^2(\Omega, \mathbb{R})$ as the element $D^\iota f$ of $L^2(\Omega)$ satisfying 
$$
    \int D^\iota f(x) \psi(x)\, dx= (-1)^{|\iota|} \int f(x) D^{\iota} \psi(x)\, dx
$$
for all smooth functions $\psi \in \mathcal{C}_c^\infty(\Omega)$. 
If such a $D^\iota f$ exists, it is unique and $f$ is called weakly differentiable. If $f\in \mathcal{C}^m(\Omega)$, its weak derivatives coincide with the corresponding derivatives. 

The Sobolev space $H^m(\Omega)$ is defined as 
$$
    H^m(\Omega) = \{ f \in L^2(\Omega) : D^\iota f \in L^2(\Omega), |\iota| \leq m \}
$$
with inner product 
$$
	\langle f, g \rangle_{H^m(\Omega)} = \sum_{\iota: |\iota| \leq m} \int D^\iota f(x) D^\iota g(x) \, dx
$$
where usually the integer $m$ is assumed positive, while $H^0(\Omega) = L^2(\Omega)$. 

As usual let $\bar{\Omega}$ denote the closure of $\Omega$, and define $\mathcal{C}^{m}(\bar{\Omega}) = \{\psi \in \mathcal{C}^{m}(\Omega) : D^{\iota} \psi = \left.\phi\right|_{\Omega}, \phi \in \mathcal{C}(\bar{\Omega}),0\leq |\iota|\leq m\}$ where $\left.\phi\right|_{\Omega}: \Omega \rightarrow \Real, x\mapsto \phi(x)$ denotes the restriction of function $\phi$ to $\Omega$.

\begin{Proposition}[cf.\ \cite{Brezis2011FunctionalAnalysis}, Corollary 8.10]
	\label{prop:product_sobolev}
	Let $I \subset \Real$ be an open interval, $k \geq 1$ an integer, and $u, v \in H^k(I)$. Then $uv \in H^k(I)$ and $D(uv) = (Du)\, v + u\, Dv$.
	\begin{proof}
		For $k = 1$, this is \citet[Corollary 8.10]{Brezis2011FunctionalAnalysis}. For $k \geq 2$, we proceed by induction on $k$: by the case $k = 1$, $uv \in H^1(I)$ with $D(uv) = (Du)\, v + u\, Dv$, where $(Du)\, v$ and $u\, Dv$ are products of functions in $H^{k-1}(I)$ and hence lie in $H^{k-1}(I)$ by the induction hypothesis. Thus $D(uv) \in H^{k-1}(I)$, that is, $uv \in H^k(I)$.
	\end{proof}
\end{Proposition}

\begin{Proposition}[cf.\ \cite{Brezis2011FunctionalAnalysis}, Corollary 8.11]
	\label{prop:sobolev_chain_rule}
	Let $I \subset \Real$ be a bounded open interval, $d \geq 1$ an integer, $g\in\Ccal^1(\Real^d)$, and $f_1, \ldots, f_d \in H^1(I)$, $f := (f_1, \ldots, f_d)$. Then $g\circ f \in H^1(I)$ and $D(g\circ f) = \sum_{j=1}^d (\partial_j g\circ f)\, D f_j$.
	\begin{proof}
		For $d = 1$, this is \citet[Corollary 8.11]{Brezis2011FunctionalAnalysis}, whose assumption $g(0) = 0$ is not needed for bounded $I$ (see footnote 9 there). For general $d$, we proceed in close analogy to its proof, spelling it out in somewhat more detail: By \citet[Theorem 8.7]{Brezis2011FunctionalAnalysis}, applied to each $f_j$, there are $f^{[n]} = (f^{[n]}_1, \ldots, f^{[n]}_d)$ with $f^{[n]}_j \in \Ccal_c^\infty(\Real)$ such that $f^{[n]}_j|_I \to f_j$ in $H^1(I)$, and hence also in $L^\infty(I)$ by \citet[Theorem 8.8]{Brezis2011FunctionalAnalysis}. In particular, there is $R > 0$ such that $|f(x)| \leq R$ and $|f^{[n]}(x)| \leq R$ for almost every $x\in I$ and all $n$. As $g$ and its partial derivatives $\partial_j g$ are bounded and uniformly continuous on the compact ball $\{z \in \Real^d : |z| \leq R\}$, the functions $g \circ f$ and $\partial_j g \circ f$ are bounded, i.e., in particular in $L^\infty(I)$. Moreover, the uniform continuity lets the $L^\infty(I)$-convergence $f^{[n]}\to f$ carry over to $g \circ f^{[n]} \to g \circ f$ and $\partial_j g \circ f^{[n]} \to \partial_j g \circ f$. In particular, $g \circ f \in L^\infty(I) \subset L^2(I)$, as $I$ is bounded, and $\sum_{j} (\partial_j g\circ f)\, D f_j \in L^2(I)$. Let $\varphi \in \Ccal_c^\infty(I)$. By the classical chain rule, $\int (g\circ f^{[n]})\, D\varphi\, dx = - \int \sum_{j} (\partial_j g\circ f^{[n]})\, D f^{[n]}_j\, \varphi\, dx$. Different applications of H\"older's inequality yield that, as $n \to \infty$, the left-hand side tends to $\int (g\circ f)\, D \varphi\, dx$, and the right-hand side to $-\int \sum_{j} (\partial_j g\circ f)\, D f_j\, \varphi\, dx$, since $g \circ f^{[n]} \to g \circ f$ and $\partial_j g \circ f^{[n]} \to \partial_j g \circ f$ in $L^\infty(I)$ and $D f^{[n]}_j \to D f_j$ in $L^2(I)$. Hence $D(g \circ f) = \sum_{j} (\partial_j g\circ f)\, D f_j$, and $g \circ f \in H^1(I)$.
	\end{proof}
\end{Proposition}

\begin{Proposition}[\cite{Brezis2011FunctionalAnalysis}, Corollary 9.15]
	\label{prop:Sobo2Diff}
	Let $I \subset \Real$ be a bounded open interval, and let $m \geq 1$ be an integer. Then every $f \in H^m(I)$ has a unique continuous representative, with which $f$ is identified in the following. In this sense, $H^m(I) \subset \Ccal^{m-1}(\bar{I})$, and the classical derivatives of $f$ of order $j \leq m-1$ coincide almost everywhere with the weak derivatives $D^j f$.
	\begin{proof}
		The representative is unique, as two continuous functions that agree almost everywhere on $I$ agree on $I$. Its existence in $\Ccal^{m-1}(\bar{I})$ is the case $N = 1$, $p = 2$ of \citet[Corollary 9.15 and footnote 16]{Brezis2011FunctionalAnalysis}: a bounded interval is of class $\Ccal^1$ in the sense of \citet[Section 9.2]{Brezis2011FunctionalAnalysis}, and as $m - 1/2$ is never an integer, the order of differentiability given there is the integer part of $m - 1/2$, that is, $m-1$. As noted above for $\Ccal^m$ functions, the classical derivatives of order $j \leq m-1$ of this representative coincide with its weak derivatives.
	\end{proof}
\end{Proposition}

\begin{Theorem}[\cite{BerlinetThomas-Agnan2004ReproducingKernelHilbertSpaces}, Theorem 121]
	\label{thm:SobolevRKHS}
	For an open set $\Omega \subset \Real$, $H^m(\Omega)$ is a reproducing kernel Hilbert space if and only if $m \geq 1$. 
	$H^m(\Real^s)$ is a reproducing kernel Hilbert	space if and only if $m > s/2$.
\end{Theorem}

The form of the kernels of $H^m(\Real^s)$ is provided by \cite{Novak2018reproducing}.

\subsection*{Introduction to the vector-valued case}

Links between Vector-valued Sobolev spaces and VVRKHS are, for instance, discussed by \citet{carmeliVectorValuedReproducing2010} and \citet{liOptimalSobolevNorm2024}.
Directly generalising the previous definition, a vector-valued Sobolev space is defined as
\begin{equation*}
	H^m(\Omega, \Yspace) = \{ f \in L^2(\Omega, \Yspace) : D^{\iota} f \in L^2(\Omega, \Yspace), |\iota|\leq m \}
\end{equation*} 
with inner product
\begin{equation*}
	\langle f, g \rangle_{H^m(\Omega, \Yspace)} = \sum_{\iota: |\iota| \leq m} \int \langle D^\iota f(x),  D^\iota g(x) \rangle_{\Yspace} \, dx.
\end{equation*}

For the precise definition of weak derivatives $D^{\iota}$ of vector-valued functions, see, e.g., \citep{Aubin2000AppliedFunctionalAnalysis}. We can directly give an equivalent definition of the vector-valued Sobolev spaces as follows.

\begin{Theorem}[\cite{Aubin2000AppliedFunctionalAnalysis}, Theorem 12.7.1]
	\label{thm:SobolevTensorProduct}
	The Sobolev space $H^m(\Omega, \Yspace)$ is isometric to the Hilbert space tensor product $H^m(\Omega) {\otimes} \Yspace$.
\end{Theorem}

Given that a VVRKHS with kernel $K(x,x')=k(x,x') \id$ defined on $\Yspace$ is isometric to $\mathcal{H}_k{\otimes} \Yspace$, with $\mathcal{H}_k$ the RKHS of $k$, the characterisation of Sobolev spaces as RKHS directly carries over to the vector-valued case:  

\begin{Corollary}
	If $H^m(\Omega)$ is an RKHS with kernel $k$, then $H^m(\Omega, \Yspace)$ is a vector-valued RKHS with kernel $K (x,x') = k(x,x') \id$. 
\end{Corollary}

Let $\Ccal^m(\Omega, \Yspace)$ denote the set of $m$ times continuously differentiable functions from $\Omega$ to $\Yspace$, and $\Ccal^m(\bar{\Omega}, \Yspace)$ the subset of functions that can be continuously extended to $\bar{\Omega}$ such that also all their derivatives can be continuously extended. 

\begin{Corollary}
	\label{cor:Sobo2Diff}
	Let $\dim(\Yspace) < \infty$, let $I \subset \Real$ be a bounded open interval, and let $m \geq 1$ be an integer. Then 
	$H^m(I,\Yspace) \subset \Ccal^{m-1}(\bar{I},\Yspace)$.
\end{Corollary}

\subsection*{Sobolev spaces and smooth manifolds}

\begin{Proposition}
	\label{proposition:continuous_inside_sobolev}
	Let $\Omega\subset\Real^s$ be open and bounded, then 
    $\mathcal{C}^{m}(\bar{\Omega}, \Yspace) \subset H^{m}(\Omega, \Yspace)$.
	\begin{proof}
		Since  the derivatives $D^{\iota} f$, $|\iota|\leq m$, of $f \in \mathcal{C}^{m}(\bar{\Omega}, \Yspace)$ exist by definition, we only need to show that they have finite square integral. By definition, the $D^{\iota} f$ can be extended to continuous functions on $\bar{\Omega}$ with continuous $x \mapsto \|D^{\iota} f(x)\|_{\Yspace}^2$ on $\bar{\Omega}$. Since $\bar{\Omega}$ is compact, 
		there exists an $x^{\iota}$ with $\|D^{\iota} f\|_{\Yspace}^2(x^{\iota}) \geq \|D^{\iota} f\|_{\Yspace}^2(x)$ for all $x \in \Omega$ and $\int \|D^{\iota} f(x)\|_{\Yspace}^2 \, dx \leq \|D^{\iota} f(x^\iota)\|_{\Yspace}^2 \int 1\, dx < \infty$.
	\end{proof}
\end{Proposition}

\begin{Corollary}
    \label{cor:continuous_mu}
	Let $\mathcal{M}$ be a differentiable manifold over a Hilbert space $\Yspace$, let $\mu: \Xspace \rightarrow \mathcal{M}$, $\Xspace \subset \mathbb{R}^s$, be $m$ times continuously differentiable, and let $\psi: \mathcal{N} \rightarrow \Ucal$ be a chart from an open set $\mathcal{N} \subset \mathcal{M}$ onto an open set $\Ucal \subset \Yspace$, i.e.\ a diffeomorphism from $\mathcal{N}$ to $\Ucal$. 
Furthermore, let $\Omega \subset \Xspace \subset \Real^s$ be non-empty, bounded, open, and such that $\bar{\Omega} \subset \mu^{-1}(\mathcal{N})$.
	Then $\left.\psi \circ \mu\right|_{\Omega} \in H^m(\Omega, \Yspace)$.
    Moreover, if there is an $x_0\in \mathrm{interior}(\Xspace)$ with $\mu(x_0) \in \mathcal{N}$, such a set $\Omega$ exists.
    
	\begin{proof}
    If $\mu(x_0) \in \mathcal{N}$ for some $x_0\in\Xspace$, i.e.\ $\mu^{-1}(\mathcal{N}) \neq \emptyset$, there is an $\epsilon>0$ such that $B_{\epsilon}(x_0) \subset \mu^{-1}(\mathcal{N})$, since the continuity of $\mu$ implies that $\mu^{-1}(\mathcal{N})$ is open. Hence, a possible non-empty choice is $\Omega = B_{\epsilon/2}(x_0)$.\newline
		Since $\mu^{-1}(\mathcal{N})$ is open, and by differentiability of $\mu$, we have $\left.\psi \circ \mu \right|_{\mu^{-1}(\mathcal{N})} \in \Ccal^m(\mu^{-1}(\mathcal{N}), \Yspace)$. Hence, $\left.\psi \circ \mu \right|_{\Omega} \in \Ccal^m(\bar{\Omega}, \Yspace)$ and belongs to $H^m(\Omega, \Yspace)$ by Proposition \ref{proposition:continuous_inside_sobolev}.  
	\end{proof}
\end{Corollary}

\begin{Proposition}
	\label{proposition:chain_rule_sobolev_real}
	Let $I \subset \Real$ be a bounded open interval, $m \geq 1$ an integer, $f\in H^m(I)$, $\Ucal \subset \Real$ open with $f(I) \subset \Ucal$, and $g\in\Ccal^m(\bar{\Ucal})$. Then $g\circ f \in H^{m}(I)$.
\begin{proof}
    We show the claim by induction on $m$, and start with the case $m=1$.
    By Proposition~\ref{prop:Sobo2Diff}, $f$ is continuous on $\bar{I}$, so that $f(\bar{I})=:[a,b]$ is a compact interval (by the intermediate value theorem). As Proposition~\ref{prop:sobolev_chain_rule} requires a function on all of $\Real$, we continue $g|_{[a,b]}$ affinely at $a$ and $b$, with slopes given by $g'(a)$ and $g'(b)$, respectively. This yields $G \in \Ccal^1(\Real)$ with $G = g$ on $f(I)$, and Proposition~\ref{prop:sobolev_chain_rule} with $d = 1$ gives $g \circ f = G \circ f \in H^1(I)$ with $D(g\circ f) = (g'\circ f)\, Df$.\newline
  	Then the induction step is as follows: Let $m \geq 2$, and assume the claim for $m - 1$. By the case $m = 1$, $g \circ f \in H^1(I)$ with $D(g\circ f) = (g'\circ f)\, Df$. Hence, it suffices to show $(g'\circ f)\, Df \in H^{m-1}(I)$, which by the product rule in Proposition~\ref{prop:product_sobolev} holds if the two factors are in $H^{m-1}(I)$. The induction hypothesis, applied to $f \in H^{m-1}(I)$ and $g' \in \Ccal^{m-1}(\bar\Ucal)$, gives $g' \circ f \in H^{m-1}(I)$, and $Df \in H^{m-1}(I)$ by definition of $H^m(I)$. 
\end{proof}
\end{Proposition}

\begin{Corollary}
	\label{corollary:chain_rule_sobolev_vector}
	Let $\dim(\Yspace)<\infty$, $I \subset \Real$ be a bounded open interval, $m \geq 1$ an integer, $f\in H^m(I,\Yspace)$, $\Ucal \subset \Yspace$ open with $f(I) \subset \Ucal$, and $g\in\Ccal^m(V,\Yspace)$ for some open $V \supset \bar{\Ucal}$. Then $g\circ f \in H^{m}(I,\Yspace)$.
	\begin{proof}
		Let $\{e_j\}_j$ be an orthonormal basis of $\Yspace$. Since $\dim(\Yspace)<\infty$, it remains to show that all the $(g\circ f)_i: x \mapsto \langle g\circ f(x), e_i\rangle_{\Yspace}$ are in $H^{m}(I)$. We also expand $f(x) = \sum_{i} f_i(x)\, e_i$ with $f_i: x\mapsto \langle f(x), e_i\rangle_{\Yspace}$ and $g(y) = \sum_{i} g_i(y)\, e_i$ with $g_i: y\mapsto \langle g(y), e_i\rangle_{\Yspace}$ with partial derivatives $\partial_j g_i(y):=\frac{\partial}{\partial e_j} g_i(y)$ into the basis directions.
		As $g \in \Ccal^m(V, \Yspace)$ with $V \supset \bar{\Ucal}$ open, in particular $g_i \in \Ccal^m(\bar{\Ucal})$ for all $i$. Moreover, $f_j \in H^m(I)$ for all $j$ by Theorem~\ref{thm:SobolevTensorProduct}, and $f \in \Ccal^{m-1}(\bar I, \Yspace)$ by Corollary~\ref{cor:Sobo2Diff}.

		We proceed by induction on $m$ and start with $m = 1$, showing that $g_i \circ f \in H^1(I)$ for all $i$, where $f_j \in H^1(I)$ for all $j$ and $g_i \in \Ccal^1(V)$ for all $i$. The case $m = 1$ of the proof of Prop.~\ref{proposition:chain_rule_sobolev_real} does not carry over directly, as it relies on $\Ucal \subset \Real$; instead, we use the open neighbourhood $V$ of $\bar{\Ucal}$. By Corollary~\ref{cor:Sobo2Diff}, $f \in \Ccal^0(\bar I, \Yspace)$ is bounded, say $\|f(x)\|_{\Yspace} < R$ for all $x \in I$, so that $f(I) \subset \Ucal \cap B_R(0)$; replacing $\Ucal$ by $\Ucal \cap B_R(0)$, whose closure is still contained in $V$, we may assume that $\bar{\Ucal}$ is compact. Identify $\Yspace$ with $\Real^d$, $d = \dim(\Yspace)$, via the basis $\{e_j\}_j$. As $\bar{\Ucal} \subset V$ is compact and $V$ is open, there is a smooth cut-off function $\chi \in \Ccal^\infty_c(V)$ with $\chi = 1$ on an open neighbourhood of $\bar{\Ucal}$. Set $G_i := \chi\, g_i$ on $V$ and $G_i := 0$ on $\Real^d \setminus V$. Then $G_i \in \Ccal^1(\Real^d)$, and $G_i = g_i$ as well as $\nabla G_i = \nabla g_i$ on $\Ucal \supset f(I)$. Proposition~\ref{prop:sobolev_chain_rule} yields $g_i \circ f = G_i \circ f \in H^1(I)$, as desired, with $D (g_i \circ f) = \sum_{j} (\partial_j g_i \circ f)\, D f_j$. This concludes the base step.\newline
		For the induction, let $m \geq 2$ now, and assume the claim for $m - 1$. By the case $m = 1$, $g_i \circ f \in H^1(I)$ with $D (g_i \circ f) = \sum_{j} (\partial_j g_i \circ f)\, D f_j$, so that it suffices to show that this sum lies in $H^{m-1}(I)$. The induction hypothesis, applied to $f \in H^{m-1}(I, \Yspace)$ and $\partial_j g = \sum_i (\partial_j g_i)\, e_i \in \Ccal^{m-1}(V, \Yspace)$, gives $\partial_j g_i \circ f \in H^{m-1}(I)$ for all $i$ and $j$, and $D f_j \in H^{m-1}(I)$ by definition of $H^m(I)$. Hence the sum lies in $H^{m-1}(I)$ by Proposition~\ref{prop:product_sobolev}.
	\end{proof}
\end{Corollary}

\begin{Corollary}
    \label{cor:change_of_charts}
	Let $\mathcal{M}$ be a differentiable manifold modelled on a finite-dimensional Hilbert space $\mathfrak{Z}$, let $\psi_i: \mathcal{N}_i \rightarrow \Ucal_i$, $i = 1, 2$, be charts from open sets $\mathcal{N}_1, \mathcal{N}_2 \subset \mathcal{M}$ onto open sets $\Ucal_1, \Ucal_2 \subset \mathfrak{Z}$, and let $\mathcal{N} \subset \mathcal{M}$ be open such that $\bar{\mathcal{N}}$ is compact with $\bar{\mathcal{N}} \subset \mathcal{N}_1 \cap \mathcal{N}_2$. Let $I \subset \Real$ be a bounded open interval, $\mu: I \rightarrow \mathcal{N}$, and $m \geq 1$ an integer. Then $\mu = \psi_1^{-1} \circ f_1$ for some $f_1 \in H^m(I, \mathfrak{Z})$ if and only if $\mu = \psi_2^{-1} \circ f_2$ for some $f_2 \in H^m(I, \mathfrak{Z})$.
	\begin{proof}
		By symmetry, it suffices to prove one implication. Let $\mu = \psi_1^{-1} \circ f_1$ with $f_1 \in H^m(I, \mathfrak{Z})$. Since $\psi_1$ is a diffeomorphism, $\Ucal := \psi_1(\mathcal{N})$ and $V := \psi_1(\mathcal{N}_1 \cap \mathcal{N}_2)$ are open with $\Ucal \subset V \subset \Ucal_1$. Moreover, $\psi_1(\bar{\mathcal{N}}) \subset V$ is compact, hence closed in $\mathfrak{Z}$, and contains $\Ucal$, so that $\bar{\Ucal} \subset \psi_1(\bar{\mathcal{N}}) \subset V$. By definition of a differentiable manifold, $g := \psi_2 \circ \psi_1^{-1} \in \Ccal^\infty(V, \mathfrak{Z})$.
		As $\mu$ takes values in $\mathcal{N}$, $f_1 = \psi_1 \circ \mu$ takes values in $\Ucal$. Corollary~\ref{corollary:chain_rule_sobolev_vector}, applied with $\mathfrak{Z}$ in place of $\Yspace$, yields $f_2 := g \circ f_1 \in H^m(I, \mathfrak{Z})$, and $\psi_2^{-1} \circ f_2 = \psi_2^{-1} \circ \psi_2 \circ \psi_1^{-1} \circ f_1 = \mu$.
	\end{proof}
\end{Corollary}

\begin{Remark}
    More generally, if $\mu: I \rightarrow \mathcal{M}$ does not take values in a single open set $\mathcal{N}$ as in Corollary~\ref{cor:change_of_charts}, we can still choose an open $\mathcal{N} \subset \mathcal{M}$ such that $\bar{\mathcal{N}}$ is compact with $\bar{\mathcal{N}} \subset \mathcal{N}_1 \cap \mathcal{N}_2$. If $\mu$ is continuous, which by Corollary~\ref{cor:Sobo2Diff} is the case whenever $\mu = \psi^{-1}_1 \circ f_1$ for some $f_1 \in H^m(I, \mathfrak{Z})$, then $\mu^{-1}(\mathcal{N})$ is open, and each of its connected components $J$ is a bounded open interval, so that Corollary~\ref{cor:change_of_charts} applies to $\left.\mu\right|_J$, with $\left.f_1\right|_J \in H^m(J, \mathfrak{Z})$.
\end{Remark}

\subsection{Choice of intercept}
\label{sec:choice_of_intercept}

From a modelling perspective, the next result shows that in some settings, the choice of $\origin$ does not affect the model considered. The intuition is that, as soon as the model class is rich enough, the curvature of the sphere can be compensated for inside of the model. \newline
In particular, Sobolev spaces are compatible with differentiable manifolds by design.
For $\Xspace \subset \Real^s$ open and bounded, denote by $H^m(\Xspace, \Yspace)$ the vector-valued Sobolev spaces of functions $f:\Xspace \to \Yspace$ of order $m \in \{1,2,\ldots \}$, and recall that for $m > s/2$, $H^m(\Xspace, \Yspace)$ is a VVRKHS (see Section~\ref{sec:sobolev} above).
 Let $\continuous^m(\bar\Xspace, \sphere)$ denote the functions $\Xspace\rightarrow\sphere$ with $m$-times continuously differentiable  extensions to some open set $\mathcal{U} \supset \overline \Xspace$ around the closure of $\Xspace$.
 The following statements illustrate that the model class is locally invariant to the choice of the intercept in two important special cases: these two settings suffice to convey the main message, while proving the corresponding result in full generality would be considerably more technical.
\begin{Proposition}
  \label{proposition:choice_of_origin_is_irrelevant_to_the_model}
  Let $\Xspace \subset \Real^s$ be open and bounded, and $\sphere \subset \Yspace$ be the sphere in a Hilbert space $\Yspace$. Let $\mu: \Xspace \to \sphere$ be a function.
  Then, 
  \begin{enumerate}
  \item If $\mu \in \continuous^m(\overline \Xspace, \sphere)$, with $m > s/2$, for any $\origin \in \sphere$ satisfying $\sup_{x \in \Xspace} \sphereDist(\origin, \mu(x)) < \pi$,
      \[
      f_\origin := \Log_\origin \circ \mu \in H^m(\Xspace, T_\origin \sphere).
    \] 
  \item If $\Xspace$ is a bounded open interval (so that $s = 1$), $\dim(\Yspace) < \infty$, and there exists $\origin' \in \sphere$ such that $c := \sup_{x \in \Xspace} \sphereDist(\origin', \mu(x)) < \pi$, then for every integer $m \geq 1$,
    \[
      \Log_{\origin'} \circ \mu \in H^m(\Xspace, T_{\origin'} \sphere) \Rightarrow
      \Log_{\origin} \circ \mu \in H^m(\Xspace, T_{\origin} \sphere) \text{ for any } \origin \in \sphere \text{ with } \sphereDist(\origin, \origin') < \pi - c.
    \]
\end{enumerate}
\begin{proof}
  We use the fact that $\sphere$ is a differentiable manifold, and obtain Statement 1 and 2 as special cases of Corollary~\ref{cor:continuous_mu} and \ref{cor:change_of_charts}, respectively.
  Denote $\mathcal{N}_p := \sphere \setminus \{-p\}$ for $p \in \sphere$, so that $\Log_p: \mathcal{N}_p \to \{v \in T_p\sphere : \Ynorm{v} < \pi\}$ is a chart with inverse $\Exp_p$.
  Together with the assumption of Statement~1, this provides the setting of Corollary~\ref{cor:continuous_mu}.

  For Statement~2, fix $r$ with $c < r < \pi - \sphereDist(\origin, \origin')$, which is possible as $\sphereDist(\origin, \origin') < \pi - c$, and set $\mathcal{N}' := \{q \in \sphere : \sphereDist(\origin', q) < r\}$. Then $\mu(\Xspace) \subset \mathcal{N}'$, because $\sphereDist(\origin', \mu(x)) \leq c < r$ for all $x \in \Xspace$, and $\overline{\mathcal{N}'} \subset \{q \in \sphere : \sphereDist(\origin', q) \leq r\} \subset \mathcal{N}_{\origin'} \cap \mathcal{N}_{\origin}$, since $\sphereDist(\origin', q) \leq r < \pi$ and $\sphereDist(\origin, q) \leq \sphereDist(\origin, \origin') + r < \pi$ for all $q \in \sphere$ with $\sphereDist(\origin', q) \leq r$. Moreover, $\overline{\mathcal{N}'}$ is compact, as $\sphere$ is compact for $\dim(\Yspace) < \infty$. Let $R: T_{\origin}\sphere \to T_{\origin'}\sphere$ be a linear isometry. The sphere $\sphere$ is a differentiable manifold modelled on $\mathfrak{Z} := T_{\origin'}\sphere$ with charts $\psi_1 := \Log_{\origin'}: \mathcal{N}_{\origin'} \to \{v \in T_{\origin'}\sphere : \Ynorm{v} < \pi\}$ and $\psi_2 := R \circ \Log_{\origin}: \mathcal{N}_{\origin} \to \{v \in T_{\origin'}\sphere : \Ynorm{v} < \pi\}$, whose inverses are $\Exp_{\origin'}$ and $\Exp_{\origin} \circ R^{-1}$. As $\mu = \Exp_{\origin'} \circ f_1 = \psi_1^{-1} \circ f_1$ with $f_1 := \Log_{\origin'} \circ \mu \in H^m(\Xspace, T_{\origin'}\sphere)$ by assumption, Corollary~\ref{cor:change_of_charts}, applied with $I = \Xspace$, $\mathcal{N} = \mathcal{N}'$ and $\mathfrak{Z} = T_{\origin'}\sphere$, yields $f_2 \in H^m(\Xspace, T_{\origin'}\sphere)$ with $\mu = \psi_2^{-1} \circ f_2 = \Exp_{\origin} \circ R^{-1} \circ f_2$, i.e., $\Log_{\origin} \circ \mu = R^{-1} \circ f_2 \in H^m(\Xspace, T_{\origin}\sphere)$, as $R^{-1}$ is a linear isometry.
\end{proof}
\end{Proposition}

\section{Main proofs}

\begin{proof}[Proof of Proposition \ref{proposition:mu_well_defined_everywhere}]
For $x\in\Xspace$ write $r(x,q):=\Expected[\sphereDist^2(q,Y)\mid X=x]=\int\sphereDist^2(q,y)\,\theLaw_{Y\mid X=x}(dy)$, well defined because the integrand is bounded and continuous, and $m(x):=\inf_{q\in\sphere}r(x,q)$.

\emph{Existence and uniqueness.} Under the assumption, $\theLaw_{Y\mid X=x}$ is concentrated on a closed ball of radius $\pi/2-\eta_x<\pi/2$ of the unit sphere, which is a complete CAT(1) space. By \citet[Theorem~B, with $p=2$]{yokotaConvexFunctionsBarycenter2017} such a measure has a unique barycenter, that is, $r(x,\cdot)$ has a unique minimiser over $\sphere$, which we denote by $\mu(x)$.

\emph{Measurability.} We show that $\mu$ is Borel measurable by showing that its graph is a Borel set.

The sphere is a closed subset of the separable Hilbert space $\Yspace$, hence a Polish space, so the fixed version of the regular conditional distribution of $Y$ given $X$ is a probability kernel from $\Xspace$ to $\sphere$: for every Borel set $E\subseteq\sphere$, the map $x\mapsto\theLaw_{Y\mid X=x}(E)$ is Borel measurable \citep[Theorem~5.3]{kallenberg1997foundations}. Integrating a nonnegative Borel function against a probability kernel gives a Borel function of the parameter \citep[Lemma~1.38(i) and its proof]{kallenberg1997foundations}. We apply this with the parameter $(x,q)\in\Xspace\times\sphere$, the kernel $((x,q),E)\mapsto\theLaw_{Y\mid X=x}(E)$, which is measurable in $(x,q)$ because it is measurable in $x$, and the nonnegative continuous integrand $((x,q),y)\mapsto\sphereDist^2(q,y)$, the product $\sigma$-algebra of $\Xspace\times\sphere$ being its Borel $\sigma$-algebra since both factors are separable metric spaces. Hence $(x,q)\mapsto r(x,q)$ is Borel on $\Xspace\times\sphere$, and in particular $x\mapsto r(x,q)$ is Borel on $\Xspace$ for every fixed $q$.

For fixed $x$, the map $q\mapsto r(x,q)$ is Lipschitz (consequence of the proof of Lemma~\ref{Lemma:ell_y_Lipschitz}), and in particular it is continuous.
The sphere is separable, so it has a countable dense subset $\{q_1,q_2,\ldots\}$, and $m(x)=\inf_nr(x,q_n)$ for every $x\in\Xspace$.
A countable infimum of Borel functions is Borel, since $\{x:\inf_ng_n(x)<a\}=\bigcup_n\{x:g_n(x)<a\}$ for every $a\in\Real$.
Hence $m$ is Borel on $\Xspace$.

Consider
\[
  G:=\{(x,q)\in\Xspace\times\sphere:\ r(x,q)=m(x)\}.
\]
This is a Borel set, the zero set of the function $(x,q)\mapsto r(x,q)-m(x)$, which is Borel because $r$ is Borel on $\Xspace\times\sphere$ and $m$ is Borel on $\Xspace$. For each $x\in\Xspace$, the section $\{q\in\sphere:(x,q)\in G\}$ is the set of minimisers of $r(x,\cdot)$, which is the single point $\mu(x)$ by uniqueness, so $G$ is the graph of $\mu$.

The spaces $\Xspace$ and $\sphere$ are Polish, hence standard Borel spaces, and a map between standard Borel spaces is Borel measurable if and only if its graph is a Borel set \citep[Theorem~14.12]{kechrisClassicalDescriptive1995}. Hence $\mu:\Xspace\to\sphere$ is Borel measurable, and it satisfies \eqref{eq:muX} uniquely for every $x\in\Xspace$, relative to the fixed version of the regular conditional distribution.
\end{proof}

\begin{proof}[Proof of Theorem~\ref{theorem:existence_of_global_minimizer_unconstrained}]
  Assume $\dim \vspan(Y_1, \ldots, Y_n, \origin) +n < \dim \Yspace$.
        Fix $f \in \Hspace$, and let 
        \begin{align*}
            \mathcal Y &= \vspan\{\origin, Y_1,\dots,Y_n\}, \\ 
            V_f &= \vspan\{f(X_1),\dots,f(X_n)\}, \\
            W_f &= \mathcal Y + V_f, \\
            \widetilde W &= \mathcal Y + \vspan\{ w_1,\ldots, w_n \}.
        \end{align*}
then 
$\mathcal Y \subseteq W_f \cap \widetilde W$, and $\dim W_f \leq \dim \widetilde W$. Hence by Lemma~\ref{Lemma:rotations}, 
$\exists R \in \orthogonal(\Yspace)$ such that
$R|_{\mathcal Y} = \id_{\mathcal Y}$
and $R W_f \subseteq \widetilde W$.
Using these properties together with Lemmas~\ref{Lemma:rotations_and_sphere}, \ref{Lemma:operator_on_Y_induces_operator_on_H} and \ref{Lemma:orthogonal_operator_Hspace}, we get
\begin{align*}
    \risk_n(f, \lambda) &= \frac 1n \sum_{i=1}^n \sphereDist^2( \Exp_{\origin}( f(X_i)), Y_i) + \lambda^2 \Hnorm{f}^2
                    \\ &= \frac 1n \sum_{i=1}^n \sphereDist^2( \Exp_{R \origin}( (\tilde R f)(X_i)), R Y_i) + \lambda^2 \Hnorm{\tilde R f}^2
                    \\ &= \frac 1n \sum_{i=1}^n \sphereDist^2( \Exp_{ \origin}( (\tilde R f)(X_i)), Y_i) + \lambda^2 \Hnorm{\tilde R f}^2
                     \\&= \risk_n( \tilde R f, \lambda),
\end{align*}
where $\tilde R$ is defined in Lemma~\ref{Lemma:operator_on_Y_induces_operator_on_H}. Applying Lemma~\ref{Lemma:minimum_interpolant_vvrkhs}, we get
\[
    \risk_n(\tilde R f, \lambda) \geq \risk_n(\tilde f, \lambda),
\]
for some $\tilde f = \sum_{i=1}^n K_{X_i} \xi_i$ with $\xi_i \in \vspan\{ R f(X_i) : i = 1,\ldots, n\}$.
Notice that $R f(X_i) \in \widetilde W$ but also $\YinnerProd{Rf(X_i), \origin}= 0$ hence $\proj_{\origin} Rf(X_i) = Rf(X_i)$. Hence 
\[
\xi_i \in \vspan\{ R f(X_i) : i = 1,\ldots, n\} = \proj_{\origin}\vspan\{ R f(X_i) : i = 1,\ldots, n\} \subseteq \proj_{\origin} \widetilde W.
\]
This proves \eqref{eq:inf_reduces_to_finite_dim_inf}.

For the existence of the minimiser if $\lambda > 0$, notice that the objective function is continuous in $f$ since $\sup_x k(x,x) < \infty$, and that the optimisation set is compact: indeed $f$ only needs to be minimized on an $\Hspace$-ball since $\risk_n(., \lambda)$ is coercive, and $\widetilde \Hspace$ is finite-dimensional.

For the last statement, take a rotation $R \in \orthogonal(\Yspace)$ such that $R\origin = \origin$, $RY_i = Y_i$ and $R w_l = -w_l$ for all $i,l$. Then 
\[
    \risk_n(\tilde R \hat f, \lambda) = \risk_n( \hat f, \lambda)
\]
so $\tilde R \hat f$ is also a minimiser. By uniqueness, $0 = \tilde R \hat f - \hat f$, and writing $\hat f$ in terms of the elements defining the spans above finishes the proof.

If $\dim \vspan(Y_1, \ldots, Y_n, \origin) +n \geq \dim \Yspace$, then take only $m := \dim \Yspace - \dim \vspan(Y_1, \ldots, Y_n, \origin)$ orthogonal vectors $w_1,\ldots, w_m \in \vspan(Y_1,\ldots, Y_n, \origin)^\perp \cap \sphere$. The first part of the proof is then trivially obtained since $\proj_\origin \widetilde W = T_\origin \sphere$.
  The uniqueness part follows directly from arguments similar to those given above.
    \end{proof}

\begin{proof}[Proof of Theorem~\ref{theorem:ell_Hessian_new}]
  The first statement is Proposition~\ref{proposition:analytic_convexity}. 
  The second statement follows from Lemma~\ref{Lemma:gradient_l_v}, the first statement, and Lemma~\ref{Lemma:convex_gradient}.
\end{proof}

Before turning to the proof of Theorem~\ref{theorem:risk_is_convex_new}, we need the following result.
\begin{Lemma}
    \label{Lemma:population_risk_convexity_inequality_new}
    Recall the definition of the population risk $\risk$ from \eqref{eq:population_risk}.
    Under assumption~\ref{assumption:support_of_Y_for_theory_new}, provided $\Expected k(X,X) < \infty$,
 $\nabla \risk(f)$  exists for all $f \in \Uset$ and
    \[
        \nabla \risk(f) =   \Expected\eval_X^\adjoint \nabla \ell_{Y}( f(X))
    \]
    where $\ell_y$ is defined in \eqref{eq:ell_function}.
    In addition, there is an $\epsilon > 0$ such that
    \[
      \risk(f) - \risk(g) \geq \HinnerProd{ \nabla \risk(g) ,  f-g }  + \frac{\epsilon}{2} \Ltwonorm{f - g}^{2} 
    \]
    for all $f,g \in \Uset$.
    \begin{proof}
      The first claim follows by applying the dominated convergence theorem. 

        For the second claim, by Theorem~\ref{theorem:ell_Hessian_new},
        \[
            \ell_y(v') \geq \ell_y(v) + \YinnerProd{ \nabla \ell_y(v), v' - v} + \frac{\epsilon}{2} \Ynorm{v'-v}^2, \quad \forall v,v' \in \Cregion.
        \]
        Using this, for all $f,g \in \Uset$,
        \begin{align*}
            \risk(f) - \risk(g) &= \Expected\Bigl[ \sphereDist^{2}\bigl(\Exp_{\origin}(f(X)), Y\bigr) - \sphereDist^{2}\bigl(\Exp_{\origin}(g(X)), Y\bigr) \Bigr] \\
                                &\ge \Expected\Bigl[ \YinnerProd{ \nabla \ell_{Y}( g(X)), (f-g)(X) } + \frac{\epsilon}{2} \Ynorm{ (f-g)(X) }^2 \Bigr] \\
                                &= \Expected\Bigl[ \YinnerProd{ \nabla \ell_{Y}( g(X)), \eval_X (f-g) } + \frac{\epsilon}{2} \Ynorm{ (f-g)(X) }^2 \Bigr] \\
                                &= \Expected\Bigl[ \HinnerProd{ \eval_X^\adjoint \nabla \ell_{Y}( g(X)),  f-g } \Bigr]  + \frac{\epsilon}{2} \Ltwonorm{f - g}^{2}\\
                                &=  \HinnerProd{  \Expected\eval_X^\adjoint \nabla \ell_{Y}( g(X)),  f-g }  + \frac{\epsilon}{2} \Ltwonorm{f - g}^{2}\\
                                &=  \HinnerProd{\nabla \risk(g) ,  f-g }  + \frac{\epsilon}{2} \Ltwonorm{f - g}^{2}.
        \end{align*}
        The permutation of expectation and inner-product is valid since 
        \begin{align*}
            \Expected \Hnorm{\eval_X^\adjoint \nabla \ell_{Y}( g(X))} &\leq \Expected \operatornorm{\eval_X \eval_X^\adjoint}^{1/2} \Ynorm{\nabla \ell_Y(g(X))} \\
            &\leq 2\pi \Expected \sqrt{k(X,X)} \leq 2\pi ( \Expected k(X,X) )^{1/2} < \infty
        \end{align*}
        where we have used Lemma~\ref{Lemma:gradient_l_v} for the second inequality.
    \end{proof}

\end{Lemma}

\begin{figure}
  \includegraphics[width=\textwidth,height=0.6\textheight,keepaspectratio]{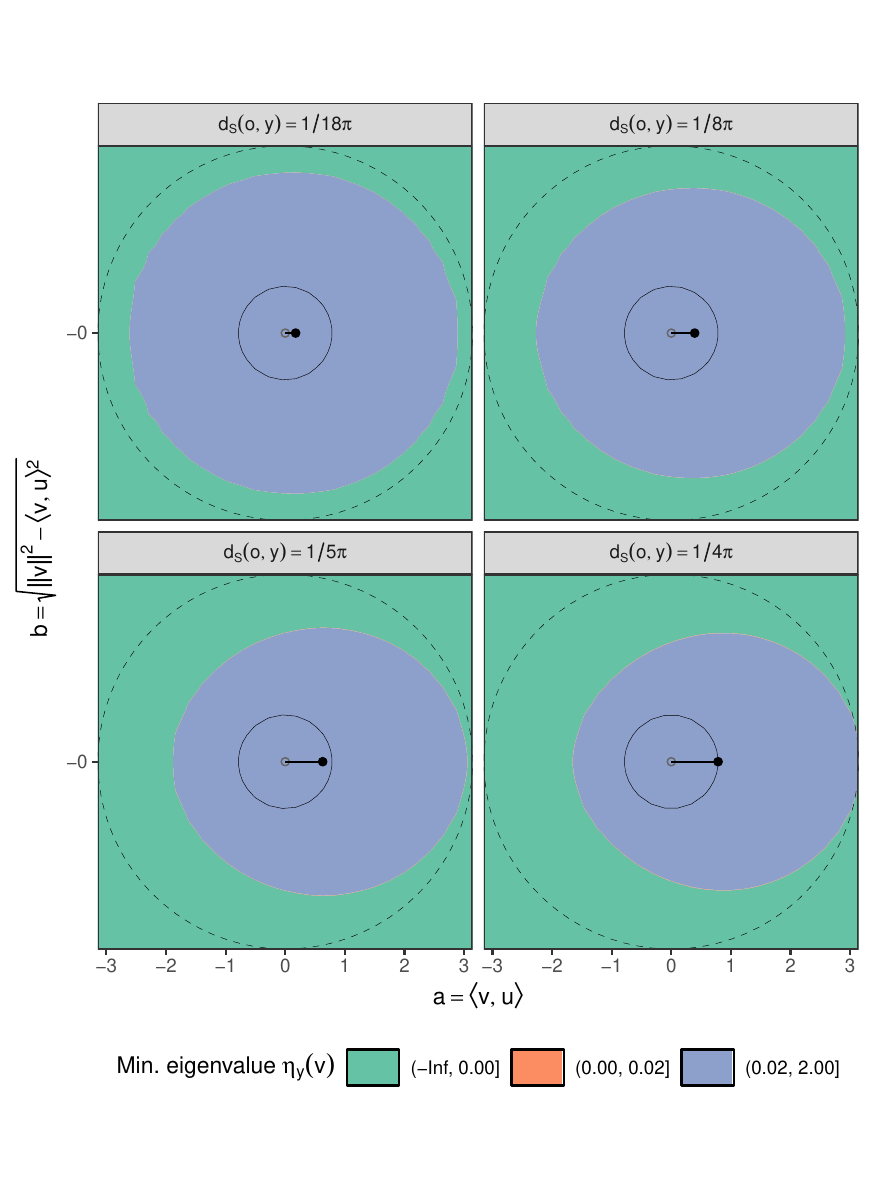}
  \caption{
    Plane $(a,b)$ of the domain of $g_\zeta$, for different values of $\zeta=\sphereDist(\origin, y)$ (one value per panel).
    The colours depict different ranges of values of $g_\zeta(a,b)$, as described in the plot's legend. 
    In each panel, the circle denotes the point $(0,0)$, and the filled circle the point $(\zeta, 0)$. The dashed circle is the circle of radius $\pi$ around $(0,0)$, and corresponds to the circle of radius $\pi$ on $T_\origin \sphere$. The solid circle is the circle of radius $\pi/4$ around $(0, 0)$. 
    This plot in conjunction with Proposition~\ref{proposition:convexity_plot} illustrates that if $\zeta = \sphereDist(\origin,y) \leq \pi/4$ then $\ell_y(v)$ is convex over the region $v \in \ball_{T_\origin \sphere}( 0, \pi/4)$. 
  }
    \label{figure:local_convexity} 
\end{figure}

\begin{proof}[Proof of Theorem~\ref{theorem:risk_is_convex_new}]
  \mbox{}

  \begin{enumerate}
    \item This is given by Lemma~\ref{Lemma:population_risk_convexity_inequality_new} and Lemma~\ref{Lemma:convex_gradient}, since $f_\origin$ minimises $\risk(\cdot)$ over $\Uset$.
    \item Recall that \( \risk_n(f, \lambda) = \frac 1n \sum_{i=1}^n \ell_{Y_i}(f(X_i)) + \lambda^2 \Hnorm{f}^2. \) The sum is therefore almost surely convex on $\Uset$ by Theorem~\ref{theorem:ell_Hessian_new}, and the term $\lambda_n^2 \Hnorm{f}^2$ is $2\lambda_n^2$-strongly convex on $\Hspace$. Hence $f \in \Uset \mapsto \risk_n(f, \lambda_n)$ is almost surely $2\lambda_n^2$-strongly convex on $\Uset$. By Lemma~\ref{Lemma:Uset_is_closed_and_nonempty_new}, $\Uset$ is closed. It is straightforward to see that $f \mapsto \risk_n(f,\lambda_n)$ is continuous and coercive, and the almost sure existence of a minimiser therefore follows from Lemma~\ref{Lemma:convex-closed-semi_continuity}. The uniqueness follows from the strict convexity.
  \end{enumerate}
\end{proof}

\begin{proof}[Proof of Theorem~\ref{theorem:rates_for_f_weakest_assumptions_corrected}]
Our proof is based on a stability argument  \citep{bousquetStabilityGeneralization2002}. We will use the following notation: for $z = (x,y) \in \Zspace := \Xspace \times \sphere$ let $\eta(f,z) := \sphereDist^2(\Exp_\origin(f(x)), y), f \in \Hspace$. We shall assume that $\{Z_i = (X_i, Y_i), Z'_i = (X'_i, Y'_i) \mid i=1,\ldots, n \}$ are an i.i.d.\ sample from $\theLaw_{X,Y}$, and write $S = \{ Z_i \}_{i=1}^n$ and $S_i = \{ Z_1,\ldots, Z_{i-1}, Z_i', Z_{i+1} , \ldots, Z_n \}$, i.e.\ the sample $S$ but with $Z_i$ replaced by $Z_i'$. For a set $A = \{z_1, \ldots, z_n\}$ we write
\[
\risk_A(f, \lambda) := \frac 1n \sum_{i=1}^n \eta(f, z_i) + \lambda^2 \Hnorm{f}^2, \quad f \in \Hspace,
\]
and if $A$ is a set of random variables, we denote by $\Expected_A(\cdot)$ the expectation taken with respect to the joint distribution of $A$.
Notice that the population risk \eqref{eq:population_risk} can be written
\[
\risk(f) = \Expected_Z \eta(f, Z), \quad f \in \Hspace.
\]
where $Z \sim \theLaw_{X,Y}$. 
For $\lambda >0$, let $\hat f$, respectively $\hat f_i$ be the minimiser of $\risk_S(\cdot, \lambda)$, respectively $\risk_{S_i}(\cdot, \lambda)$, over $\Uset$, which exist and are unique almost surely by Theorem~\ref{theorem:risk_is_convex_new}.

\paragraph{Step 1}
We first show that $\eta(\cdot, z)$ is $2 \pi c$-Lipschitz on $\Uset$ for all $z \in \Zspace$.
Let $f,g\in\Uset$ and $z = (x,y) \in \Zspace$. Since $\eta(f,z) = \ell_y(f(x))$, Lemma~\ref{Lemma:ell_y_Lipschitz} and Lemma~\ref{Lemma:bounding_fX} give
\[
    |\eta(f,z)-\eta(g,z)|
    \leq 2\pi\Ynorm{f(x)-g(x)}
    \leq 2\pi\sqrt{k(x,x)}\,\Hnorm{f-g}
    \leq 2 \pi c \Hnorm{f-g}.
\]

Since $\risk_S( \cdot , \lambda), \risk_{S_i}(\cdot, \lambda)$ are almost surely $2\lambda^2$-strongly convex (Theorem~\ref{theorem:risk_is_convex_new}), by definition of $\hat f, \hat f_i$,
\begin{align*}
    \risk_S(\hat f_i, \lambda) &\geq \risk_S(\hat f, \lambda) + \lambda^2\Hnorm{\hat f_i-\hat f}^2, \quad \text{a.s.,}\\
    \risk_{S_i}(\hat f, \lambda) &\geq \risk_{S_i}(\hat f_i, \lambda) + \lambda^2\Hnorm{\hat f-\hat f_i}^2, \quad \text{a.s.}
\end{align*}
Adding the two inequalities and cancelling the contributions of the
common observations gives the following almost sure inequalities,
\begin{align*}
    2\lambda^2\Hnorm{\hat f-\hat f_i}^2
    &\leq \risk_S(\hat f_i, \lambda) - \risk_S(\hat f, \lambda) + \risk_{S_i}(\hat f, \lambda) - \risk_{S_i}(\hat f_i, \lambda)
    \\ &= \frac1n \left[ \eta(\hat f_i, Z_i) - \eta(\hat f, Z_i) + \eta(\hat f, Z_i') - \eta(\hat f_i, Z_i') \right] 
    \intertext{and using the Lipschitz property,}
    &\leq \frac{4 \pi c}{n} \Hnorm{\hat f-\hat f_i}.
\end{align*}
Consequently,
\[
    \Hnorm{\hat f-\hat f_i} \leq \frac{2 \pi c}{\lambda^2 n}, \quad \text{a.s.}
\]
Applying the Lipschitz bound once more, for all $z \in \Zspace$,
\[
    |\eta(\hat f, z)-\eta(\hat f_i,z)| \leq 2\pi c\Hnorm{\hat f-\hat f_i} \leq \frac{4 \pi^2 c^2}{n\lambda^2} =: \beta_n, \quad \text{a.s.}
\]

\paragraph{Step 2}

Let $S' = \{ Z_i'\}_{i=1}^n$. We have the following inequalities, almost surely.
\begin{align*}
    \risk_{S'}(\hat f, 0) - \risk_S(\hat f, 0) &= \frac 1n \sum_{i=1}^n \left[\eta(\hat f, Z_i')  - \eta(\hat f, Z_i)\right] 
    \\ &= \frac 1n \sum_{i=1}^n \left[\eta(\hat f, Z_i') - \eta(\hat f_i, Z_i') +\eta(\hat f_i, Z_i') - \eta(\hat f, Z_i)\right] 
    \\ &\leq \beta_n + \frac 1n \sum_{i=1}^n \left[ \eta(\hat f_i, Z_i') - \eta(\hat f, Z_i) \right]
\end{align*}
Now since $\Expected_{S'} [\risk_{S'}(\hat f,0)]= \risk(\hat f)$, we have
\begin{align*}
    \Expected_{S} [\risk(\hat f) - \risk_S(\hat f, 0)] &= \Expected_{S,S'}[ \risk_{S'}(\hat f, 0) - \risk_S(\hat f, 0)]
    \\ &\leq \beta_n + \frac 1n \sum_{i=1}^n \Expected_{S,S'}[\eta(\hat f_i, Z_i') - \eta(\hat f, Z_i) ]
    \\ &= \beta_n + \frac 1n \sum_{i=1}^n \left[ \Expected_{S_i}[\eta(\hat f_i, Z_i')] - \Expected_{S}[\eta(\hat f, Z_i) ] \right].
    \\ &= \beta_n,
\end{align*}
since $\eta(\hat f_i, Z_i')$ and $\eta(\hat f, Z_i)$ have the same marginal distribution.

\paragraph{Step 3}

By the optimality of $\hat f$ over $\Uset$, and
because $f_\origin \in \Uset$,
\[
    \risk_S(\hat f, 0)
    \leq \risk_S(\hat f, 0)
    +
    \lambda^2\Hnorm{\hat f}^2
    \leq
    \risk_S(f_\origin,0)
    +
    \lambda^2\Hnorm{f_\origin}^2.
\]
Hence,
\[
    \risk_S(\hat f,0)-\risk_S(f_\origin, 0)
    \leq
    \lambda^2\Hnorm{f_\origin}^2.
\]
Using $\Expected[\risk_S(f_\origin, 0)]=\risk(f_\origin)$, we obtain
\begin{align*}
    \Expected\left[
    \risk(\hat f)-\risk(f_\origin)
    \right]
    &= \Expected\left[ \risk(\hat f)-\risk_S(f_\origin, 0) \right] \\
    &= \Expected\left[ \risk(\hat f)-\risk_S(\hat f,0) \right] + \Expected\left[ \risk_S(\hat f,0)-\risk_S(f_\origin, 0) \right] 
    \intertext{using step 2,}
    &\leq \beta_n  + \lambda^2\Hnorm{f_\origin}^2 \\
    &= \frac{4\pi^2c^2}{n\lambda^2} + \lambda^2\Hnorm{f_\origin}^2.
\end{align*}

\paragraph{Step 4}

Using Theorem~\ref{theorem:risk_is_convex_new},
    \[
    \Ltwonorm{\hat f-f_\origin}^2 \leq \frac{2}{\epsilon} [\risk(\hat f)-\risk(f_\origin)],
    \]
    and therefore, by step 3,
\[
    \Expected \Ltwonorm{\hat f-f_\origin}^2 \leq \frac{2}{\epsilon} \left[ \frac{4\pi^2c^2}{n\lambda^2} + \lambda^2\Hnorm{f_\origin}^2 \right].
\]
Setting $\lambda^2 = 8\pi^2 c^2 n^{-1/2}$, we get 
\[
    \Expected \Ltwonorm{\hat f-f_\origin}^2 \leq n^{-1/2}\left[ \epsilon^{-1} + 16 \pi^2 c^2 \epsilon^{-1} \Hnorm{f_\origin}^2 \right].
\]
By Markov's inequality, setting $\tilde c^2 = \epsilon^{-1} + 16 \pi^2 c^2 \epsilon^{-1} \Hnorm{f_\origin}^2$, we get
\[
\theLaw( \Ltwonorm{\hat f - f_\origin }  \geq \alpha ) \leq \alpha^{-2} \Expected \Ltwonorm{\hat f - f_\origin}^2  \leq \tilde c^2 \alpha^{-2} n^{-1/2}.
\]
Now fix $\delta \in (0,1)$ and set $\alpha = \frac{\tilde c}{\sqrt \delta}\, n^{-1/4}$. Then, with probability $1-\delta$, for any $\delta \in (0,1)$,
\[
\Ltwonorm{ \hat f - f_\origin } \leq \frac{\tilde c}{\sqrt{\delta}} n^{-1/4}.
\]

The claims for $\hat \mu_n$ follow directly from Lemma~\ref{Lemma:distance_and_exp_fixed_p}.
\end{proof}

\begin{proof}[Proof of Theorem~\ref{theorem:rates_for_f_depending_on_smoothness_new}]
        Let $\nu_0 = \epsilon$, given in Theorem~\ref{theorem:ell_Hessian_new}.
        We shall apply \citet[][Theorem~3.4.6 p.434]{vandervaartWeakConvergenceEmpirical2023a}  to
        \begin{align*}
            M_n(f) = M(f) &=  - \Expected \sphereDist^2(\Exp_\origin f(X), Y), f \in \Hspace,
            \\ \Mprocess_n(f) &= - \frac 1n \sum_{i=1}^n \sphereDist^2(\Exp_\origin f(X_i), Y_i),
            \\ d_n(f,g) &= d(f,g) = \frac{\sqrt{\nu_0}}{2 \sqrt 2} \Ltwonorm{f-g} 
            \\ \Theta_n &= \Uset
            \\ \theta_n &= \theta_{n,0} = f_\origin
            \\ \mathcal J_n(f) &= \Hnorm{f}
            \\ ( \lambda_n )_{n\geq 1} &\subset (0, \infty),
            \\ \underline{\delta}_n &= \lambda_n \Hnorm{f_\origin}
        \end{align*}
        Assume throughout that $\delta > \underline{\delta}_n$.
        By Lemma~\ref{Lemma:population_risk_convexity_inequality_new} and Lemma~\ref{Lemma:convex_gradient}, 
        \[
            \sup_{f \in \Uset: \delta/2 < d(f,f_\origin) \leq \delta} M(f) - M(f_\origin) \leq -\delta^2.
        \]
        so the first assumption of \citet[][Theorem~3.4.6 p.434]{vandervaartWeakConvergenceEmpirical2023a} holds. Next, we need to bound
        \[
            v_n(\delta) := \Expected^* \sup_{f \in \Uset: d(f,f_\origin) \leq \delta, \Hnorm{f} < \delta/\lambda_n} 
            \sqrt{n}\left| (\Mprocess_n - M)(f) - (\Mprocess_n - M)(f_\origin) \right|
        \]
        Notice that
        \[
            (\Mprocess_n - M)(f) - (\Mprocess_n - M)(f_\origin)  = (\Pprocess_n - \theLaw_{X,Y})(\psi_f - \psi_{f_\origin}),
        \]
        where  $\psi_f(x,y) = -\sphereDist^2( \Exp_\origin( f(x) ), y)$, 
        $\theLaw_{X,Y}$ is the law of $(X,Y)$ and $\Pprocess_n$ the empirical law of the iid sample $\{(X_i,Y_i) \mid i=1,\ldots, n\}$.  Let
        \[
            \Fset(\delta, \lambda_n) := \{f \in \Uset \mid \Ltwonorm{f-f_\origin} \le \frac{2\sqrt{2}}{\sqrt{\nu_0}} \delta,\ \Hnorm{f} < \delta/\lambda_n\},
        \]
        and
        \[
            \Gset(\delta, \lambda_n) := \{g \in \Hspace \mid \Ltwonorm{g} \le \frac{2\sqrt{2}}{\sqrt{\nu_0}} \delta,\ \Hnorm{g} < 2\delta/\lambda_n\},
        \]
        then
        \begin{align*}
            \sqrt{n} v_n(\delta) &= n \Expected^* \sup_{f \in \Fset(\delta, \lambda_n)} | (\Pprocess_n - \theLaw_{X,Y})(\psi_f - \psi_{f_\origin}) |
            \intertext{Letting $\vep_1,\ldots, \vep_n$ be independent Rademacher variables, taking values $\pm 1$ with probability $1/2$, symmetrisation \citep[][Lemma~2.3.1]{vandervaartWeakConvergenceEmpirical2023a} gives,}
                                 &\leq 2  \Expected^* \sup_{f \in \Fset(\delta, \lambda_n)} \left| \sum_{i=1}^n \vep_i (\psi_f(X_i,Y_i) - \psi_{f_\origin}(X_i,Y_i))\right|,
                              \\ &=   2\Expected^* \sup_{f \in \Fset(\delta, \lambda_n)} \left| \sum_{i=1}^n \vep_i \left( \sphereDist^2( \Exp_\origin( f_\origin(X_i) ), Y_i) -\sphereDist^2( \Exp_\origin( f(X_i) ), Y_i) \right) \right|,
                              \intertext{let $g = f - f_\origin$, then notice that $f \in \Fset(\delta, \lambda_n) \Rightarrow g \in \Gset(\delta, \lambda_n)$, hence}
                               &\leq   2\Expected^* \sup_{g \in \Gset(\delta, \lambda_n)} \left| \sum_{i=1}^n \vep_i \left( \sphereDist^2( \Exp_\origin( f_\origin(X_i) ), Y_i) -\sphereDist^2( \Exp_\origin( g(X_i) + f_\origin(X_i) ), Y_i) \right) \right|,
                              \intertext{letting $\Expected_{X,Y}^*$ be the outer expectation with respect to $(X_i,Y_i), i=1,\ldots,n$ and $\Expected_\vep^*$ be the outer expectation with respect to the Rademacher variables $\{ \vep_i \}$. Conditional on the sample $\{ (X_i, Y_i) \}$, let $h_i:T_\origin \sphere \to \Real$ be defined by $h_i(z) = \sphereDist^2( \Exp_\origin( f_\origin(X_i) ), Y_i) -\sphereDist^2( \Exp_\origin( z + f_\origin(X_i) ), Y_i)$. Then, }
                                 &=   2 \Expected^*_{X,Y} \Expected^*_\vep   \sup_{g \in \Gset(\delta, \lambda_n)} \left| \sum_{i=1}^n \vep_i h_i(g(X_i)) \right|
                                 \intertext{Since $h_i(z) = \ell_{Y_i}(f_\origin(X_i)) - \ell_{Y_i}(z + f_\origin(X_i))$, Lemma~\ref{Lemma:ell_y_Lipschitz} shows that the $h_i$ are Lipschitz with constant $2 \pi$. Furthermore, $g=0 \in \Gset(\delta, \lambda_n)$ and $h_i(0) = 0$ for all $i$. Applying Proposition~\ref{proposition:vector-contraction-with-absolute-value}, we get, for any fixed orthonormal basis $\{ e_j: j=1,\ldots, J \}$ of $T_\origin \sphere$, where $J = \dim(T_\origin \sphere) = \dim(\Yspace)-1$,}
                                 &\leq 8 \sqrt{2} \pi   \Expected^*_{X,Y} \Expected_\vep^*   \sup_{g \in \Gset(\delta, \lambda_n)} \sum_{i=1}^n \sum_{j=1}^J \vep_{ij} \YinnerProd{g(X_i), e_j}
                                 \intertext{where $\{\vep_{ij}\}$ are all independent Rademacher variables. Then,}
                                 &\leq 8 \sqrt{2} \pi \sum_{j=1}^J   \Expected^*_{X,Y} \Expected_\vep^*   \sup_{g \in \Gset(\delta, \lambda_n)} \sum_{i=1}^n \vep_{ij} \YinnerProd{g(X_i), e_j}
                              \intertext{Letting $\Gset_k(\delta, \lambda_n) := \{g \in \Hspace_k \mid \Ltwonorm{g} \le \frac{2\sqrt{2}}{\sqrt{\nu_0}} \delta,\ \Hknorm{g} < 2\delta/\lambda_n\} \subset \Hspace_k$, where $\Hspace_k$ is the RKHS associated to the scalar kernel $k$, Lemma~\ref{Lemma:projection_against_y_is_in_scalar_RKHS} yields}
                               &\leq 8 \sqrt{2} \pi\sum_{j=1}^J   \Expected^*_{X,Y} \Expected_\vep^*   \sup_{\check g \in \Gset_k(\delta, \lambda_n)} \sum_{i=1}^n \vep_{ij} {\check g(X_i)}
                              \\ &= 8 \sqrt{2} \pi \frac{2\delta}{\lambda_n} \sum_{j=1}^J   \Expected^*_{X,Y} \Expected_\vep^*   \sup_{\check g \in \Gset_k(\delta, \lambda_n)} \sum_{i=1}^n \vep_{ij} \frac{\check g(X_i)}{2\delta/\lambda_n}
                              \intertext{Notice that $\Hknorm{\check g/(2\delta/\lambda_n)} \leq 1$ and $\Ltwonorm{\check g/(2\delta/\lambda_n)}^2 \leq 2 \lambda_n^2/\nu_0$. Letting 
$\sigma_l$ be the $l$th eigenvalue of the spectral decomposition of $\KIntegralOperator$,  \citet[][Theorem~6.5]{bartlettLocalRademacherComplexities2005} yields}
                               &\leq 16 \sqrt{2} \pi \frac{\delta}{\lambda_n} n \sum_{j=1}^J \left( \frac{2}{n}  \sum_{l \geq 1} \min\{ \sigma_l, 2 \lambda_n^2/\nu_0 \} \right)^{1/2}
                               \\
                               & = 32 \pi \frac{\delta}{\lambda_n} \sqrt{n} J \left(  \sum_{l \geq 1} \frac{2 \lambda_n^2}{\nu_0} \min\{ \frac{\sigma_l}{2 \lambda_n^2/\nu_0}, 1 \} \right)^{1/2} 
                            \\ &= \frac{32{\sqrt{2}}}{\sqrt{\nu_0}} \pi J \sqrt{n} {\delta} \left( \sum_{l \geq 1} \min\{ 1, \frac{\sigma_l}{2\lambda_n^2/\nu_0}\} \right)^{1/2}
                        \intertext{and writing $\check \lambda_n^2 =  2\lambda_n^2/\nu_0$,}
                             &\lesssim  \sqrt{n} {\delta} \left( \sum_{l \geq 1} \min\{ 1, \frac{\sigma_l}{\check \lambda_n^2}\} \right)^{1/2}.
                             \intertext{and using the inequality $\min(1, a/b) \leq 2a/(a+b)$ we get}
                             &\lesssim  \sqrt{n} {\delta} \left( N(\check \lambda_n^2) \right)^{1/2},
        \end{align*}
                             where $N(\check \lambda) = \sum_{l \geq 1} \sigma_l/(\sigma_l + \check \lambda)$ is known as the `effective dimension' or `degree of freedom'.
                             Hence \[ v_n(\delta) \lesssim \delta (N(\check \lambda_n^2))^{1/2} =: \phi_n(\delta). \]
        $\phi_n(\delta)$ is sub-quadratic and 
        \[
            \phi_n(\delta) \leq \sqrt{n}\delta^2 \iff \delta \geq n^{-1/2} (N(\check \lambda_n^2))^{1/2}.
        \]
        For any $\delta_n \geq \text{max}(n^{-1/2}(N(\check \lambda_n^2))^{1/2}, \sqrt{\nu_0/2} \check \lambda_n \Hnorm{f_\origin})$, \citet[][Theorem~3.4.6]{vandervaartWeakConvergenceEmpirical2023a} implies $\Ltwonorm{\hat f - f_\origin}^2 = O_{\theLaw^*}(\delta_n^2)$.

        We now consider the various types of eigenvalue decays.
        \begin{description}
            \item[Finite-rank kernel.]  In this case $\sigma_l = 0$ for $l > L$ and hence $N(\lambda) \leq L$. The choice $\lambda_n \asymp e^{-n}$ gives
                $\Ltwonorm{\hat f_n - f_\origin}^2 = O_{\theLaw^*}(n^{-1})$.

            \item[Polynomial decay.]  In this case $\sigma_l \leq C l^{-p}$ with $p > 1$. Lemma~\ref{Lemma:effective_dimension_polynomial_decay} gives $(N(\lambda^2))^{1/2} = O(\lambda^{-1/p})$. The choice $\lambda_n \asymp n^{ -\frac{p}{2(p+1)} }$ 
                gives $\Ltwonorm{\hat f_n - f_\origin}^2 = O_{\theLaw^*}(n^{-\frac{p}{p+1}})$.

            \item[Stretched exponential.] 
            In this case $\sigma_l \leq C e^{-\alpha l^{1/q}}$ with $C, \alpha > 0$ and $q > 0$. Lemma~\ref{Lemma:effective_dimension_stretched_exponential} gives $(N(\lambda^2))^{1/2} = O( (\log(1/\lambda))^{q/2} )$. The choice $\lambda_n \asymp n^{ -1/2 }$
            gives $\Ltwonorm{\hat f_n - f_\origin}^2 = O_{\theLaw^*}( (\log(n))^{q}/n)$.
        \end{description}

        The rates for $\int \sphereDist^2(\mu(x), \hat \mu_n(x)) \ud\theLaw_X(x)$ follow directly from Lemma~\ref{Lemma:distance_and_exp_fixed_p}.
    \end{proof}

The following lemma is used for the second statement of Theorem~\ref{theorem:strong_consistency_new}. It bounds the values of a function in $\Hspace$ on the support of $\theLaw_X$ by its $\LtwoX$ norm and its $\Hspace$ norm, where the latter enters with a factor that can be made arbitrarily small.

\newcommand{\Lset}{\mathcal{L}}
\begin{Lemma} \label{Lemma:sup_norm_via_L2_and_H_norm} 
    Assume that Assumption~\ref{assumption:for_almost_sure_consistency} holds. Let $\kappa := \sup_{x \in \Xspace} \sqrt{k(x,x)} < \infty$ and $\Lset = \{ l \geq 1 | \sigma_l > 0 \}$. Assume $\Lset \neq \emptyset$.
    For each $L \in \Lset$ there exists a $T_L \geq 0$ such that
    \[
    g \in \Hspace \implies \sup_{x \in \support(\theLaw_X)} \Ynorm{g(x)} \leq \frac{\kappa}{\sqrt{\sigma_L}} \Ltwonorm{g} + T_L \Hnorm{g}.
    \]
    Furthermore, $\inf_{L \in \Lset} T_L = 0$.
    \begin{proof}
        Recall that $\Hspace_k$ is the RKHS associated to the scalar kernel $k$.
        We start by recalling some results from \citet{steinwartMercerTheoremGeneral2012}. 
        Let $I: \Hspace_k \to \LtwoX$ be the mapping from a function $h \in \Hspace_k$ to its equivalence class in $\LtwoX$. 
        Under our assumptions, $I$ is a compact operator, and $\KIntegralOperator = I I^\adjoint$ has $\LtwoX$-orthonormal eigenfunctions $\{\varphi_l\}_{l \in \Lset}$ with eigenvalues $\{\sigma_l\}_{l \in \Lset} \subset (0, \infty)$, and we can take the functions $\{ \varphi_l \}_{l \in \Lset}$ to be in $\Hspace_k$. 
        Furthermore, $\sigma_l \varphi_l(x) = (I^\adjoint \varphi_l )(x)$ for all $x \in \Xspace, l \in \Lset$ and $\{ \chi_l := \sqrt{\sigma_l} \varphi_l \}_{l \in \Lset}$ is an orthonormal system in $\Hspace_k$. 
        Letting $S = \support(\theLaw_X)$, Mercer's Theorem \citep[][Corollary~3.5]{steinwartMercerTheoremGeneral2012} gives
        \[
        k(x,x') = \sum_{l \in \Lset} \sigma_l \varphi_l(x) \varphi_l(x') = \sum_{l \in \Lset} \chi_l(x) \chi_l(x')
        \]
        where the series converge absolutely and uniformly on $S \times S$.
        For $L \in \Lset,$ set 
        \[
            T_L = \sup_{x \in S} \sqrt{k(x,x) - \sum_{l=1}^L \sigma_l \varphi_l^2(x)},
        \]
        and notice that $T_L \geq 0$ is well defined, with $\inf_{L \in \Lset} T_L = 0$.

        For $L \in \Lset$, let $P_L: \Hspace_k \to \Hspace_k$ be the orthogonal projection onto the span of $\chi_1,\ldots, \chi_L$. For any $h \in \Hspace_k$ and $x \in S$, writing $k_x = k(x, \cdot)$,
        \begin{align*}
        |h(x)| &= |\HkinnerProd{h, k_x}|
                \\ &\leq |\HkinnerProd{h - P_L h, k_x}| + |\HkinnerProd{P_Lh, k_x}|
        \intertext{Using the fact that $P_L$ is an orthogonal projection,}
                &\leq |\HkinnerProd{h - P_L h, k_x - P_L k_x}| + |\HkinnerProd{P_Lh, k_x}|
        \\        &\leq \Hknorm{h - P_L h} \Hknorm{k_x - P_L k_x} + |\HkinnerProd{P_Lh, k_x}|
        \intertext{Since $P_L$ is an orthogonal projection, and  $\Hknorm{k_x - P_L k_x}^2 = k(x,x) - \sum_{l=1}^L \sigma_l \varphi_l^2(x) \leq T_L^2$,}
                &\leq \Hknorm{h} T_L  + |\HkinnerProd{P_Lh, k_x}|.
                \intertext{By definition of $P_L$,}
                &= \Hknorm{h} T_L  + |\sum_{l=1}^L \HkinnerProd{h, \chi_l} \chi_l(x)| 
        \intertext{Using the Cauchy--Schwarz inequality,}
                &\leq \Hknorm{h} T_L  + \left( \sum_{l=1}^L \HkinnerProd{h, \chi_l}^2 \right)^{1/2} \left( \sum_{l=1}^L \chi_l^2(x) \right)^{1/2} 
                \intertext{Using Mercer's theorem,}
                &\leq \Hknorm{h} T_L  + \left( \sum_{l=1}^L \HkinnerProd{h, \chi_l}^2 \right)^{1/2} \sqrt{k(x,x)}
                \intertext{Since $\HkinnerProd{h, \chi_l} = \HkinnerProd{h, (\sigma_l)^{-1/2} I^\adjoint \varphi_l} = \innerProd{h, \varphi_l}_{\LtwoX}/\sqrt{\sigma_l}$}
                &\leq \Hknorm{h} T_L  + \left( \sum_{l=1}^L \innerProd{h, \varphi_l}_{\LtwoX}^2/{\sigma_l} \right)^{1/2} \kappa. 
                \intertext{Since the $\{\sigma_l\}$ are non-increasing and strictly positive,}
                &\leq \Hknorm{h} T_L  + \frac{1}{\sqrt{\sigma_L}}\left(  \sum_{l=1}^L \innerProd{h, \varphi_l}_{\LtwoX}^2 \right)^{1/2} \kappa 
                \intertext{Since $\{\varphi_l\}$ are orthonormal in $\LtwoX$,} 
                &\leq \Hknorm{h} T_L  + \frac{1}{\sqrt{\sigma_L}} \norm{h}_{\LtwoX} \kappa. 
        \end{align*}
        This shows
        \begin{equation} \label{eq:Hk_pointwise_bound}
        \sup_{x \in S } |h(x)| \leq \frac{\kappa}{\sqrt{\sigma_L}} \Ltwonorm{h} + T_L \Hknorm{h}, \quad \forall h \in \Hspace_k.
        \end{equation}
        We now turn this bound into a bound for all functions in $\Hspace$. Let $J = \dim(\Yspace) - 1$, and notice $J < \infty$ by assumption.
        Let $\{e_j : j = 1, \ldots, J\}$ be an orthonormal basis of $T_\origin \sphere$, and let $g_j := \YinnerProd{g(\cdot), e_j}$, which belongs to $\Hspace_k$ by Lemma~\ref{Lemma:projection_against_y_is_in_scalar_RKHS}. 
        Then $\Ynorm{g(x)}^2 = \sum_{j=1}^J g_j(x)^2$ for all $x \in \Xspace$, $\Ltwonorm{g}^2 = \sum_{j=1}^J \Ltwonorm{g_j}^2$, and $\Hnorm{g}^2 = \sum_{j=1}^J \Hknorm{g_j}^2$ by Lemma~\ref{Lemma:projection_against_y_is_in_scalar_RKHS}. 
        Applying \eqref{eq:Hk_pointwise_bound} to each $g_j$ and the triangle inequality in $\Real^J$ gives, for all $x \in S$,
        \begin{align*}
            \Ynorm{g(x)}
            &\leq \Big( \sum_{j=1}^J \Big( \frac{\kappa}{\sqrt{\sigma_L}} \Ltwonorm{g_j} + T_L \Hknorm{g_j} \Big)^2 \Big)^{1/2}
            \\ &\leq \frac{\kappa}{\sqrt{\sigma_L}} \Big( \sum_{j=1}^J \Ltwonorm{g_j}^2 \Big)^{1/2} + T_L \Big( \sum_{j=1}^J \Hknorm{g_j}^2 \Big)^{1/2}
            \\ &= \frac{\kappa}{\sqrt{\sigma_L}} \Ltwonorm{g} + T_L \Hnorm{g}.
        \end{align*}
    \end{proof}
\end{Lemma}

\begin{proof}[Proof of Theorem~\ref{theorem:strong_consistency_new}]
The following proof draws inspiration from \citet[][Theorem~5.7]{van2000asymptotic}. 
    Let 
    $\psi_f(x,y) = \sphereDist^2( \Exp_\origin( f(x) ), y)$
    and
    \[
        \mathcal{F} = \{\psi_f \mid f \in \widetilde \Uset \}, \quad 
        \mathcal{F}' = \{\psi_f \mid f \in \ball_\Hspace(0,C) \}
    \]
Since $\widetilde \Uset \subset \ball_\Hspace(0,C)$, the monotonicity of bracketing numbers and Lemma~\ref{Lemma:bracketing_for_strong_consistency} imply that
 \[
     \bracketing( \vep, \mathcal F, L^1(\theLaw_{X,Y})) \leq 
     \bracketing( \vep, \mathcal F ' , L^1(\theLaw_{X,Y})) \leq \covering( \vep(4\pi)^{-1}, \ball_\Hspace(0,C), \inftynorm{\cdot}).
 \]
 Lemma~\ref{Lemma:ball_in_VVRKHS_compact_in_sup_norm} implies therefore that $\bracketing( \vep, \mathcal F , L^1(\theLaw_{X,Y})) < \infty $ for every $\vep>0$. 
 Let $\Pprocess_n=n^{-1}\sum_{i=1}^n\delta_{(X_i,Y_i)}$. 
 Since $\risk_n(f, \lambda_n) = \Pprocess_n \psi_f + \lambda_n^2 \Hnorm{f}^2$ and $\risk(f) = \theLaw_{X,Y} \psi_f$, we get
 \[
 \sup_{f \in \widetilde \Uset} \left| \risk_n(f, \lambda_n) - \risk(f) \right| \leq 
 \sup_{f \in \widetilde \Uset} \left| (\Pprocess_n - \theLaw_{X,Y})\psi_f \right| + \sup_{f \in \widetilde \Uset} \lambda_n^2 \Hnorm{f}^2.
 \]
 The first term converges outer almost surely to zero by
 \citet[][Theorem~2.4.1]{vandervaartWeakConvergenceEmpirical2023a} and the second term is bounded by $\lambda_n^2 C^2 \to 0$ since $\lambda_n \to 0$. Therefore,
 \begin{equation} \label{eq:sup_convergence_of_risk}
     \sup_{f \in \widetilde \Uset} |\risk_n(f, \lambda_n)- \risk(f)| \outerASconv 0.
 \end{equation}
    
We have that for every $n$,
\begin{align*}
  \risk(\tilde f_n) - \risk(f_\origin)
  &= \big( \risk(\tilde f_n) - \risk_n(\tilde f_n,\lambda_n) \big)
   + \big( \risk_n(\tilde f_n,\lambda_n) - \risk(f_\origin) \big) \\
  &\le \big( \risk(\tilde f_n) - \risk_n(\tilde f_n,\lambda_n) \big)
   + \big( \risk_n(f_\origin,\lambda_n) - \risk(f_\origin) \big),
\end{align*}
where we used that $\risk_n(\tilde f_n,\lambda_n) \le \risk_n(f_\origin,\lambda_n)$ for all $n$.
Since $\tilde f_n \in \tilde\Uset$, we further obtain
\begin{align}
  \risk(\tilde f_n) - \risk(f_\origin)
  &\le \sup_{f\in\tilde\Uset} \big| \risk(f) - \risk_n(f,\lambda_n) \big|
   + \big( \risk_n(f_\origin,\lambda_n) - \risk(f_\origin) \big).
  \label{eq:asconsist-step}
\end{align}
By \eqref{eq:sup_convergence_of_risk},
$
  \risk_n(f_\origin,\lambda_n) \outerASconv  \risk(f_\origin).
$
By definition of outer almost sure convergence, and using \eqref{eq:asconsist-step}, we get that there exists a measurable $\Delta_n$ such that $\risk(\tilde f_n) - \risk(f_\origin) \leq \Delta_n$ and $\Delta_n \ASconv 0$.
Theorem~\ref{theorem:risk_is_convex_new} therefore implies that 
\[
  \Ltwonorm{\tilde f_n - f_\origin}^2 \leq \frac 2 \epsilon \Delta_n \ASconv 0.
\]
This in conjunction with Lemma~\ref{Lemma:distance_and_exp_fixed_p} finishes the proof of the first statement.

For the second statement, let $S = \support(\theLaw_X)$ and let $\kappa = \sup_{x \in \Xspace} \sqrt{k(x,x)}$.
Lemma~\ref{Lemma:sup_norm_via_L2_and_H_norm} with the $L^2$ bound above gives
\[
     \sup_{x \in S} \Ynorm{\tilde f_n(x) - f_\origin(x)} \leq \kappa \sqrt{\frac{2}{\epsilon \sigma_L}} \Delta_n^{1/2} + 2 C T_L,
\]
for $L \in \Lset = \{ l \geq 1 | \sigma_l > 0 \}$, and $\inf_{L \in \Lset} T_L = 0$.
Since $\Delta_n \ASconv 0$, we get
\[
     \limsup_n \sup_{x \in S} \Ynorm{\tilde f_n(x) - f_\origin(x)} \leq  2 C T_L, \quad \forall L \in \Lset,
\]
on the same probability one event $A$ for which $\Delta_n \ASconv 0$.
Therefore, 
\[
     \limsup_n \sup_{x \in S} \Ynorm{\tilde f_n(x) - f_\origin(x)} \leq  \inf_{L \in \Lset} 2 C T_L = 0,
\]
almost surely. This shows $ \sup_{x \in S} \Ynorm{\tilde f_n(x) - f_\origin(x)} \outerASconv 0$. 

Finally, $\mu(x) = \Exp_\origin(f_\origin(x))$ by \eqref{eq:f0}, so that Lemma~\ref{Lemma:distance_and_exp_fixed_p} gives $\sphereDist(\Exp_\origin(\tilde f_n(x)), \mu(x)) \leq \Ynorm{\tilde f_n(x) - f_\origin(x)}$ for every $x \in \Xspace$, and the last claim follows.
\end{proof}

\begin{proof}[proof of Lemma~\ref{Lemma:Bclosedsupnorm}]
  Since the operator-kernel is bounded and continuous, $B \subset C(\Xspace, \Yspace)$. \citep[Proposition 2]{carmeliVectorValuedReproducing2010}
  Let $\{ f_i\}_i \subset B$ be a Cauchy sequence in the topology of $C(\Xspace, \Yspace)$. Since the latter is a Banach space, it is in particular complete, and there exists  an $f_\infty \in C(\Xspace, \Yspace)$ such that 
  $$ \sup_{x \in \Xspace} \Ynorm{f_i(x)  - f_\infty(x)} \xrightarrow{i \rightarrow \infty} 0.  $$
  In the topology of $\Hspace$, $B$ is weakly sequentially compact \citep[Theorem 1.96]{barbu2012convexity}, so there exists a subsequence $\{ f_{i_j}\}_j$ that weakly converges to $f^*\in B$, that we will rename $\{ f_{j}\}_j$ . For all $y \in \Yspace$, the reproducing property and weak convergence give
  \begin{align*}
    \YinnerProd{  f_j(x) - f^*(x), y } = \HinnerProd{  f_j - f^*, K_x y }  \xrightarrow{j\rightarrow \infty} 0.
  \end{align*}
  Therefore for each fixed $x \in \Xspace$, we have weak convergence in $\Yspace$ of $f_j(x)$ to $f^*(x)$. 
  Since $\Yspace$ is finite-dimensional, weak convergence is equivalent to strong convergence and $\Ynorm{f_j(x) -f^*(x)} \xrightarrow{j \rightarrow\infty} 0$ for each $x \in \Xspace$. Therefore, $f_j$ converges pointwise to $f^* \in B$. Since it also converges uniformly to $f_\infty$, $ f_\infty = f^*  \in  B$ and $B$ is closed in $C(\Xspace, \Yspace)$.
\end{proof}

\begin{proof}[proof of Lemma~\ref{Lemma:ball_in_VVRKHS_compact_in_sup_norm}]
  Following \citet{pillonettoSolutionsNonlinearControl2008}, we show the conditions of the Arzelà-Ascoli theorem \citep{abraham2012manifolds}. We have to show the following:
  \begin{enumerate}
    \item $B$ is equicontinuous: for all $\epsilon>0$, there exists $\delta>0$ such that for all $f \in B   $ for all $x, \check{x} \in \Xspace$ with $d_X(x,\check{x}) < \delta $, $\Ynorm{f(x)-f(\check{x})}< \epsilon$.
    \item $B $ is pointwise bounded, for all $x \in \Xspace$, $\{f(x), f \in B \}$ is bounded in $\Yspace$.
    \item $B $ is closed in the supremum norm.
  \end{enumerate}
  We show these in order.
  \begin{enumerate}
    \item Let $\epsilon > 0$.
      Since $K$ is continuous on a compact set, it is uniformly continuous.
      Choose $\eta > 0$ such that for all $x_1, x_2 \in \Xspace$, 
      $d_X(x_1, x_2) < \eta$ implies 
      $\operatornorm{K(x_1,x_1)-K(x_1, x_2)} \leq \frac{\epsilon^2}{2r^2}$ and 
      $\operatornorm{K(x_2,x_2)-K(x_1, x_2)} \leq \frac{\epsilon^2}{2r^2}$.
      Let $y \in \Yspace$, $\Ynorm{y}=1$, and consider $x, \check{x} \in \Xspace$ with 
      $d_X(x, \check x) < \eta$.
      Using the reproducing property, for any such $y$ we have
      \begin{align*}
        \YinnerProd{ f(x)-f(\check{x}), y } 
    &= \HinnerProd{ f, K_{x}y-K_{\check{x}}y } \\
    &\leq \Hnorm{f}\,\Hnorm{K_{x}y-K_{\check{x}}y}.
      \end{align*}
      Moreover,
      \begin{align*}
        \Hnorm{K_{x}y-K_{\check{x}}y}^2
    &= \YinnerProd{(K(x,x)-K(x,\check x))y,y}
    + \YinnerProd{(K(\check x,\check x)-K(x,\check x))y,y} \\
    &\le \operatornorm{K(x,x)-K(x,\check x)}
    + \operatornorm{K(\check x,\check x)-K(x,\check x)},
      \end{align*}
      where the last inequality uses $\Ynorm{y}=1$ and the definition of the operator norm.
      Hence
      \begin{align*}
        \Ynorm{ f(x) - f(\check x)} 
    &= \sup_{\Ynorm{y}=1} \YinnerProd{ f(x)-f(\check{x}), y } \\
    &\leq \Hnorm{f}\,
    \bigl(\operatornorm{K(x,x)-K(x, \check x)} 
    + \operatornorm{K(\check x,\check x)-K(x, \check x)}\bigr)^{1/2} \\
    &\leq r\left(\frac{\epsilon^2}{2r^2} + \frac{\epsilon^2}{2r^2}\right)^{1/2} 
    = \epsilon.
      \end{align*}

    \item From Lemma~\ref{Lemma:bounding_fX}, since $x \mapsto K(x,x)$ is continuous on a compact set,  
      \[
        \Ynorm{f(x)}\leq \sup_{x \in \Xspace} \operatornorm{K(x,x)}^{1/2} \Hnorm{f} \leq C_k r < \infty,
      \]
      where $C_k = \sup_{x \in \Xspace} \operatornorm{K(x,x)}^{1/2} < \infty$.
    \item This follows from Lemma \ref{Lemma:Bclosedsupnorm}.
  \end{enumerate}
\end{proof}

\clearpage 

\section{Results under numerical proofs}
\label{sec:results_numerical_proofs}

\subsection{Assumptions}

\begin{Definition} \label{defn:Vset_numerical}
    
    For $\rho > 0$, we define the \emph{deterministic} set $\Vset(\rho) \subset \Hspace$ by
    \[
        \Vset(\rho)
        :=
        \left\{
            f \in \Hspace \,\middle|\,
            \begin{array}{l}
                \Ynorm{f(X)} < \pi,
                \quad \mathrm{a.s.} \\
              \Ynorm{ f(X) - \Log_\origin Y } \leq \rho
                \quad \mathrm{a.s.}, \\[0.3em]
            \end{array}
        \right\}.
    \]
\end{Definition}
\begin{Assumption}
    \label{assumption:support_of_Y_for_theory_numerical}
Assume that, for one of the $(r, \rho)$ pairs in Table~\ref{table:r_rho_pairs}, we have
     \[
       \sphereDist(\origin, Y) \leq r \quad \text{ almost surely},
     \]
     and $f_\origin \in \Vset(\rho)$.
      \begin{table}[h!]
        \centering
        \begin{tabular}{ll}
          $r$ & $\rho$ \\ 
          \hline  \\[0.1ex]
          $\frac{\pi }{2}$ & $\frac{\pi}{2.1}$ \\ \\
          $\frac{2\pi}{3}$ & $\frac{\pi}{6.6}$ \\
        \end{tabular}

        \caption{Data radius $r$ and error magnitude $\rho$ considered in Assumption~\ref{assumption:support_of_Y_for_theory_numerical}.} \label{table:r_rho_pairs}
      \end{table}
\end{Assumption}

\subsection{Convexity results}

\begin{Lemma} \label{lemma:ell_Hessian_numerical}
Let $\Dregion_\rho(y) = \ball_{T_\origin \sphere}(\Log_\origin y,  \rho)$.
Assume $(r, \rho)$ is a pair given in Table~\ref{table:r_rho_pairs}, and $\sphereDist(\origin, y) \leq r$. Then,
    \begin{enumerate}
        \item The minimal eigenvalue of the Hessian of $\ell_y$ is at least $\epsilon > 0$ on $\Dregion_\rho(y)$, for some $\epsilon > 0$.
        \item $\nabla \ell_y(v)$ is well defined on $\Dregion_\rho(y)$ and
            \[
                \ell_y(v') \geq \ell_y(v) + \YinnerProd{v'-v, \nabla \ell_y(v)} + \frac{\epsilon}{2} \Ynorm{v'-v}^2 
            \]
            In particular, $\ell_y$ is convex on $\Dregion_\rho(y)$.
    \end{enumerate}
    \begin{proof}[Numerical proof]
      Our proof relies on a numerical argument, requiring the computation of the minimal eigenvalue of the Hessian of $\ell_y$ and its evaluation on a range of values of $ \zeta = \sphereDist(y, \origin)$, $\innerProd{v,y}$ and $\norm{v}$, using Proposition~\ref{proposition:convexity_plot}. 
      For the first statement, Proposition~\ref{proposition:convexity_plot} tells us that it is enough to show that ${g_\zeta}{|\ball_{\Real^2}( (\zeta, 0), \rho)} \geq \epsilon > 0$ for $\zeta := \sphereDist(\origin, y) \in [0,r]$, where $g_\zeta$ is defined in Corollary~\ref{corollary:ell_Hessian_decomposition}. We do this by numerical evaluations on a grid of values for $\zeta$ and the arguments of $g_\zeta$. The results of the numerical evaluations are depicted in Figures~\ref{fig:pi2-convexity-overall} and~\ref{fig:pi2-convexity-zoom} for $r=\pi/2$ and Figures~\ref{fig:2pi3-convexity-overall} and~\ref{fig:2pi3-convexity-zoom} for $r=2\pi/3$. The code for reproducing the numerical calculations is available in the \thesupplement.

      The second statement follows from Lemma~\ref{Lemma:gradient_l_v}, the first statement, and Lemma~\ref{Lemma:convex_gradient}.

    \end{proof}
\end{Lemma}

\begin{figure}
  \includegraphics[width=\textwidth,height=0.6\textheight,keepaspectratio]{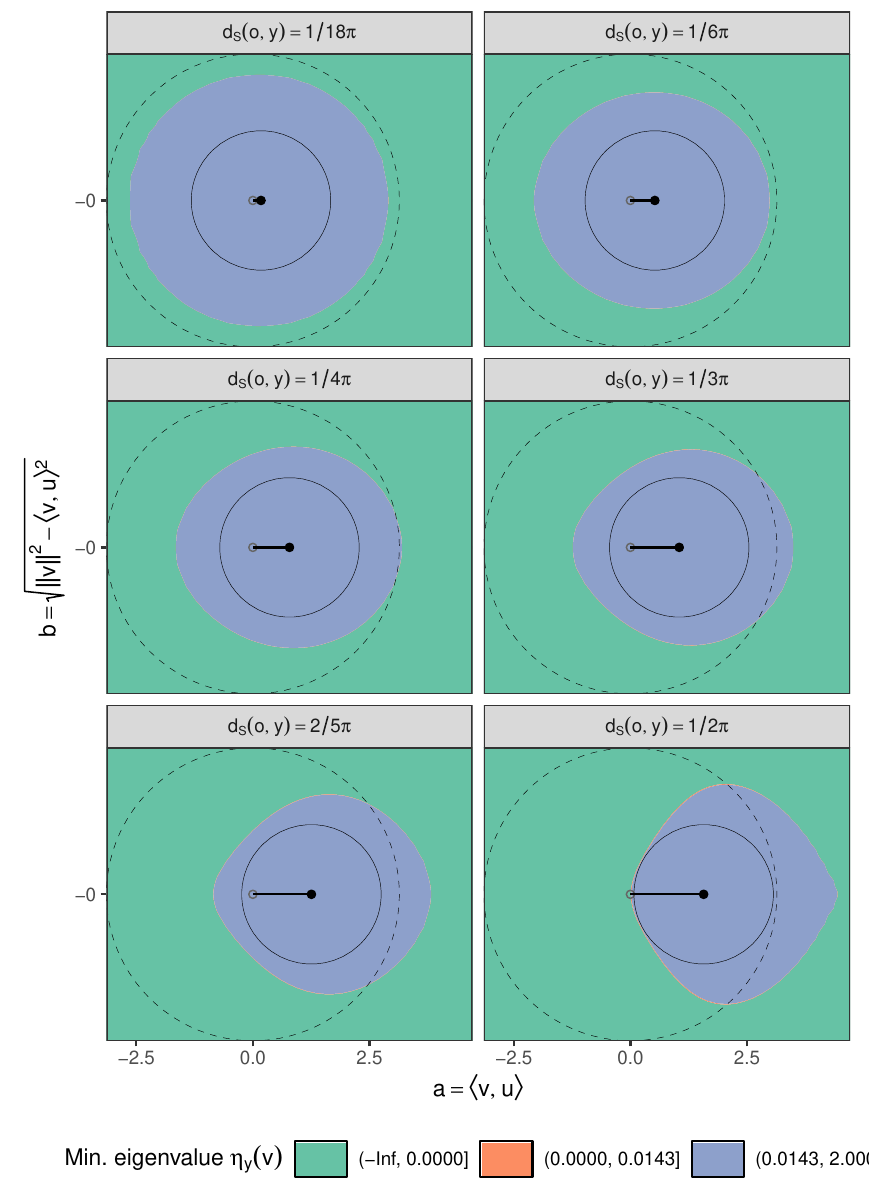}
  \caption{
    Plane $(a,b)$ of the domain of $g_\zeta$, for different values of $\zeta=\sphereDist(\origin, y)$ (one value per panel).
    The colours depict different ranges of values of $g_\zeta(a,b)$, as described in the plot's legend. 
    In each panel, the circle denotes the point $(0,0)$, and the filled circle the point $(\zeta, 0)$. The dashed circle is the circle of radius $\pi$ around $(0,0)$, and corresponds to the circle of radius $\pi$ on $T_\origin \sphere$. The solid circle is the circle of radius $\pi/2.1$ around $(\zeta, 0)$. 
    This plot in conjunction with Proposition~\ref{proposition:convexity_plot} shows numerically that if $\zeta = \sphereDist(\origin,y) \leq \pi/2$ then $\ell_y(v)$ is convex over the region $v \in \ball_{T_\origin}( \Log_\origin(y), \pi/2.1)$.
    A zoom of the bottom right panel is given in Figure~\ref{fig:pi2-convexity-zoom}.
  }
  \label{fig:pi2-convexity-overall}
\end{figure}

\begin{figure}
  \includegraphics[width=\textwidth,height=0.6\textheight,keepaspectratio]{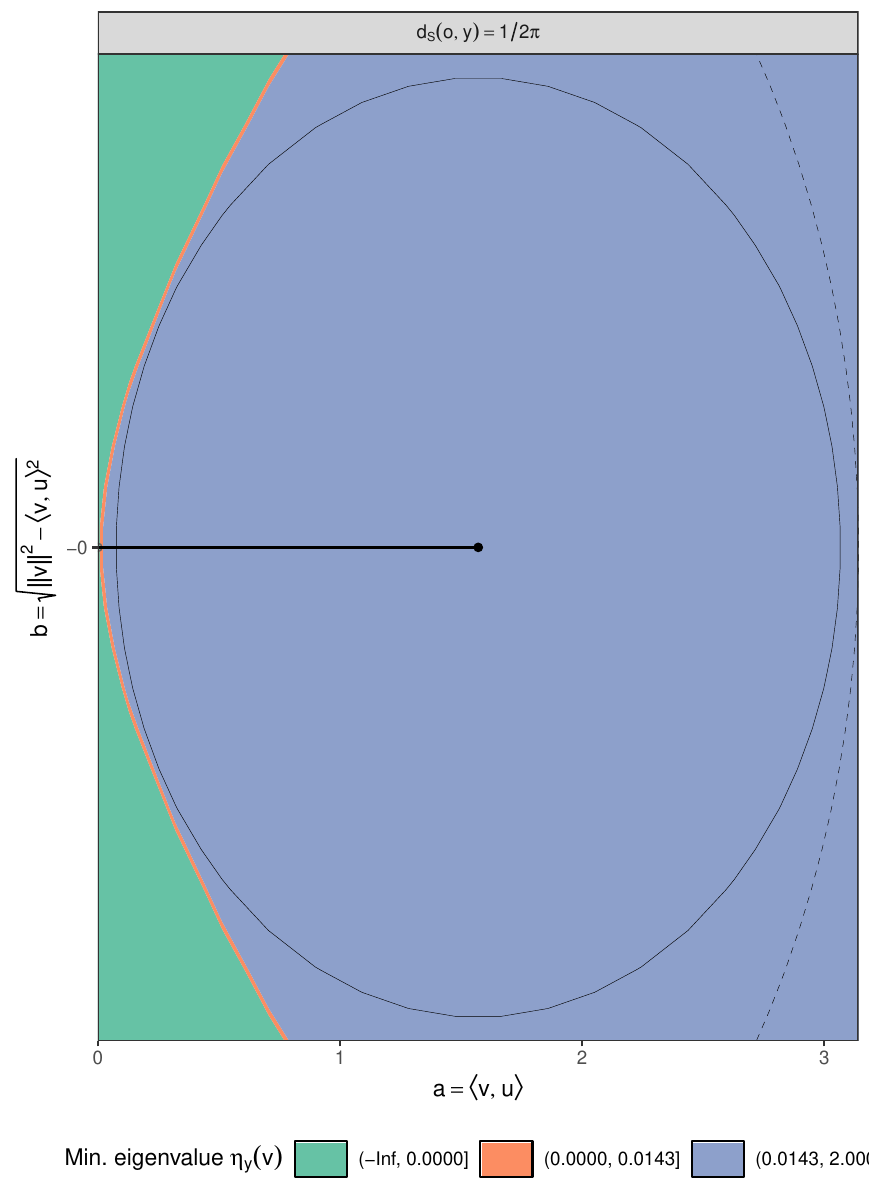}
  \caption{Zoom of the bottom right panel of Figure~\ref{fig:pi2-convexity-overall}.}
  \label{fig:pi2-convexity-zoom}
\end{figure}

\begin{figure}
  \includegraphics[width=\textwidth,height=0.6\textheight,keepaspectratio]{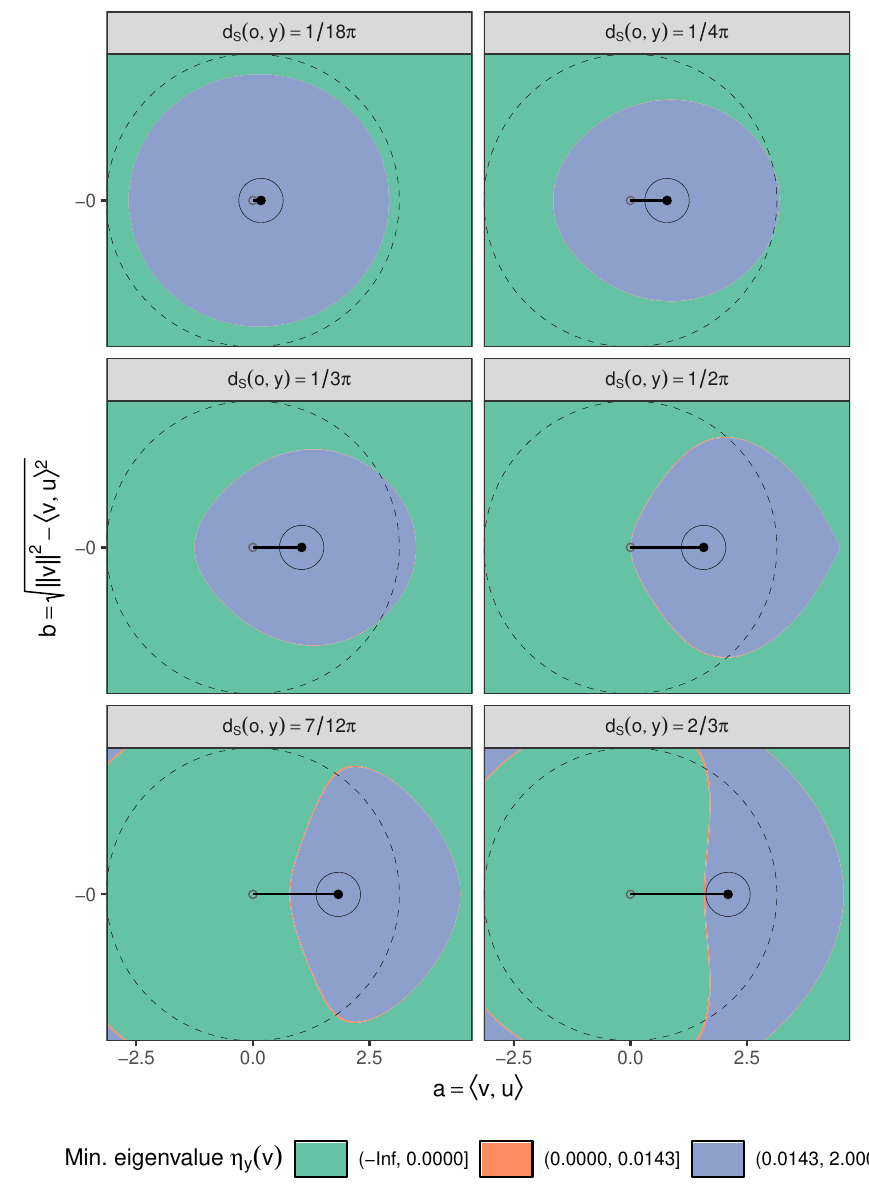}
  \caption{
    Plane $(a,b)$ of the domain of $g_\zeta$, for different values of $\zeta=\sphereDist(\origin, y)$ (one value per panel).
    The colours depict different ranges of values of $g_\zeta(a,b)$, as described in the plot's legend. 
    In each panel, the circle denotes the point $(0,0)$, and the filled circle the point $(\zeta, 0)$. The dashed circle is the circle of radius $\pi$ around $(0,0)$, and corresponds to the circle of radius $\pi$ on $T_\origin \sphere$. The solid circle is the circle of radius $\pi/6.6$ around $(\zeta, 0)$. 
    This plot in conjunction with Proposition~\ref{proposition:convexity_plot} shows numerically that if $\zeta = \sphereDist(\origin,y) \leq 2\pi/3$ then $\ell_y(v)$ is convex over the region $v \in \ball_{T_\origin}( \Log_\origin(y), \pi/6.6)$.
    A zoom of the bottom right panel is given in Figure~\ref{fig:2pi3-convexity-zoom}.
  }
  \label{fig:2pi3-convexity-overall}
\end{figure}

\begin{figure}
  \includegraphics[width=\textwidth,height=0.6\textheight,keepaspectratio]{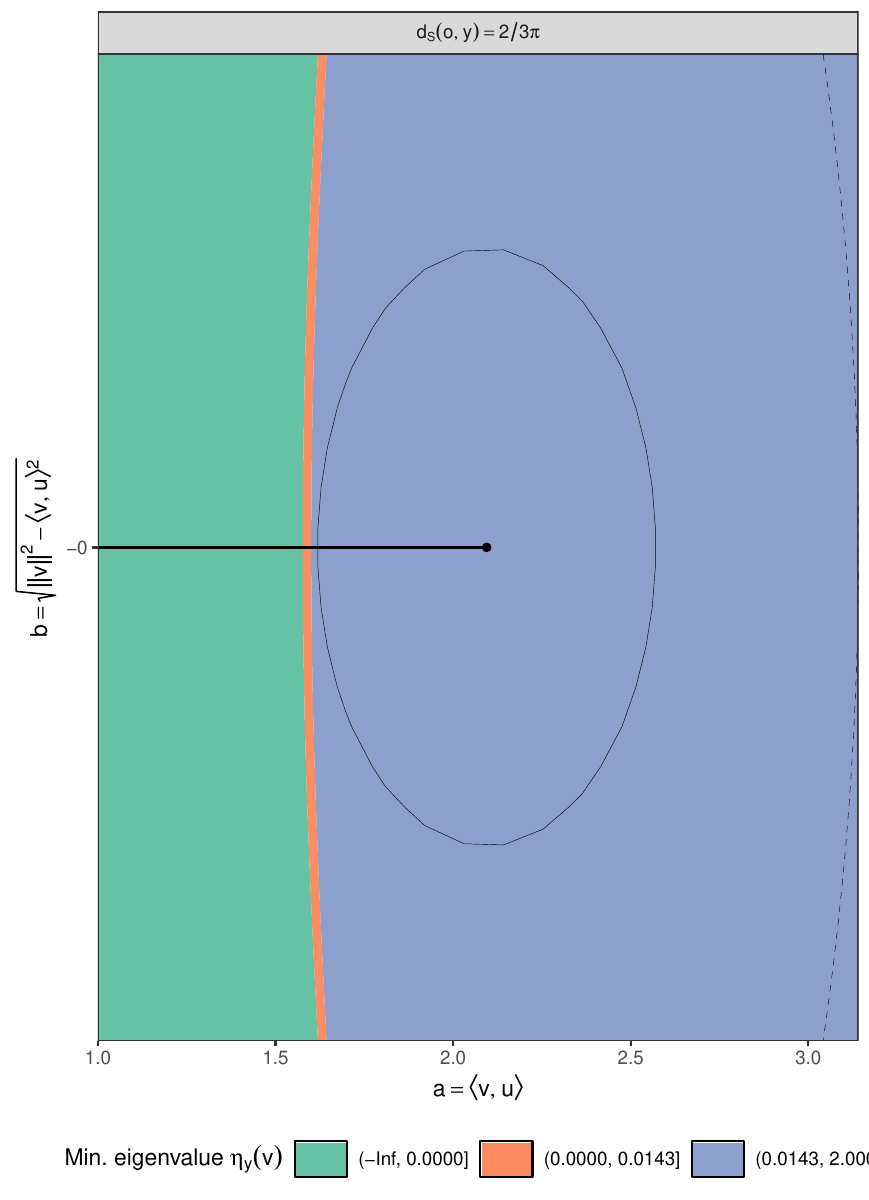}
  \caption{Zoom of the bottom right panel of Figure~\ref{fig:2pi3-convexity-overall}.}
  \label{fig:2pi3-convexity-zoom}
\end{figure}

\begin{Lemma}
    \label{Lemma:population_risk_convexity_inequality_numerical}
    Recall the definition of the population risk $\risk$ from \eqref{eq:population_risk}.
    Under assumption~\ref{assumption:support_of_Y_for_theory_numerical}, provided $\Expected k(X,X) < \infty$,
 $\nabla \risk(f)$  exists for all $f \in \Vset(\rho)$ and
    \[
        \nabla \risk(f) =   \Expected\eval_X^\adjoint \nabla \ell_{Y}( f(X))
    \]
    where $\ell_y$ is defined in \eqref{eq:ell_function}.
    In addition, 
    \[
        \risk(f) - \risk(g) \geq \HinnerProd{ \nabla \risk(g) ,  f-g }  + \frac \epsilon 2 \Ltwonorm{f - g}^{2} 
    \]
    for all $f,g \in \Vset(\rho)$.
    \begin{proof}
      The proof mimics the proof of Lemma~\ref{Lemma:population_risk_convexity_inequality_new}.
    \end{proof}
\end{Lemma}

\begin{Theorem} \label{theorem:risk_is_convex_numerical}
    Under Assumption~\ref{assumption:support_of_Y_for_theory_numerical}, 
    \begin{enumerate}
        \item For all $f \in \Vset(\rho)$, $\risk(f) \geq \risk(f_\origin) + \frac{\epsilon}{2} \Ltwonorm{f-f_\origin}^2$. 
        \item For $\lambda_n > 0$, $\risk_n(\cdot, \lambda_n)$ is strictly convex on $\Vset(\rho)$, and provided $\sup_x k(x,x) < \infty$, it admits a unique minimiser $\hat f_n$ over $\Vset(\rho)$.  
    \end{enumerate}
    
\begin{proof}
  The proof mimics the proof of Theorem~\ref{theorem:risk_is_convex_new}.
\end{proof}
\end{Theorem}

\subsubsection{Rates of convergence under minimal assumptions}

\begin{Theorem} \label{theorem:rates_for_f_weakest_assumptions_numerical}
    Assume  $\sup_{x \in \Xspace} \sqrt{k(x,x)} = c < \infty$, and let $\hat \mu_n = \Exp_\origin \circ \hat f_n$.
    Under Assumption~\ref{assumption:support_of_Y_for_theory_new}, for $\lambda_n^2 = 8\pi^2 c^2 n^{-1/2}$,
    each of the following bounds holds with probability $1-\delta$, for any $\delta \in (0, 1)$:
\[
\Ltwonorm{ \hat f_n - f_\origin } \leq \tilde c \delta^{-1/2} n^{-1/4},  \qquad  \int_\Xspace \sphereDist^2(\mu(x), \hat \mu_n(x)) \ud\theLaw_X(x) \leq \tilde c^2 \delta^{-1} n^{-1/2},
\]
    where $\tilde c$ is a constant depending only on $\Hnorm{f_\origin}, c$ and $\epsilon$ defined in Theorem~\ref{theorem:risk_is_convex_new}.
    In particular,
    \[
        \Ltwonorm{\hat f_n - f_\origin} = O_{\theLaw^*}(n^{-1/4}) \quad \text{and} \quad \int_\Xspace \sphereDist^2(\mu(x), \hat \mu_n(x)) \ud\theLaw_X(x) = O_{\theLaw^*}(n^{-1/2}).
    \]
\end{Theorem}
\begin{proof}
  The proof follows exactly the same steps as the proof of Theorem~\ref{theorem:rates_for_f_weakest_assumptions_corrected}.
\end{proof}

The rates obtained translate directly into rates for the conditional mean estimator $\hat \mu_n(x)= \Exp_\origin( \hat f_n(x) )$. Recall the definition of the conditional mean $\mu$ from \eqref{eq:muX}.
\begin{Theorem} \label{theorem:L2_convergence_mu_weakest_assumptions_numerical}
  Assume that Assumption~\ref{assumption:support_of_Y_for_theory_numerical} holds,  and $\sup_x k(x,x) =c < \infty$. 
  For $\lambda_n \asymp n^{-1/4}$,
  \[
    \int_\Xspace \sphereDist^2(\mu(x), \hat \mu_n(x)) \ud\theLaw_X(x) = O_{\theLaw^*}(n^{-1/2})
  \]
\end{Theorem}
\begin{proof}
  The proof follows exactly the same steps as the proof of Theorem~\ref{theorem:rates_for_f_weakest_assumptions_corrected}.
\end{proof}

\subsubsection{Rates of convergence depending on VVRKHS smoothness}

\begin{Theorem} \label{theorem:rates_for_f_depending_on_smoothness_numerical}
  Assume that Assumption~\ref{assumption:support_of_Y_for_theory_numerical} holds, $\dim(\Yspace) < \infty$, and $\sup_x k(x,x) < \infty$.  We have the following results, depending on the rates of decay of the eigenvalues of $\KIntegralOperator$ given in Theorem~\ref{theorem:rates_for_f_depending_on_smoothness_new}:
            \[
               \Ltwonorm{\hat f_n - f_\origin}^2 = \begin{cases}
                     O_{\theLaw^*}(n^{-1}) & \text{for finite-rank kernel and } \lambda_n \asymp e^{-n}, \\
                 O_{\theLaw^*}(n^{-\frac{p}{p+1}}) & \text{for polynomial decay and } \lambda_n \asymp n^{-\frac{p}{2(p+1)}},\\
                 O_{\theLaw^*}( (\log(n))^{q}/n) & \text{for stretched exponential decay and } \lambda_n \asymp n^{-1/2}.
                \end{cases}
            \]
                The same rates hold for $\int_\Xspace \sphereDist^2(\mu(x), \hat \mu_n(x)) \ud\theLaw_X(x)$.
\end{Theorem}
\begin{proof}
  The proof follows exactly the same steps as the proof of Theorem~\ref{theorem:rates_for_f_depending_on_smoothness_new}.
\end{proof}

\subsection{Strong Consistency}

Let $\widetilde \Vset(\rho) = \{ f \in \Vset(\rho) \mid \Hnorm{f} \leq C \}$ for some arbitrary $C > \Hnorm{f_\origin}$.
Let $\tilde f_n$ be the minimiser of $\risk_n(\cdot, \lambda_n)$ over $\widetilde \Vset(\rho)$, and let $\tilde \mu_n = \Exp_\origin \circ \tilde f_n$.
\begin{Theorem}\label{theorem:strong_consistency_numerical} 
  Assume that Assumptions~\ref{assumption:support_of_Y_for_theory_numerical} and~\ref{assumption:for_almost_sure_consistency} hold.
  If $\lambda_n \downarrow 0$ as $n \to \infty$,
  \begin{enumerate}
        \item $\Ltwonorm{\tilde f_n - f_\origin} \outerASconv 0$ and  $\int_\Xspace \sphereDist^2(\tilde \mu_n(x), \mu(x)) d\theLaw_X(x) \outerASconv 0$  as $n \to \infty$;
        \item If $\KIntegralOperator \neq 0$, $\sup_{x \in \support(\theLaw_X)} \Ynorm{\tilde f_n(x) - f_\origin(x)} \outerASconv 0$ and 
        
        $\sup_{x \in \support(\theLaw_X)} \sphereDist(\tilde \mu_n(x), \mu(x)) \outerASconv 0$, as $n \to \infty$.
  \end{enumerate}
\end{Theorem}
\begin{proof}
The proof follows exactly the same steps as the proof of Theorem~\ref{theorem:strong_consistency_new}.
\end{proof}

\else
\fi

\end{document}